\documentclass[twoside,11pt]{article}
\usepackage[top=1in,bottom=1in,left=1in,right=1in]{geometry}
\usepackage{nameref}
\usepackage{authblk}
\usepackage{graphicx}
\usepackage{subcaption}
\usepackage{caption}
\usepackage[T1]{fontenc}    
\usepackage[utf8]{inputenc} 
\usepackage{color, colortbl}
\usepackage{amsfonts, amsmath, amssymb, amsthm, dsfont}
\usepackage{nicefrac}
\usepackage{algorithm}
\usepackage{algorithmic}
\usepackage{url}
\usepackage[numbers]{natbib}

\newtheorem{theorem}{Theorem}
\newtheorem{lemma}{Lemma}
\newtheorem{corollary}{Corollary}
\newtheorem{definition}{Definition}

\newcommand{\MMD}{\text{MMD}}

\DeclareMathOperator*{\argmax}{argmax}

\begin{document}

\title{Supervised Contamination Detection, with Flow Cytometry Application}

\author[1]{Solenne Gaucher \thanks{solenne.gaucher@ensae.fr}}
\author[2]{Gilles Blanchard}
\author[2]{Frédéric Chazal}
\affil[1]{CREST, ENSAE, IP Paris, INRIA}
\affil[2]{Université Paris-Saclay, CNRS, INRIA}

\maketitle

\begin{abstract} The contamination detection problem aims to determine whether a set of observations has been contaminated, i.e. whether it contains points drawn from a distribution different from the reference distribution. Here, we consider a supervised problem, where labeled samples drawn from both the reference distribution and the contamination distribution are available at training time. This problem is motivated by the detection of rare cells in flow cytometry. Compared to novelty detection problems or two-sample testing, where only samples from the reference distribution are available, the challenge lies in efficiently leveraging the observations from the contamination detection to design more powerful tests. 

In this article, we introduce a test for the supervised contamination detection problem. We provide non-asymptotic guarantees on its Type I error, and characterize its detection rate. The test relies on estimating reference and contamination densities using histograms, and its power depends strongly on the choice of the corresponding partition. We present an algorithm for judiciously choosing the partition that results in a powerful test. Simulations illustrate the good empirical performances of our partition selection algorithm and the efficiency of our test. Finally, we showcase our method and apply it to a real flow cytometry dataset.
\end{abstract}


\section{Introduction}

\subsection{Motivation}

Anomaly detection is an important research topic in data analysis, aiming to identify rare elements in a dataset that deviate significantly from the majority of observations. Its applications are widespread, ranging from cybersecurity to social media content moderation and fraud detection, making it a thoroughly investigated field. The primary goal of anomaly detection is typically to identify unusual elements within a dataset. However, in specific scenarios, the objective shifts to testing for the presence of anomalies in the entire dataset, a problem called contamination detection. In this paper, we focus on the problem of \textit{supervised contamination detection}, which seeks to assess whether a sample has been contaminated by data points stemming from a distinct distribution, based on samples from the non-contaminated distribution and the contamination distribution. Our study is motivated, among other problems, by the challenge of detecting rare populations in \textit{flow cytometry.}

\paragraph{Flow cytometry} Flow cytometry is a widely used technology employed to identify and quantify specific cell populations within a sample of cells. This method enables the simultaneous quantification of multiple surface and intracellular markers at the individual cell level. To achieve this, cells are labeled with antibodies coupled to fluorescent molecules, which bind specifically to particular receptors or intracellular markers. These cells then pass one-by-one in front of an optical device that measures the emission spectrum of the markers attached to the cell, thereby revealing the presence of the target receptor. Modern flow cytometers can detect up to 15 or 20 parameters (i.e., markers) routinely, with throughput rates exceeding 10,000 cells per second. Flow cytometry datasets commonly contain several million cells \citep{cytoComp}. 
    
Data analysis for flow cytometry has traditionally been conducted through ``manual gating", involving successive visual inspections of two-dimensional scatterplots followed by selections of regions in this two-dimensional space containing cells of interest. Typically, markers are selected based on their ability to distinguish between different types of cells, allowing for meaningful cell clustering when separating cells according to their values. However, this time-consuming method has significant limitations, including its lack of scalability and the subjectivity in selecting dimensions and regions of interest, leading to a lack of reproducibility. Consequently, substantial efforts have been devoted to the development of automated analysis methods. These approaches include unsupervised clustering algorithms as well as supervised techniques trained on manually annotated samples. The primary goal of the latter type of algorithms is not only to identify the cell clusters but also to automatically annotate these clusters.

One problem faced in flow cytometry data analysis is the significant heterogeneity in the number of cells within each population, and the importance of small sub-populations in biological contexts. Notably, a major challenge in clinical applications is the detection of rare pathogenic objects in patient blood. These objects may include circulating tumor cells, which are very rare during the early stages of cancer development \citep{ctc}, or various microorganisms and parasites present in the blood during acute blood infections \citep{ijms21072323}. These examples emphasize the importance of detecting small cell populations. Unlike many anomaly detection scenarios where information about contaminations is unavailable, in this specific problem, manually labeled samples containing instances of these small populations are accessible and can be leveraged for training our algorithm. In this article, we therefore explore the problem of contamination detection using observations from both uncontaminated and contaminating samples, motivated by the challenges related to small cell populations identification in flow cytometry data.

\subsection{Related works}
Contamination detection is closely related to anomaly detection, which focuses on identifying instances that originate from a different distribution or class than the majority. The settings for anomaly detection vary based on the observations available to the statistician. In completely unsupervised anomaly detection, the statistician only has access to unlabeled data. Anomalies, also referred to as outliers, are commonly defined as ``observations which deviate so significantly from other observations as to raise suspicions that they were generated by a different mechanism" \citep{Enderlein1987HawkinsDM}. In semi-supervised anomaly detection, also referred to as novelty detection, it is assumed that the statistician has access to a ``pure" sample of observations drawn from the non-contaminated density $f^0$, and to a test sample that is a mixture of $f^0$ and of an arbitrary distribution $f^1$  \citep{6287347, JMLR:v7:vert06a, BlanchardNovelty, VILLAPEREZ2021106878}. Finally, in supervised anomaly detection, samples from both $f^0$ and $f^1$ are available at training time. Consequently, the problem reduces to a two-class classification scenario, which can be addressed using classical classification algorithms  \citep{Aggarwal2017}. We emphasize that all these problems ultimately aim to identify anomalous observations, i.e., observations drawn from a distribution different from $f^0$, and classify them as such. In contrast, in the present paper, we focus on testing the presence of corruptions or anomalies in the test sample, i.e. on testing whether the distribution $f^1$ has a positive proportion in the mixture model defining the density of the test sample.

Contamination detection from labeled samples is closely related to the problems of label shift estimation and of quantification. Label shift estimation  \citep{10.5555/3495724.3496001} addresses the task of estimating the target label distribution in a classification task under label shift. On the other hand, quantification  \citep{10.1145/3117807} focuses on estimating the class distribution for unlabeled test sets using models trained on a dataset with a different class distribution. In both problems, the test distribution can be viewed as a mixture of distributions corresponding to different labels, which can be learned from a training sample. Methods to solve these problems can typically be categorized into two types. The first type includes variants of the classify-and-count approach  \citep{count_class, DBLP:conf/icml/LiptonWS18, 10.5555/3524938.3524960}, which relies on fitting a classifier to the training samples and approximating the distributions of the classes in the test sample based on the distribution of predictions. Another approach considers the test distribution as a mixture model of distributions that can be learned from the training sample  \citep{pmlr-v32-iyer14, GONZALEZCASTRO2013146, Dussap}. The task of label shift quantification then involves estimating the weights of the components in the source dataset. Our approach is similar to the latter approach. It is important to note that these problems are typically concerned with obtaining guarantees on the aggregated error for reconstructing the weights of the mixture model, whereas our objective is to test if the weight of a given component is zero.

The contamination model examined in this paper is closely associated with the admixture model, which finds applications in false discovery problems, genetics, and astronomy, and has been extensively investigated (see, for example, \cite{Bordes10AOS, CAI10AOS, Celisse10JSIP,Nguyen14}). In this model, the distribution of the non-contaminated density is assumed to be known, while the weight of the mixture and the distribution of the contaminations are unknown. Previous works have mainly focused on estimating these weights and on testing the equality of contamination distributions across different samples \citep{Milhaud24}.

Finally, it is worth noting that without observations drawn from the contamination distribution $f^1$, the problem of contamination detection reduces to that of two-sample testing. As highlighted in \cite{BlanchardNovelty}, the optimal approach in this scenario is to test if the testing and training samples have the same distribution. In contrast, this paper takes advantage of additional information to develop an adaptive test with high power against contaminations drawn from a distribution from which we observe samples.

\subsection{Preliminary considerations}

\paragraph{Supervised contamination detection problem} We consider the following contamination detection problem : given a non-contaminated distribution $f^0$ and a contamination distribution $f^1$, we want to determine whether a test sample $\mathbb{X} = (X_1, \ldots, X_n)$ of size $n$ is drawn from $f^0$, or if contains a positive fraction of points drawn from $f^1$. More precisely, denoting by $f$ the distribution of the test sample, we want to test the hypothesis
\begin{center}
    $H_0 \ :\ f = f_0 $.
\end{center}
 against the alternative
\begin{center}
    $H_1 \ :\ f = (1-\theta)\times f_0 + \theta \times f_1$ for some $\theta \in (0, 1]$.
\end{center}
In the supervised setting, we have access to a pure sample denoted $\mathbb{X}^0 =(X_1^0, \ldots, X_{n^0}^0)$ of $n^0$ points drawn independently from $f^0$. Additionally, we observe a sample of $n^1$ labeled anomalies drawn independently and identically from $f^1$, denoted as $\mathbb{X}^1 =(X_1^1, \ldots, X_{n^1}^1)$. In the following, we assume that the densities $f^0$ and $f^1$ have support bounded in $[0,1]^d$. \footnote{In practice, this assumption holds true in flow cytometry analysis, where the data is often pre-processed in a standardized fashion.}

\paragraph{Detection rates in case of known distributions}
To gain insight into the complexity of the detection problem, let us first consider the simpler scenario where densities $f^0$ and $f^1$ are known. Arguably, the most straightforward case for testing $H_0$ against $H_1$ arises when both $f^0$ and $f^1$ are positive on $[0,1]^d$. In this situation, classical results from asymptotic statistics can be applied to develop a test asymptotically of level $\alpha$, and to establish its detection rate, as demonstrated in the following lemma.
\begin{lemma}\label{lem:known_f}
Assuming that $f^0$ and $f^1$ are positive on $[0,1]^d$, define
    \begin{align*}
       \sigma^2 = \int\left(\frac{f^1(x)}{f^0(x)}-1\right)^2f^0(x)dx \quad \quad \text{ and }\quad \quad S_n(\mathbb{X}) = \frac{1}{n}\sum_{i\leq n}\left(\frac{f^1(X_i)}{f^0(X_i)} - 1\right),
    \end{align*}
and let $\Phi^{-1}$ be the quantile function of the normal distribution. Then, for all $\alpha \in (0,1)$, the test rejecting $H_0$ if
\begin{align*}
    S_n(\mathbb{X}) \geq \frac{\sqrt{\sigma^2} \Phi(1-\alpha)^{-1}}{\sqrt{n}}
\end{align*} is asymptotically of level $\alpha$. Moreover, it has asymptotically power $\beta > \alpha$ against alternatives $\theta = \frac{\Phi^{-1}(1-\alpha)- \Phi^{-1}(1-\beta)}{\sqrt{n\sigma^2}}$.
\end{lemma}
Under the assumptions of Lemma \ref{lem:known_f}, the rates for testing $H_0$ against $H_1$ are parametric: the test has constant power against alternative containing a proportion of order $1/\sqrt{n\sigma^2}$ of points from $f^1$. This result highlights the role played by the signal term $\sigma^2$, given by the $L^2(f^0)$ norm between the density ratio $f^1/f^0$ and the constant function equal to $1$.  In comparison, two-sample tests typically have constant power against alternatives $f$ that are located at a given distance from the null hypothesis $f^0$, where the specific definition of this distance depends on the choice of the test. Consequently, we expect that additional information on $f^1$ will lead to improved detection rates if the ratio $f^1/f^0$ takes large values.

If the densities $f^0$ and $f^1$ have disjoint support, the signal $\sigma^2$ defined above is not defined, and the results of Lemma \ref{lem:known_f} do not hold. However, the problem becomes easier, since observing points in a region belonging to the support of $f^1$ but not of $f^0$ provides enough evidence to reject the null hypothesis for any confidence level. More generally, if there exists a sequence of regions $R_n$ with sufficiently low probability under $f^0$, and higher probability under $f^1$, the following lemma shows that a well-tailored test can achieve fast detection rates.
\begin{lemma}\label{lem:known_f_somewhat_fast_rates}
Assume that there exists a sequence of measurable sets $R_n\subset [0,1]^d$ such that, for some constants $A, B, \gamma >0$, we have 
\begin{align*}
    \int_{R_n} f^0(x)dx \leq \frac{A}{n} \quad \quad \text{ and } \quad \quad \int_{R_n} f^1(x)dx \geq \frac{B}{n^{\gamma}}.
\end{align*}
Then, for all $\alpha \in (c_{A, n},1)$, the test rejecting $H_0$ if $\sum_{i\leq n}\mathds{1}\{X_i \in R_n\} \geq 1$ is of level $\alpha$, where $c_{A, n}$ is a constant depending on $n$ and $A$.  Moreover, this test has power $\beta$ against alternatives $\theta =  \frac{c_{B,\beta}}{n^{1-\gamma}}$, where $c_{B,n}$ is a constant depending on $B$ and $\beta$.
\end{lemma}
Under the assumptions of Lemma \ref{lem:known_f_somewhat_fast_rates}, the test that rejects the null hypothesis if $R_n$ contains at least one observation can detect alternatives with proportions of points drawn from $f^1$ as small as $n^{\gamma-1}$. When $\gamma\leq 1/2$, meaning that the probability of the region $R_n$ under $f^1$ is large enough,  this detection rate is faster than the parametric rate.  One extreme scenario arises when the support of $f^1$ is not contained within that of $f^0$, and there exists a fixed region $R$ with constant, positive probability under $f^1$ and null probability under $f^0$. In this case, the test proposed in Lemma \ref{lem:known_f_somewhat_fast_rates} detects the alternative $H_1$ for $\theta$ of order $1/n$, or equivalently with a finite number of points drawn from $f^1$. As $\gamma$ increases, we observe a smooth transition to the parametric regime: in this situation, the probability of the regions $R_n$ decreases under $f^1$, leading to an increase in the detection rate up to the parametric rate of $1/\sqrt{n}$. Finally, when $\gamma> 1/2$, the test proposed in Lemma \ref{lem:known_f_somewhat_fast_rates} exhibits a detection rate slower than the parametric rate, suggesting that such regions have a negligible contribution to the signal of the contamination detection problem.

\subsection{Outline and contribution}
Although limited to the simplified scenario where the densities are known, the preliminary results established above shed light on the factors driving the complexity of the contamination detection problem. In particular, they show that if the densities are positive, then the detection rate of the alternative is of order $1/\sqrt{\sigma^2 n}$. On the other hand, if there are regions of sufficiently low weight under $f^0$ and sufficiently high weight under $f^1$, faster detection rates can be achieved. When compared to the two-sample testing problem, these results highlight the advantage gained from the knowledge of both densities, enabling the development of effective tests tailored to the specific problem. When these densities are unknown, the following question arises: can we design a test that achieves fast detection rate under the assumptions of Lemma \ref{lem:known_f_somewhat_fast_rates}, and competitive parametric rates otherwise, all without knowing the densities $f^0$ and $f^1$?

To address this question, we develop a non-asymptotic, non-parametric test for the problem of supervised contamination detection. More precisely, we construct an estimate $\widehat{S}(\mathbb{X})$ of the test statistic $S_n(\mathbb{X})$ used in Lemma \ref{lem:known_f_somewhat_fast_rates}. To do so, we estimate the densities $f^0$ and $f^1$ using thresholded histograms introduced in Section \ref{sec:histo}. We derive non-asymptotic bounds on the fluctuations of $\widehat{S}(\mathbb{X})$ under $H_0$, which allow us to design a non-asymptotic test of level $\alpha$, presented in Section \ref{sec:def_test}. We also provide a characterization of the detection rate based solely on quantities available to the statistician. Our test is rather intuitive: thanks to its simplicity, its results of can be easy interpreted by practitioners.

Our test critically depends on the choice of the partition of the sample space, which defines the histograms for estimating $f^0$ and $f^1$. We provide an explicit characterization of this dependence, and propose in Section \ref{sec:DROP} a heuristic for selecting a partition that corresponds to a significant signal, leading to a faster detection rate. This heuristic is adapted from Classification and Regression Trees (CART) and can be easily implemented using standard packages. 

Simulations provided in Sections \ref{sec:exp_part} and \ref{sec:exp_simul} illustrate respectively the efficiency of the partitioning algorithm in identifying high-signal regions in the space, and the good performances of our test.  In Section \ref{sec:exp_MMD}, we conduct a thorough investigation of the empirical power of our test using simulated datasets, comparing our algorithm with a bootstrap version of the test, as well as with a benchmark test for two-sample testing. Our results highlight the computational and statistical efficiency of our method. Finally, in Section \ref{sec:exp_HIPC} we apply our algorithm to analyze a publicly available flow cytometry dataset from the Human Immune Phenotyping Consortium.

Proofs and additionnal experiments are postponed to the Appendix.

\section{A non-asymptotic test for contamination detection}\label{sec:test}
\subsection{Estimation of the densities} \label{sec:histo}
Building upon our preliminary findings on the toy problem with known densities, we propose to estimate the statistic $S_n(\mathbb{X})$ using the training samples $\mathbb{X}^0$ and $\mathbb{X}^1$. This requires estimating the densities $f^0$ and $f^1$, which we do using histograms. 

\paragraph{Estimation by piecewise-constant functions} Let us divide the space $[0,1]^d$ into $K$ bins using a partition $\mathcal{P} = \left(B_k\right)_{k \leq K}$. The unknown functions $f^0$ and $f^1$ can be approximated using piece-wise constant functions $\overline{f}^0_{\mathcal{P}}$ and $\overline{f}^1_{\mathcal{P}}$, defined as follows:
\begin{align*}
    &\overline{f}^0_{\mathcal{P}}(x) = \sum_{k\leq K}h_k^0 \mathds{1}\{x\in B_k\} \quad \quad \text{ and }\quad \quad \overline{f}^1_{\mathcal{P}}(x) = \sum_{k\leq K}h_k^1 \mathds{1}\{x\in B_k\},
\end{align*}
where $h_k^0 = \frac{\int_{B_k}f^0(x)dx}{\int_{B_k}dx}$ and $h_k^1 = \frac{\int_{B_k}f^1(x)dx}{\int_{B_k}dx}$. The optimal choice of partition, and particularly the number of bins $K$, involves a trade-off between the variance of estimates for the probabilities $h_k^0$ and $h_k^1$, and the discretization cost incurred while approximating the densities using piecewise linear functions. This approximation indeed comes at the expense of diminishing the signal when testing $H_0$ against $H_1$. Specifically, if these functions are positive, Lemma \ref{lem:known_f} reveals that the signal corresponding to the problem with densities $\overline{f}^0_{\mathcal{P}}$ and $\overline{f}^1_{\mathcal{P}}$ is given by 
\begin{align}
    \overline{\sigma}^2_{\mathcal{P}} = \int\left(\frac{\overline{f}^1_{\mathcal{P}}(x)}{\overline{f}^0_{\mathcal{P}}(x)}-1\right)^2\overline{f}^0_{\mathcal{P}}(x)dx.
\end{align}
We can easily verify that $\overline{\sigma}^2_{\mathcal{P}}\leq \sigma^2$ (see Lemma \ref{app:sigma_decroit} in Appendix). In general, the choice of the partition $\mathcal{P}$ will impact our ability to distinguish $\overline{f}^0_{\mathcal{P}}$ from $\overline{f}^1_{\mathcal{P}}$, thereby influencing the test's power. The choice of this partition will be discussed in Section  \ref{sec:DROP}; for now, we consider $\mathcal{P}$ as fixed, and independent from the training and testing samples.

\paragraph{Histogram estimators} Recall that we observe samples $\mathbb{X}^0 =(X_1^0, \ldots, X_{n^0}^0)$ drawn i.i.d. from $f^0$, $\mathbb{X}^1 =(X_1^1, \ldots, X_{n^1}^1)$ drawn i.i.d. from $f^1$, and $\mathbb{X} =(X_1, \ldots, X_{n})$ drawn i.i.d. from $f$. For all $k\leq K$, let $N^0_k$ (resp. $N^1_k$ and $N_k$) be the number of points from sample $\mathbb{X}^0$ (resp. $\mathbb{X}^1$ and $\mathbb{X}$) falling into bin $B_k$. More formally, we have:
\begin{align*}
    N^0_k= \sum_{i\leq n^0}\mathds{1}\{X_i^0 \in B_k\}\text{, }\qquad N^1_k= \sum_{i\leq n^1}\mathds{1}\{X_i^1 \in B_k\}, \qquad \text{and }\qquad N_k= \sum_{i\leq n^1}\mathds{1}\{X_i \in B_k\}.
\end{align*}
Recall that our test statistic estimates $S_n(\mathbb{X}) = \frac{1}{n}\sum_{i\leq n}\left(\frac{f^1(X_i)}{f^0(X_i)} - 1\right)$, and that our test rejects $H_0$ for large values of $\widehat{S}(\mathbb{X})$. Thus, a significant variance in our estimate of $f^0$ in regions where it is small can substantially impact the quality of our test. To ensure the robustness of our estimates, the estimators for $h^0_k$ are lower-thresholded at the levels $\epsilon^0$, where 
\begin{align}\label{eq:def_eps_0}
    \epsilon^0 = \frac{3u}{n^0} \vee \frac{t}{n},
\end{align}
and $u$ and $t$ are some positive constants to be defined later. We argue that lower-thresholding the probabilities $h_k^0$ at $\epsilon_0$ does not significantly diminish the test signal for this bin, given by $\big(\frac{h_k^1}{h_k^0}-1\big)^2 h_k^0$. Indeed, the preliminary results outlined in Lemma \ref{lem:known_f_somewhat_fast_rates} indicate that a probability $h^0_k$ of order $1/n$ or lower will have the same contribution to the signal strength.

Following the same reasoning, we propose to threshold the estimates for $h^1_k$ if the number of points falling into bin $B_k$ is excessively small. Lemma \ref{lem:known_f_somewhat_fast_rates} shows that regions where $h_k^1 \leq 1/\sqrt{n}$ have a negligible contribution to the signal. For this reason, we propose to set the thresholds for $h^1_k$ at the levels $\epsilon^1$, where 
\begin{align}\label{eq:def_eps_1}
    \epsilon^1 = \sqrt{\frac{3u}{n^1}}
\end{align}
In the subsequent analysis, we assume that $u$ and $t$ are small, while $n^0$ and $n^1$ are large, so that we typically have $\epsilon^0, \epsilon^1 \ll 1$. Moreover, we assume that $n^1 \leq n^0$: indeed, our goal is to detect the presence of points drawn from $f^1$, which are typically rare. Consequently, in most scenarios of interest, the size of the training sample $\mathbb{X}_1$ will be considerably smaller than that of $\mathbb{X}_0$.

Finally, we define the sets where the estimates for the probabilities $h^0_k$ and $h_k^1$ are small and need to be thresholded as follows:
\begin{align*}
    \Omega^0 &= \Big\{k: \frac{N^1_k}{n^1} > \epsilon^1\text{ and } \frac{N^0_k}{n^0} \leq \epsilon^0 \Big\},\\    \Omega^{01} &= \Big\{k:  \frac{N^1_k}{n^1} \leq \epsilon^1\text{ and } \frac{N^0_k}{n^0} \leq \epsilon^1 \Big\},\\
     \text{ and } \qquad \Omega^1 &= \Big\{k:  \frac{N^1_k}{n^1} \leq \epsilon^1 \text{ and }\frac{N^0_k}{n^0} > \epsilon^1 \Big\}.
\end{align*}
We emphasize that these sets do not intersect. The set $\Omega^0$ includes bins where $N_1^k$ is sufficiently large for estimating $h^1_k$, but where $N_0^k$ is too small and needs to be thresholded. The set $\Omega^1$ includes bins where $N_0^k$ is sufficiently large for estimating $h^0_k$, but where $N_1^k$ is too small and need to be thresholded. Finally, the set $\Omega^{01}$ corresponds to bins where both densities are low. The signal in these bins is too weak to provide meaningful information, so we set the density ratio to 1 within these bins.

Having introduced the thresholds $\epsilon^0$ and $\epsilon^1$, as well as the sets $\Omega^0$, $\Omega^1$, and $\Omega^{01}$, we are now ready to introduce the thresholded histogram estimators. They are defined as follows:
\begin{align}\label{eq:def_h_est}
    \widehat{h}^0_k = \begin{cases}
        3\epsilon^{0} \text{ if } k \in  \Omega^{0},\\
        3\epsilon^{1} \text{ if } k \in  \Omega^{01},\\
        \frac{N^0_k}{n^0} \text{ else, }
    \end{cases} \quad \quad \text{ and }\quad \quad 
    \widehat{h}^1_k = \begin{cases}
        3\epsilon^{1} \text{ if } k \in  \Omega^{1} \cup \Omega^{01},\\
        \frac{N^1_k}{n^1} \text{ else.}
    \end{cases}
\end{align}

\subsection{Estimated Density Ratio Test} \label{sec:def_test}
Following the preliminary results presented in Lemma \ref{lem:known_f}, we consider the following Estimated Density Ratio Test.
\begin{definition}[Estimated Density Ratio Test]\label{defi:T}
For a desired confidence level $\alpha \in (0,1)$, let $u = \log\left(\frac{4K}{\alpha}\right)$, $t = \log\left(\frac{2}{\alpha}\right)$, and let $\epsilon^0$, $\epsilon^1$, $\widehat{h}^0_k$ and $\widehat{h}^1_k$, be defined respectively by Equations \eqref{eq:def_eps_0}, \eqref{eq:def_eps_1}, and \eqref{eq:def_h_est}. 
Let us define
\begin{align*}
\widehat{r}_k = \frac{\widehat{h}_k^1}{\widehat{h}_k^0},  \qquad \widehat{S}(\mathbb{X}) = \frac{1}{n} \sum_{k\leq K}\left(\widehat{r}_k - 1\right)N_k, \qquad \text{ and }\qquad \widehat{\sigma}^2&= \sum_{k \leq K}\left(\widehat{r}_k - 1\right)^2\widehat{h}_k^0.
\end{align*}
The Estimated Density Ratio Test is defined as the test rejecting $H_0$ if
\begin{align*}
    \widehat{S}(\mathbb{X}) \geq \sqrt{\widehat{\sigma}^2}\left(\sqrt{\frac{10uK}{n^0}} +\sqrt{\frac{6t}{n}}\right)+ \frac{t}{3n}\max_{k}\vert \widehat{r}_k - 1\vert +3 \epsilon^1 K^1
\end{align*}
where $K^1 = \left(\vert \Omega^{01}\vert + \vert \Omega^1\vert \right).$
\end{definition}
\noindent Note that the test statistic $\widehat{S}(\mathbb{X})$ defined above verifies
\begin{align*}
     \widehat{S}(\mathbb{X}) = \frac{1}{n}\sum_{i\leq n}\left(\frac{\widehat{f}^1(X_i)}{\widehat{f}^0(X_i)} - 1\right),
\end{align*} 
where 
\begin{align*}
    &\widehat{f}^0(x) = \sum_{k\leq K}\widehat{h}_k^0 \mathds{1}\{x\in B_k\} \quad \quad \text{ and }\quad \quad \widehat{f}^1(x) = \sum_{k\leq K}\widehat{h}_k^1 \mathds{1}\{x\in B_k\},
\end{align*}
are the estimators of $f^0$ and $f^1$. Thus, the statistic $\widehat{S}$ approximates the test statistic defined in Lemma \ref{lem:known_f}. Similarly, $\widehat{\sigma}^2$ is an estimator for the signal $\overline{\sigma}^2_{\mathcal{P}}$ corresponding to the problem where the densities $f^0$ and $f^1$ are approximated by their piece-wise constant approximation on the partition $\mathcal{P}$. The following theorem controls the type I and type II error of the Estimated Density Ratio Test.
\begin{theorem}\label{thm:type_I_II}
Assume that $3\epsilon^0 \leq \epsilon^1\leq 1$, and that $n^1 \leq n^0$. Then, the Estimated Density Ratio Test has type I error lower than $\alpha$. Moreover, under the alternative hypothesis $H_1$, this test rejects the hypothesis $H_0$ with probability larger than $1-\alpha$ if $\theta$ satisfies
\begin{align*}
\theta \geq& C\left(\sqrt{\frac{t}{n\widehat{\sigma}^2}}\left(1 + K^0\sqrt{\frac{t}{n\widehat{\sigma}^2}}\right) + \frac{K^1\sqrt{u}}{\widehat{\sigma}^2\sqrt{n^1}} + \sqrt{\frac{u}{n^0\widehat{\sigma}^2}}\left(\sqrt{K} + K^0 \sqrt{\frac{u}{n^0\widehat{\sigma}^2}}\right)\right).
\end{align*}
where $C$ is an absolute constant, $K^1 = \vert \Omega^1\vert + \vert \Omega^{01}\vert$, and $K^0 = \vert \Omega^0\vert$.
\end{theorem}
Theorem \ref{thm:type_I_II} establishes that the Estimated Density Ratio Test at a confidence level $\alpha$ indeed has type I error smaller than $\alpha$. Moreover, it offers a characterization of alternatives that are detected with a probability greater than $1-\alpha$. It is important to note that this characterization depends solely on known quantities, such as $\widehat{\sigma}^2$, $n^0$, $n^1$, and $n$. Thus, the results presented in Theorem \ref{thm:type_I_II} enable the statistician to assess, based solely on the training data and on the size of the test sample, the typical proportions $\theta$ of elements drawn from $f^1$ that the test could detect. 

As a first remark, note that we typically have $K^0\sqrt{\frac{t}{n\widehat{\sigma}^2}} \ll 1$ and $K^0 \sqrt{\frac{u}{n^0\widehat{\sigma}^2}} \ll 1$, so the detection rate is driven by $\sqrt{\frac{t}{n\widehat{\sigma}^2}} \vee \frac{K_1 \sqrt{u}}{\widehat{\sigma}^2\sqrt{n^1}} \vee \sqrt{\frac{uK}{n^0\widehat{\sigma}^2}}$. This result highlights the very different roles played by the sizes of the training samples $n^0$ and $n^1$, along with the size of the test sample $n$. When enough observations from the distributions $f^1$ fall in all bins, and $K^1 = 0$, the detection rates does not depend on $n^1$ (provided that $\epsilon^1 \leq 1$). In words, the (constant) additive error in the estimates $h_1^k$ implies an error in estimating the signal $\widehat{\sigma}$, which is deflated by a multiplicative factor. Note that this behavior breaks down when $K^1 > 0$, as the signal is poorly estimated in those bins. 

By contrast, both $n$ and $n^0$ can in general be limiting factors in the detection rate. On the one hand, when $n \geq \frac{n^0}{K}$, the size of the training sample $\mathbb{X}^0$ becomes limiting, restricting the precision with which the distribution $f^0$ is estimated. On the other hand, when $n \leq \frac{n^0}{K}$, the proportion $\theta$ required to reject $H_0$ with a constant probability is dictated by the size of the test sample. In the following section, we analyze this detection rate under different assumptions regarding the densities $f^0$ and $f^1$, corresponding to the various detection regimes examined in our preliminary remarks.

\subsection{Detection rate}\label{sec:detection_rates}
\paragraph{Positive densities} We begin by analyzing the regret rate under the assumptions of Lemma \ref{lem:known_f}: in this setting, both $f^0$ and $f^1$ have a positive density, and therefore, so do $\overline{f}^0_{\mathcal{P}}$ and $\overline{f}^1_{\mathcal{P}}$. In fact, we demonstrate that under the milder assumption that  the signal $\overline{\sigma}^2_{\mathcal{P}}$ is finite, the estimated signal $\widehat{\sigma}^2$ is close to $\overline{\sigma}^2_{\mathcal{P}}$. This is done in the following Lemma.
\begin{lemma}\label{lem:control_sigma}
Assume that $3\epsilon^0 \leq \epsilon^1\leq 1$, that $n^1 \leq n^0$, and that $\overline{\sigma}^{2}_{\mathcal{P}}$ is finite. Then, there exists a constant $C_{\overline{f}_{\mathcal{P}}^0, \overline{f}_{\mathcal{P}}^1, K}$ depending only on $\overline{f}_{\mathcal{P}}^0$ and $ \overline{f}_{\mathcal{P}}^1$ such that for all $n$, $n^0$ and $n^1$ large enough, with probability $1-\alpha$,
    \begin{align*}
    \left \vert \widehat{\sigma}^2- \overline{\sigma}^2_{\mathcal{P}} \right \vert \leq C_{\overline{f}_{\mathcal{P}}^0, \overline{f}_{\mathcal{P}}^1, K}\left(\sqrt{\frac{u}{n^1}}+\sqrt{\frac{u}{n^0}}\right).
\end{align*}
\end{lemma}
Lemma \ref{lem:control_sigma} reveals that $\widehat{\sigma}^2 \rightarrow \overline{\sigma}^2_{\mathcal{P}}$ as the size of the training samples goes to infinity. Combining this result with Theorem \ref{thm:type_I_II} yields the following corollary.
\begin{corollary}\label{cor:power_slow_rate}
     Assume that $3\epsilon^0 \leq \epsilon^1\leq 1$, that $n^1 \leq n^0$, and that $\overline{\sigma}^{2}_{\mathcal{P}}$ is finite. Then, there exists a constant $C_{\alpha, K}$ depending on $\alpha$ such that for all $n$, $n^0$ and $n^1$ large enough, the Estimated Density Ratio Test of level $\alpha$ has power $1-\alpha$ against alternatives 
    \begin{align*}
        \theta \geq& \frac{C_{\alpha, K}}{\sqrt{\overline{\sigma}_{\mathcal{P}} \left(n \land n^0\right)}}.
    \end{align*}
\end{corollary}
Under the hypothesis of Corollary \ref{cor:power_slow_rate}, when $n^0, n^1 \geq n$, the detection rate is of order $1/\sqrt{\overline{\sigma}_{\mathcal{P}}n}$. In this case, the test is asymptotically equivalent (up to a numerical multiplicative constant) to the test in the setting where $\overline{f}_{\mathcal{P}}^0$ and $\overline{f}_{\mathcal{P}}^1$ are known. On the contrary, when the training samples are too small, the test's complexity is driven by the challenge of learning $\overline{f}_{\mathcal{P}}^0$ and $\overline{f}_{\mathcal{P}}^1$.

If we consider regimes where the number of points in both training samples becomes large compared to the size of the testing samples, we anticipate that the inherent complexity of the problem reduces to that of the toy problem where the densities $f^0$ and $f^1$ are known. A closer examination of Lemma \ref{lem:control_sigma} reveals that the bound on $\left \vert \widehat{\sigma}^2- \overline{\sigma}^2_{\mathcal{P}} \right \vert$ remains valid as long as the number of training points from the samples $\mathbb{X}^0$ and $\mathbb{X}^1$ remains substantial in each bin (specifically, respectively larger than a constant, and larger than $\sqrt{n^1}$). Therefore, if $f^0$ and $f^1$ are positive, when $n^0/n, n^1/n \rightarrow \infty$, we can choose a sequence of partitions $\mathcal{P}$ of increasing size to better approximate the density functions. Under additional assumptions on the regularity of the functions $f^0$ and $f^1$, we then have $\overline{\sigma}^2_{\mathcal{P}} \rightarrow \sigma^2$ as the training sample sizes increase. In the limit where $n^0/n, n^1/n \rightarrow + \infty$, the test has asymptotically constant power against alternatives of order $\sqrt{n\sigma^2}$, and the problem indeed reduces to that with known densities.

\paragraph{Faster rates for densities with disjoint support} When the densities $f^0$ and $f^1$ have disjoint support, or when certain regions of the space have low probability under $f^0$ and higher probability under $f^1$, we hope to achieve a faster detection rate than $1/\sqrt{n}$,  i.e., to recover the fast detection rates highlighted by our preliminary analysis. However, this can only happen if one of the bins in the partition coincides with such a highly informative region. Indeed, if the piece-wise constant approximations to $f^0$ and $f^1$ are positive, the conclusions of the previous analysis still hold. In this case, the test cannot leverage the information contained in low-density areas under $f^0$ because they are not captured by the partition. On the contrary, if there exists a bin $B_k$ such that $h_k^0 = 0$ and $h_k^1>0$, classical concentration results reveal that $\widehat{h}_k^1$ is asymptotically of order $h_k^1$, while $h_k^0$ is always thresholded at level $\epsilon^0$.  The following Lemma indicates that in this scenario, the detection rate can be of order $1/n$. We recall that $h_k^1 = \nicefrac{\int_{B_k}f^1(x)dx}{\int_{B_k}dx}$ is the probability that a points belongs to bin $B_k$ under $f^1$.

\begin{lemma}\label{lem:lb_h}
Assume that $3\epsilon^0 \leq \epsilon^1\leq 1$, that $n^1 \leq n^0$, and that there exists $k \leq K$ such that $h_k^1 \geq \sqrt{\frac{5u}{n^1}}$ and $h_k^0 = 0$. Then, with probability $1-\alpha$, 
\begin{align*}
    \widehat{\sigma}^2&\geq \left(1 - \left(\frac{u}{n^1}\right)^{\frac{1}{4}}\right)^2 \frac{h_k^1}{\epsilon^0}.
\end{align*}
where we recall that $\epsilon^0 = \nicefrac{3u}{n^0} \vee \nicefrac{t}{n}$. Furthermore, for $n$, $n^0$, and $n^1$ large enough, the Estimated Density Ratio Test has power $1-\alpha$ against alternatives
\begin{align*}
    \theta & \geq \frac{C_{\alpha, K}\epsilon^0}{h_k^1},
\end{align*}
where the constant $C_{\alpha, K}>0$ depends on $\alpha$ and $K$.
\end{lemma}
Lemma \ref{lem:lb_h} characterizes the detection rate in the case where the partition $\mathcal{P}$ captures a high-signal region $B_k$ such that $h_k^0 = 0$ and $h_k^1 >0.$ In this case, fast detection rates can be achieved: if $n^0> n$, the test detects alternatives corresponding to proportion $\theta$ of order $1/(nh^1_k)$. Note that as the size of the training sample $\mathbb{X}^1$ increases, our test is able to take into account high-signal regions with ever decreasing probabilities $h_k^1$. In the limit where $n^1 \approx n$, these regions have probabilities of order $n^{-1/2}$, corresponding to the limiting regime highlighted in Lemma \ref{lem:known_f_somewhat_fast_rates}.

\subsection{DROP: a partitioning algorithm}\label{sec:DROP}
The analysis of the two cases outlined above, corresponding to the regimes explored in Lemmas \ref{lem:known_f} and \ref{lem:known_f_somewhat_fast_rates} yielding respectively parametric and fast rates, emphasizes the critical role of partition selection. In the first scenario, it was demonstrated that the detection rate is asymptotically $1/\sqrt{n \overline{\sigma}_{\mathcal{P}}^2}$. In this situation, opting for a partition $\mathcal{P}$ associated with a substantial signal $\overline{\sigma}_{\mathcal{P}}^2$ allows for an enhancement of this rate by a multiplicative factor. In the second scenario, the choice of partition becomes even more pivotal, as selecting a partition that captures regions contributing the most to the signal is crucial in achieving fast detection rates. 

The remainder of this section is dedicated to introducing the Density-Ratio Oriented Partitioning (DROP) algorithm. In order to obtain meaningful piece-wise constant approximations of the densities $f^0$ and $f^1$, we propose constructing a sequence of nested partitions with increasing sizes. The optimal partition size will then be selected to maximize the lower bound on the test power as established in Theorem \ref{thm:type_I_II}. We begin by outlining our partitioning approach, followed by a discussion on the selection of the optimal partition within this sequence. As a first step, we split the training sample $\mathbb{X}^0$ (resp. $\mathbb{X}^1$) into two samples each, denoted $\mathbb{X}^0_{part}$  and $\mathbb{X}^0_{est}$  (resp. $\mathbb{X}^1_{part}$  and $\mathbb{X}^1_{est}$ ), of respective sizes $n^0_{part}$ and $n^0_{est}$ (resp. $n^1_{part}$ and $n^1_{est}$). The samples $\mathbb{X}^0_{part}$  and $\mathbb{X}^0_{part}$  will be used for choosing the partition, while the samples $\mathbb{X}^0_{est}$  and $\mathbb{X}^1_{est}$  will be used for estimating the densities. This separation ensures that the density estimates remain independent of the partition, preserving the validity of the theoretical guarantees established for a fixed partition in Theorem \ref{thm:type_I_II}.

\paragraph{A greedy partitioning scheme} To construct a partition with a significant signal $\widehat{\sigma}^2$, we begin by generating a sequence of nested partitions with increasing size $(\mathcal{P}_K)_{K\in \mathbb{N}}$. At each iteration $K$, the partition $\mathcal{P}_K$ is obtained by subdividing a bin from $\mathcal{P}_{K-1}$ into two separate bins. To achieve a partition $\mathcal{P}_K$ with a substantial estimated signal $\widehat{\sigma}_{\mathcal{P}_K}^2$, we choose this division in a greedy manner. Before specifying the splitting criteria, let us introduce additional notations.

We emphasize the dependence of the density estimator on the partition by denoting by $\widehat{h_k^0}(\mathcal{P})$ (resp. $\widehat{h_k^1}(\mathcal{P})$) the estimated probability that a point belongs to the $k$-th bin in partition $\mathcal{P}$ under $f^0$ (resp. $f^1$). Assume that the partition $\mathcal{P}_K$ is obtained by splitting the $k-th$ bin of partition $\mathcal{P}_{K-1}$. We denote the resulting bins as $k_L$ and $k_R$, referred to as the ``left" and ``right" bins, respectively. The increase in signal, denoted by $\Delta$, can be expressed as $\Delta = \widehat{\sigma}_{\mathcal{P}_K}^2 - \widehat{\sigma}_{\mathcal{P}_{K-1}}^2$, or equivalently as
\begin{align}\label{eq:impurity}
\Delta = & \left(\frac{\widehat{h}_{k_R}^1(\mathcal{P}_K)}{\widehat{h}_{k_R}^0(\mathcal{P}_K)}-1\right)^2\widehat{h}_{k_R}^0(\mathcal{P}_K) + \left(\frac{\widehat{h}_{k_L}^1(\mathcal{P}_K)}{\widehat{h}_{k_L}^0(\mathcal{P}_K)}-1\right)^2\widehat{h}_{k_L}^0(\mathcal{P}_K) \nonumber \\
& \qquad - \left(\frac{\widehat{h}_{k}^1(\mathcal{P}_{K-1})}{\widehat{h}_{k}^0(\mathcal{P}_{K-1})}-1\right)^2\widehat{h}_{k}^0(\mathcal{P}_{K-1}).
\end{align}
To obtain partition $\mathcal{P}_{K}$ from partition $\mathcal{P}_{K-1}$, the split that maximizes the increase in signal $\Delta$ is selected. We recall that all quantities in Equation \eqref{eq:impurity} are estimated using the samples $\mathbb{X}^1_{part}$ and $\mathbb{X}^1_{part}$. Note that both $\widehat{h}_{k}^0$ and $\widehat{h}_{k}^1$ depend on the size of the partition $K$ though $\epsilon^0$ and $\epsilon^1$, and more precisely through the parameter $u$, set as $\log(4K/\alpha)$ in Theorem \ref{thm:type_I_II}. We adopt a conservative approach, and choose beforehand a maximum size $K_{\max}$ for the partition. For reasons that will become clear later on, we then set $u = \log(8K_{\max}/\alpha)$.

\paragraph{DROP as a Classification And Regression Tree algorithm} This greedy partitioning approach can be formulated as a variation of the Classification And Regression Tree (CART) algorithm. Originally developed for classification and regression problems, CART relies on an impurity measure, such as the Gini impurity measure for classification tasks or the variance of responses in regression problems. The goal of CART is to create a partition of the covariate space, or equivalently, a tree structure, where the leaf nodes have low impurity. To achieve this, the algorithm examines every potential split at each step, considering all possible directions and splitting points. It computes the reduction in impurity corresponding to each split and selects the one with the maximum reduction. After constructing the tree, various heuristics can be applied to prune it and prevent overfitting.

To adapt this algorithm to our problem, we introduce an artificial labeling scheme: we assign the label $Y=0$ (respectively, $Y=1$) to the points in the dataset $\mathbb{X}^0_{part}$  (respectively, $\mathbb{X}^1_{part}$ ). These labeled datasets are then concatenated. Finding a split with maximum signal reduces to discovering a split with low impurity measure, where the impurity measure is given by $-\Delta$, and $\Delta$ is defined in Equation \eqref{eq:impurity}.

\paragraph{Choice of the partition size}The aforementioned partitioning scheme allows us to build a sequence of nested partitions with increasing size $(\mathcal{P}_K)_{1\leq K \leq K_{max}}$. Next, we estimate the signal of each partition using the samples $\mathbb{X}^0_{est}$ and $\mathbb{X}^1_{est}$. We argue that choosing $u = \log(4 \times 2K_{\max}/\alpha)$ allows to control the error of $\widehat{h}_k^0$ and $\widehat{h}_k^1$ obtained using the sample $\mathbb{X}^0_{est}$ and $\mathbb{X}^1_{est}$ uniformly over all bins of all partitions in the sequence, since they are at most $2K_{max}$ such bins. Consequently, we select the partition $\mathcal{P}_{K^*}$, where
\begin{align}\label{eq:choice_K}
    K^* \in \underset{K\leq K_{max}}{\text{argmin}}\ \left(\sqrt{\frac{t}{n\widehat{\sigma}_{\mathcal{P}_K}^2}}\left(1 + K^0\sqrt{\frac{t}{n\widehat{\sigma}_{\mathcal{P}_K}^2}}\right) + \frac{\sqrt{u}K^1}{\widehat{\sigma}_{\mathcal{P}_K}^2\sqrt{n^1}} + \sqrt{\frac{u}{n^0\widehat{\sigma}_{\mathcal{P}_K}^2}}\left(\sqrt{K} + K^0 \sqrt{\frac{u}{n^0\widehat{\sigma}_{\mathcal{P}_K}^2}}\right)\right)
\end{align}
where $\widehat{\sigma}_{\mathcal{P}_K}^2$ correspond to the signal estimated on the partition $\mathcal{P}_k$ using the sample $\mathbb{X}^0_{est}$ and $\mathbb{X}^1_{est}$, and $n^0$ is the size of the sample  $\mathbb{X}^0_{est}$. Theorem \ref{thm:type_I_II} implies the following lower bound on the detection rate of our test, when the sequence of partitions is obtained using a sample independent from the sample used to choose the partition and conduct the test.
\begin{corollary}\label{cor:max_signal}
Let $(\mathcal{P}_K)_{K\leq K_{max}}$ be a fixed sequence of nested partitions. For each partition $\mathcal{P}_K$, let EDRT$(\mathcal{P}_{K})$ be the Estimated Density Ratio Test for this partition with parameters  $u = \log(8K_{\max}/\alpha)$ and $t = \log(2/\alpha)$, and let $\widehat{\sigma}^2_{\mathcal{P}_K}$ be the estimated signal for this partition. Then, if $3\epsilon^0 \leq \epsilon^1 \leq 1$ and $n^0\geq n^1$, for the choice $K^*$ given in Equation \eqref{eq:choice_K}, the test EDRT$(\mathcal{P}_{K^*})$ has type I error lower than $\alpha$. Moreover, EDRT$(\mathcal{P}_{K^*})$ has power larger than $1-\alpha$ against alternatives characterized by $\theta$ verifying
\begin{align*}
\theta \geq C\underset{K \leq K_{max}}{\min} \ \ \left(\sqrt{\frac{t}{n\widehat{\sigma}_{\mathcal{P}_kK}^2}}\left(1 + K^0\sqrt{\frac{t}{n\widehat{\sigma}_{\mathcal{P}_K}^2}}\right) + \frac{\sqrt{u}K^1}{\widehat{\sigma}_{\mathcal{P}_K}^2\sqrt{n^1}} + \sqrt{\frac{u}{n^0\widehat{\sigma}_{\mathcal{P}_K}^2}}\left(\sqrt{K} + K^0 \sqrt{\frac{u}{n^0\widehat{\sigma}_{\mathcal{P}_K}^2}}\right)\right)
\end{align*}
where $C$ is an absolute constant.
\end{corollary}

The DROP algorithm, combined with the Estimated Density Ratio Test, provides an easily interpretable test for detecting the presence of a specific type of cells. Indeed, the computation of the signal $\widehat{\sigma}_{\mathcal{P}_{K^*}}$ yields, as a byproduct, the contribution of different bins to the signal, indicating their significance in the test. A region with a substantial contribution corresponds necessarily to an area with a high probability under $f^1$ and a low probability under $f^0$. Intuitively, the test assigns significant weight to observations in the test sample falling into these regions, and it rejects the null hypothesis if they are too numerous. The partitioning of the space using a tree is akin to the process of ``manual gating", and the obtained regions can easily be manually inspected to assess the good behavior of the test.

\section{Experiments}\label{sec:expe}

In this Section, we investigate the empirical performance of our test. We begin by showcasing in Section \ref{sec:exp_part} how the partition selected by DROP captures the high-signal regions of the sample space. In Section \ref{sec:exp_simul}, we investigate the power of the Estimated Density Ratio Test on simulated datasets. Our experiment reveal that this test is overly conservative. To circumvent this problem, we propose in Section \ref{sec:exp_MMD} a Boostraped Estimated Density Ratio Test, and we compare its performance to that of the Estimated Density Ratio Test and of a benchmark algorithm for two-sample testing based on Maximum Mean Discrepancy. Finally, we illustrate our method by applying it to real flow cytometry data in Section \ref{sec:exp_HIPC}. Additional simulations illustrating the robustness of EDRT against variations of the density $f^0$ can be found in Appendix \ref{app:simu}. The code for all experiments can be found at \url{https://github.com/SolenneGaucher/DensityRatioTest}.

    \paragraph{Implementation} As mentioned earlier, the partition-sequence-building component of the DROP algorithm can be considered a variation of the CART algorithm, where the impurity measure is defined as $-\Delta$. Building on this remark, we make use of the {\texttt{R}} \citep{R} package {\texttt{rpart}} \citep{rpart} to implement our algorithm. This package enables users to specify the impurity measure for splitting criteria and handles the required bookkeeping efficiently through an optimized \texttt{C} implementation. Moreover, we implement a constraint in our algorithm, preventing the splitting of a bin if it contains fewer than $3n^1_{part}\epsilon^1$ points. Note that in most settings, $n^1_{part} \leq n^0_{part}$. In this case, if a bin contains less than $3n^1_{part}\epsilon^1$ points, this bin belongs to $\Omega^{01}$, and its contribution to the estimated signal is null. In practice, when selecting the optimal size for the partition, we observe that in our experiments, the leading terms in criteria \eqref{eq:choice_K} are $\sqrt{\frac{t}{n\widehat{\sigma}_{\mathcal{P}_K}^2}}$ and $\sqrt{\frac{u}{n^0\widehat{\sigma}_{\mathcal{P}_K}^2}}$. We therefore simplify our selection criteria, and choose the partition with size 
    \begin{align}\label{eq:true_K}
    K^* \in \underset{K\leq K_{max}}{\text{argmin}}\ \sqrt{\frac{1}{n\widehat{\sigma}_{\mathcal{P}_K}^2}} \vee \sqrt{\frac{K}{n^0\widehat{\sigma}_{\mathcal{P}_K}^2}}.
\end{align}

\paragraph{Simulation settings}  Unless otherwise mentioned, we consider the following three experimental settings. We consider a two-dimensional problem for visualization purposes, where  $f^0$ and $f^1$ are two truncated Gaussians in $[0,1]^2$ with covariance $\left({1/100 \atop 0} {0 \atop 1/100}\right)$. In Setting A, their means are respectively $(3/10, 3/10)$ and $(7/10, 7/10)$.  In Setting B, their means are respectively $(4/10, 4/10)$ and $(6/10, 6/10)$.  In Setting C, their means are respectively $(4/10, 4/10)$ and $(5/10, 5/10)$. Note that these settings are ordered by increasing difficulty of the contamination detection problem, or equivalently by decreasing signal $ \widehat{\sigma}^2$. We denote by $n^{train}$ the total number of training samples, and we assume that $n^0 = 0.7 \times n^{train}$ and $n^1 = 0.3 \times n^{train}$. In our experiments, we let $n^{train}$ vary from $1000$ to $1000000$. 

\subsection{Choice of the partition}\label{sec:exp_part}
In our first experiments, we assess the ability of the DROP algorithm to identify a high-signal partition of the sample space. Throughout this section, we contrast our results with a benchmark CART algorithm using the popular Gini Index. 

\paragraph{CART with Gini Index} When optimizing the partition with regard to the Gini Index, the split is determined to maximize
$$p_L(1-p_L) + p_R(1-p_R)$$
where $p_L$ (and $p_R$) denotes the proportion of points with label 1 in the left (and right) bin, respectively. It is important to note that the Gini index criterion is influenced by the proportion of points in the bin with label 0, leading to potentially distinct partitions based on different values of the sample size ratio $n^0/n^1$. In our simulations, we consider $n^0/n^1 = 7/3$, and we compare the partitions obtained when the partition sequence is constructed using the Gini index splitting criteria, to that obtained using the impurity measure $-\Delta$. The size of the partition is chosen in both cases according to the criterion in Equation \eqref{eq:true_K}.

\paragraph{Experiment} We use half of the observations to grow the partition sequence, while the other half is used to estimate the parameters $\widehat{h}_k^0$, $\widehat{h}_k^1$, $\widehat{r}_k$ and $\widehat{\sigma}$, and choose the partition. Note that in the regimes considered in Experiment 1, Experiment 2 and Experiment 3a, we assume that $n \approx n_0$, so that $K^*$ simply maximizes $\nicefrac{\widehat{\sigma}_{\mathcal{P}_K}^2}{K}$. In Experiment 3b, we investigate the impact of choosing a larger partition size on the signal $\widehat{\sigma}^2$. More precisely, we assume that $\nicefrac{n^0}{n} \geq \log(n^{train})$, so that $K^* \geq \log(n^{train})$.

\begin{figure}[!h]
\subcaptionbox{\small Experiment 1\label{fig:1}}
{\includegraphics[width=0.43\textwidth]{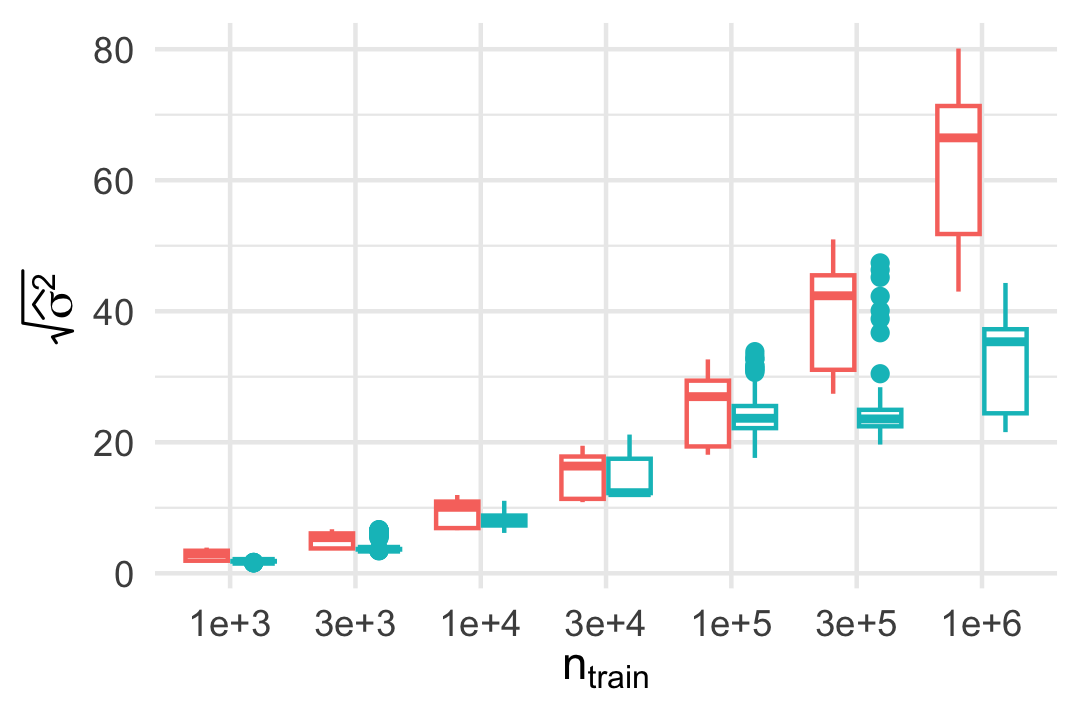}}
\subcaptionbox{\small Experiment 2\label{fig:2}}
{\includegraphics[width=0.43\textwidth]{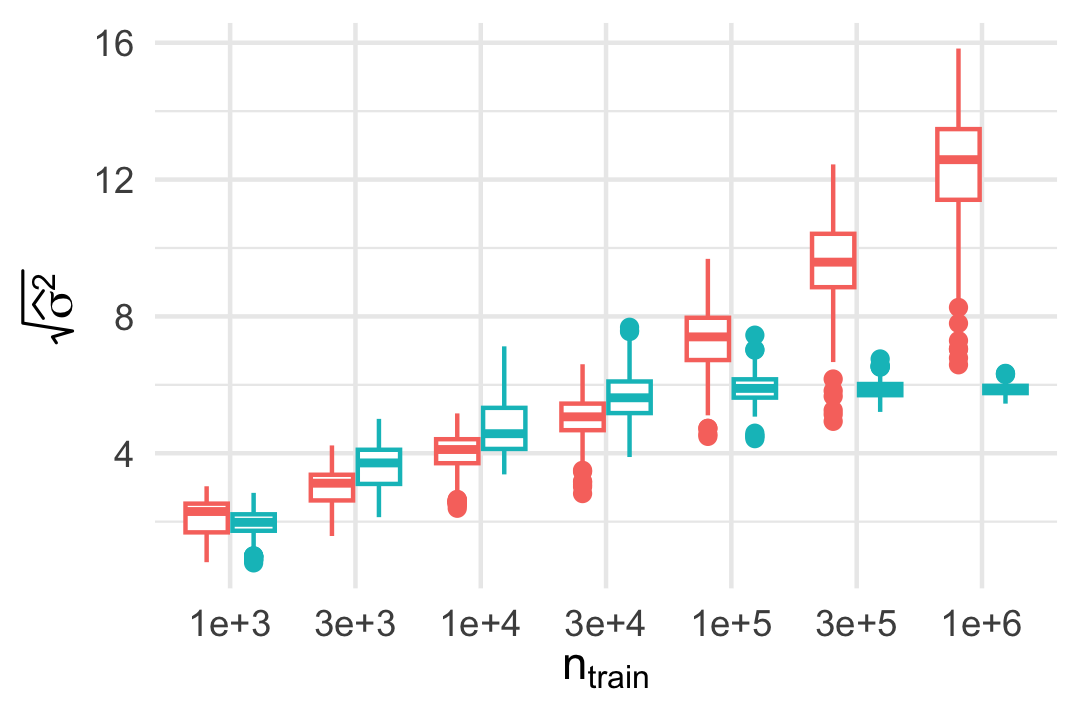}}\\
\subcaptionbox{\small Experiment 3a\label{fig:3a}}
{\includegraphics[width=0.43\textwidth]{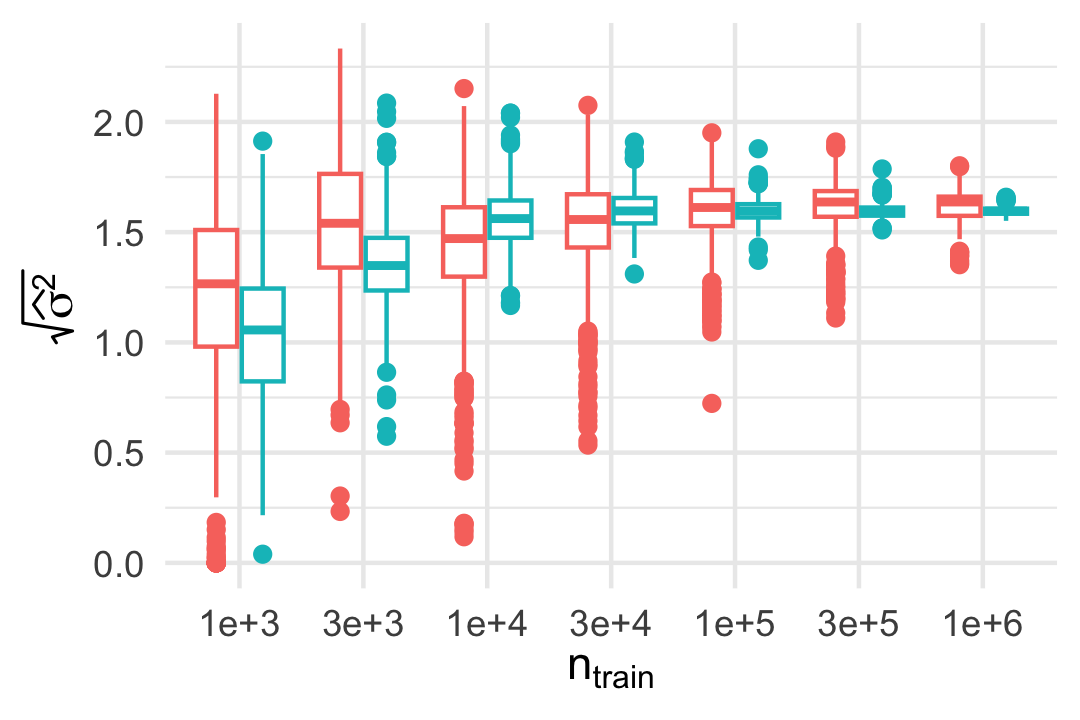}}
\subcaptionbox{\small Experiment 3b\label{fig:3b}}
{\includegraphics[width=0.56\textwidth]{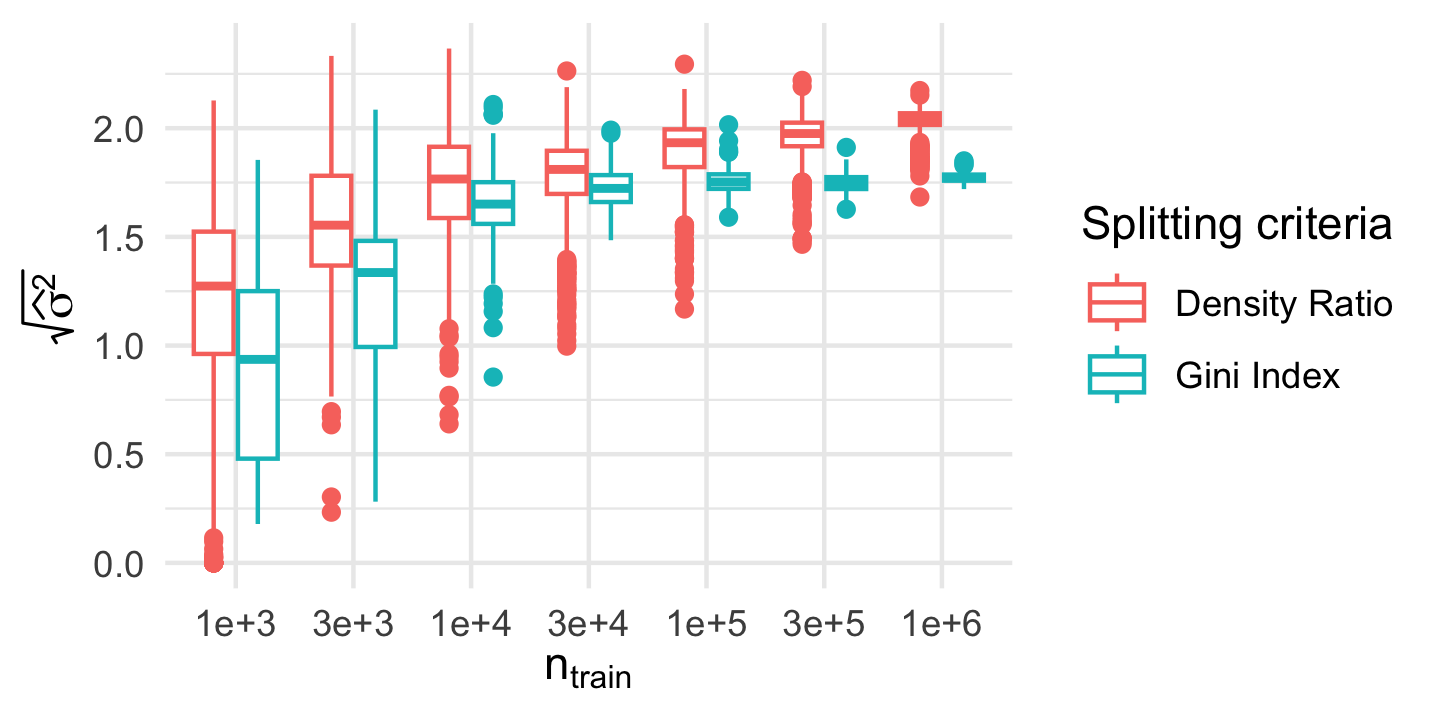}}
\caption{Estimated signal $\sqrt{\widehat{\sigma}^2}$ for different size of training samples $n^{train}$ corresponding to the partition obtained using the density ratio criteria (\texttt{red}) and the Gini index (\texttt{blue}). The distributions $f^0$ and $f^1$ correspond to Setting A (\texttt{top left}), Setting B (\texttt{top right}), and Setting C (\texttt{bottom}). We consider respectively $n \approx n^0$ (\texttt{top} and \texttt{bottom left}) and $\nicefrac{n^0}{n} \approx \log(n^{train})$ (\texttt{bottom right}) when choosing $K^*$ according to criterion \eqref{eq:choice_K}. Each experiment is reproduced 100 times.}
\label{fig:signal_variying_means}
\end{figure}

As anticipated, a comparison of Figures \ref{fig:1}, \ref{fig:2}, and \ref{fig:3a} reveals that as the separation between the means of the two distributions widens, the signal for testing $H_0$ against $H_1$ increases for both choices of partitions. Nevertheless, in some regimes we observe a strong contrast in the signal strength captured by those partitions.

\paragraph{Large signal setting} In Settings A and B, certain regions of the sample space $[0,1]^d$ exhibit high density under $f^1$ and low density under $f^0$. These settings allow for the possibility of achieving detection rates faster than the typical $1/\sqrt{n}$, as demonstrated in Lemma \ref{lem:known_f_somewhat_fast_rates}. Our simulations in Figure \ref{fig:1} and \ref{fig:2} reveal that the partition selected by DROP can adapt to identify these high-signal regions. Furthermore, the estimated signal for these partitions sharply increases with the size of the training sample. Indeed, the density ratio criterion seeks to identify regions with low density under $f^0$ and high density under $f^1$. With an increasing number of points, the partition can be refined in regions of the space where too few points were previously available. 

In contrast, the partition generated using CART with the Gini index appears unable to capture this rise in signal. To understand this phenomenon, note that CART aims to distinguish between areas where the majority of points originates from distribution $f^0$ and regions dominated by points from distribution $f^1$. Consequently, increasing the number of points allows for a more precise estimation of the boundary that separates these two regions without significantly altering its position. Since the bins remain unchanged, the estimated signal within these bins remains relatively constant. This phenomenon is illustrated in Figure \ref{fig:partition}, underlining that the partition resulting from the Gini index remains almost constant as the sample size increase. In contrast, the boundaries delineated by the partition obtained through the density ratio criterion shift towards regions characterized by very low density $f^0$ and high signal $(f^1)^2/f^0$.

\begin{figure}[h!]
\centering
\subcaptionbox{\small Partition in Setting B, $n^0+n^1 = 10 000$\label{fig:part1}}
{\includegraphics[width=0.44\textwidth]{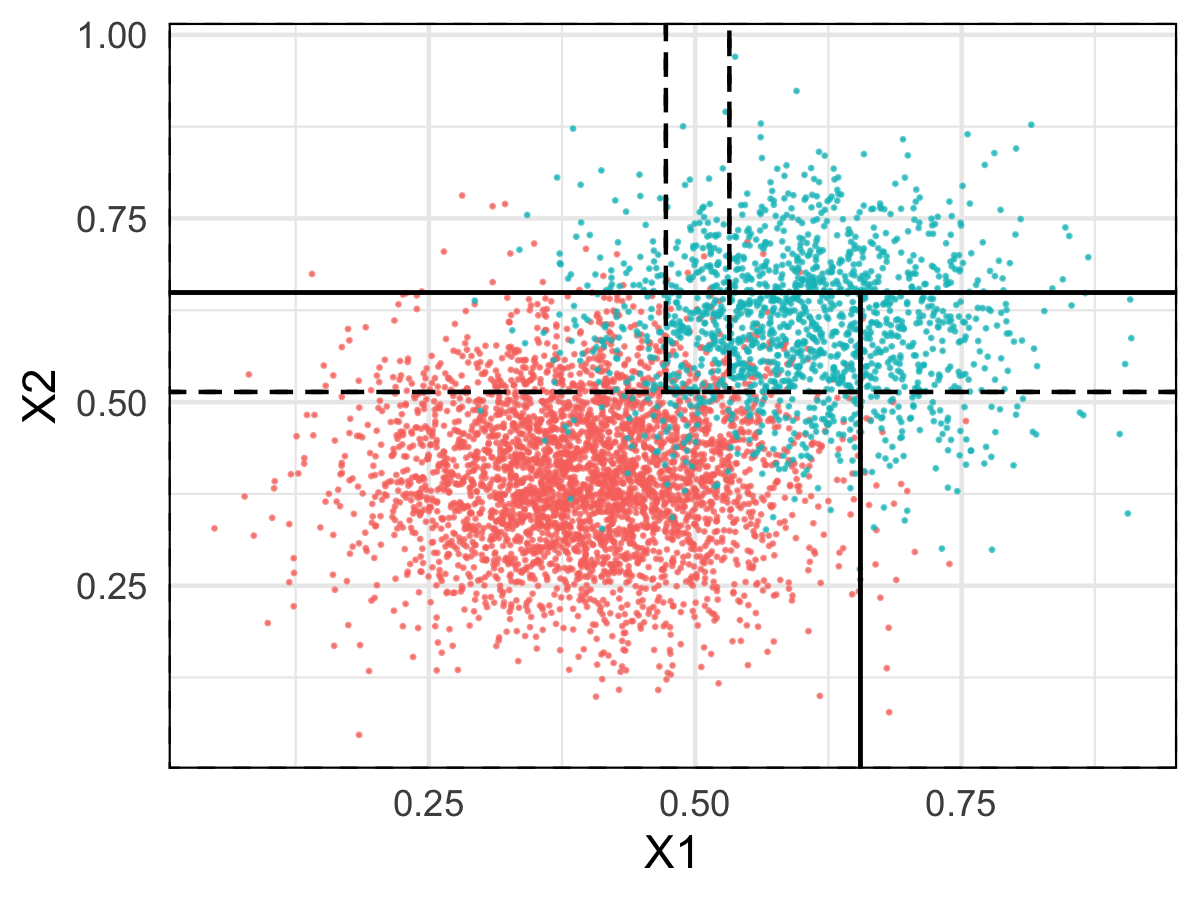}}
\subcaptionbox{\small Partition in Setting B, $n^0+n^1 = 1 000 000$\label{fig:part2}}
{\includegraphics[width=0.55\textwidth]{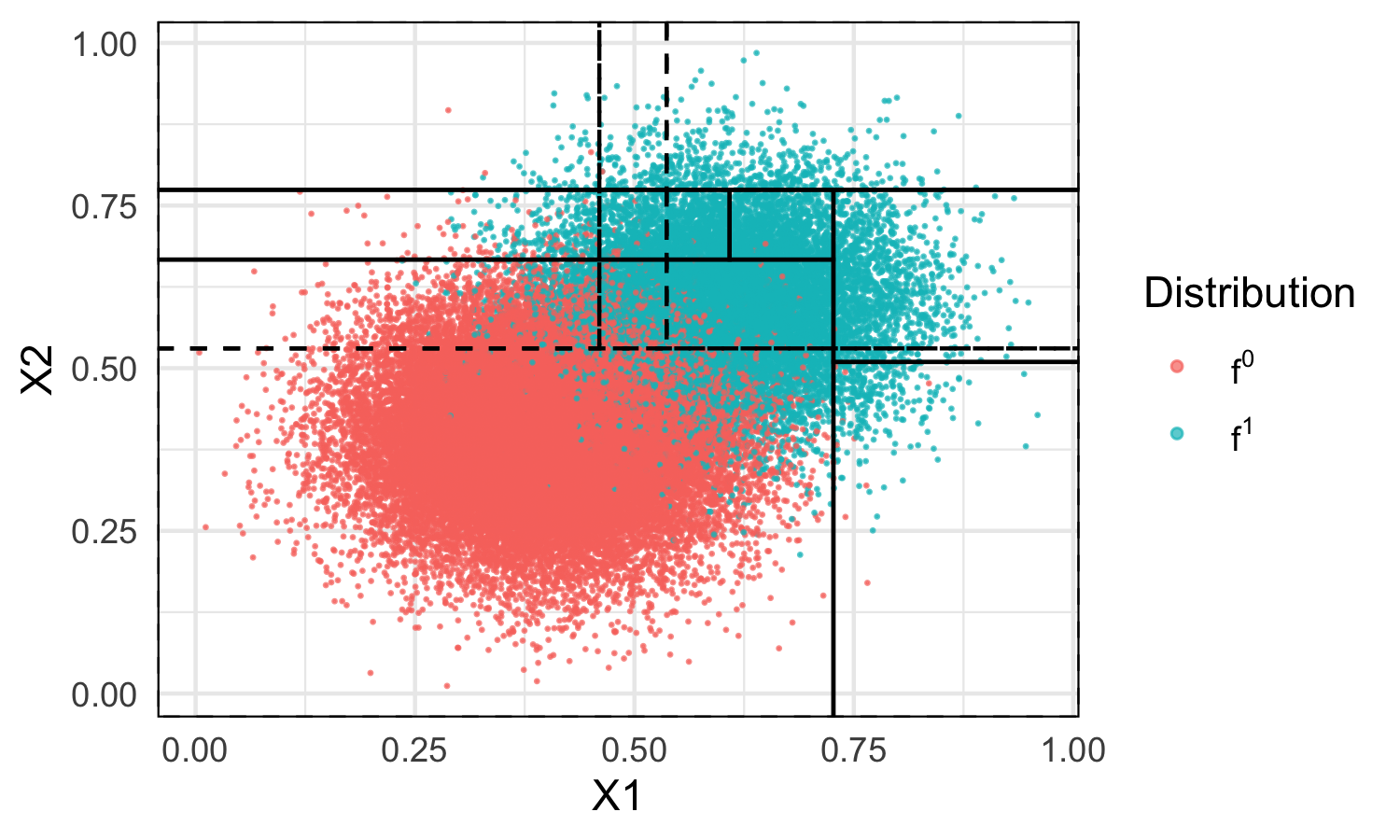}}\\
\caption{Partition chosen in Setting B for $n^{train} = 10 000$ (\texttt{left}) and $n^{train} = 1 000 000$ (\texttt{right}), using the density ratio criteria (full lines) and the Gini index (dotted lines). The points drawn from $f^0$ are in red, the points drawn from $f^1$ are in blue.}
\label{fig:partition}
\end{figure}

\paragraph{Low signal setting}
Experiments 3a and 3b correspond to scenarios where the distributions $f^0$ and $f^1$ are highly similar, making the test problem inherently challenging. Unlike the previous scenarios, the estimated signal does not show significant increase even as the number of points exceeds 1000. It is noteworthy that unless $n^0/n \rightarrow \infty$, and we compel $K^*$ to increase, the partition chosen based on the Gini index and the partition selected using the density ratio criterion yield equivalent signal terms $\widehat{\sigma}^2$. However, if $K^*$ increases with the number of training samples, the signal term appears to exhibit gradual growth with the size of the partition chosen according to the density ratio criterion, while it seems to remain unchanged when using the Gini index. Once again, this observation can be attributed to the fact that the Gini index aims to separate the two densities, while the density ratio criterion seeks regions characterized by high signal $(\nicefrac{f^1}{f^0} - 1)^2 \times f^0$.

\subsection{Type I and II error for simulated data}\label{sec:exp_simul}

\paragraph{Detection threshold for different signal strengths} We now investigate the empirical type I and type II error of the Estimated Density Ratio Test. We consider the distributions $f^0$ and $f^1$ of Settings A, B and C. We simulate mixtures sampled from $(1-\theta)\times f^0 + \theta\times f^1$ for $\theta$ varying between 0.0003 and 0.3.

\begin{figure}[!ht]
\subcaptionbox{\small Setting A.\label{fig:exp9_low}}
{\includegraphics[width=0.3\textwidth]{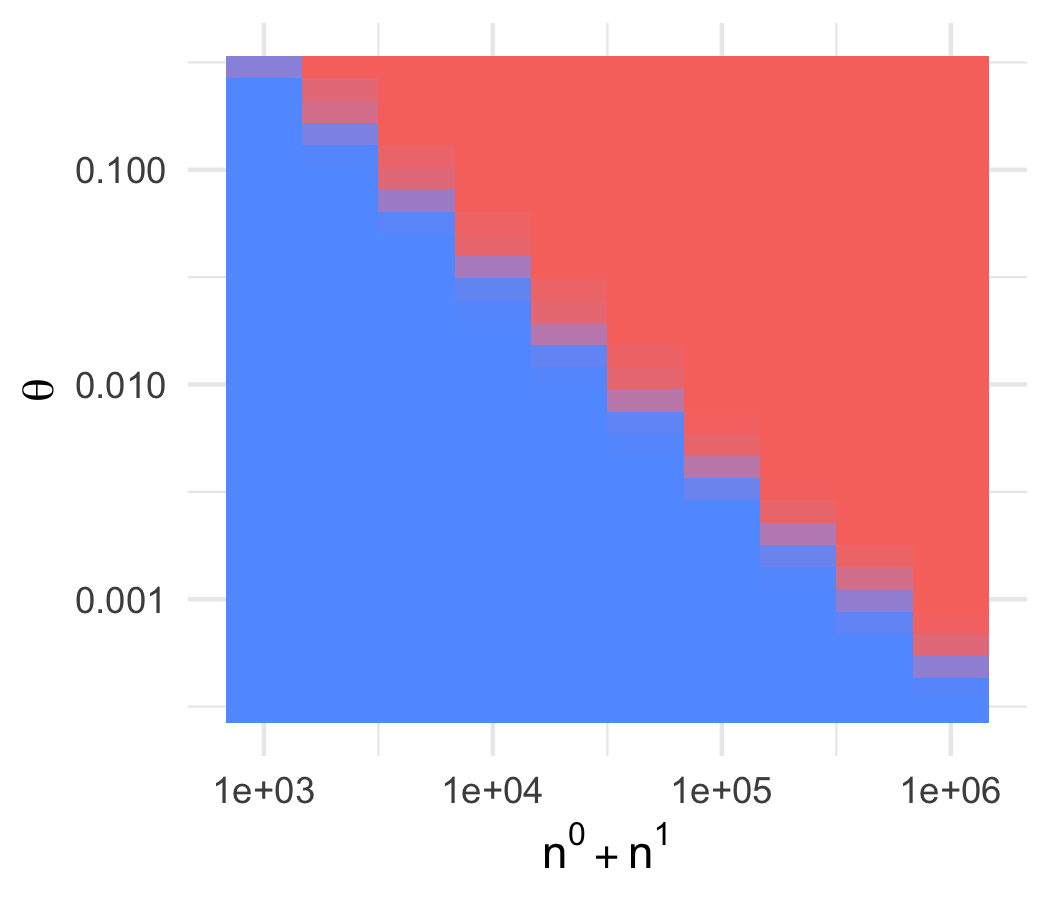}}
\subcaptionbox{\small Setting B.\label{fig:exp9_low_2}}
{\includegraphics[width=0.3\textwidth]{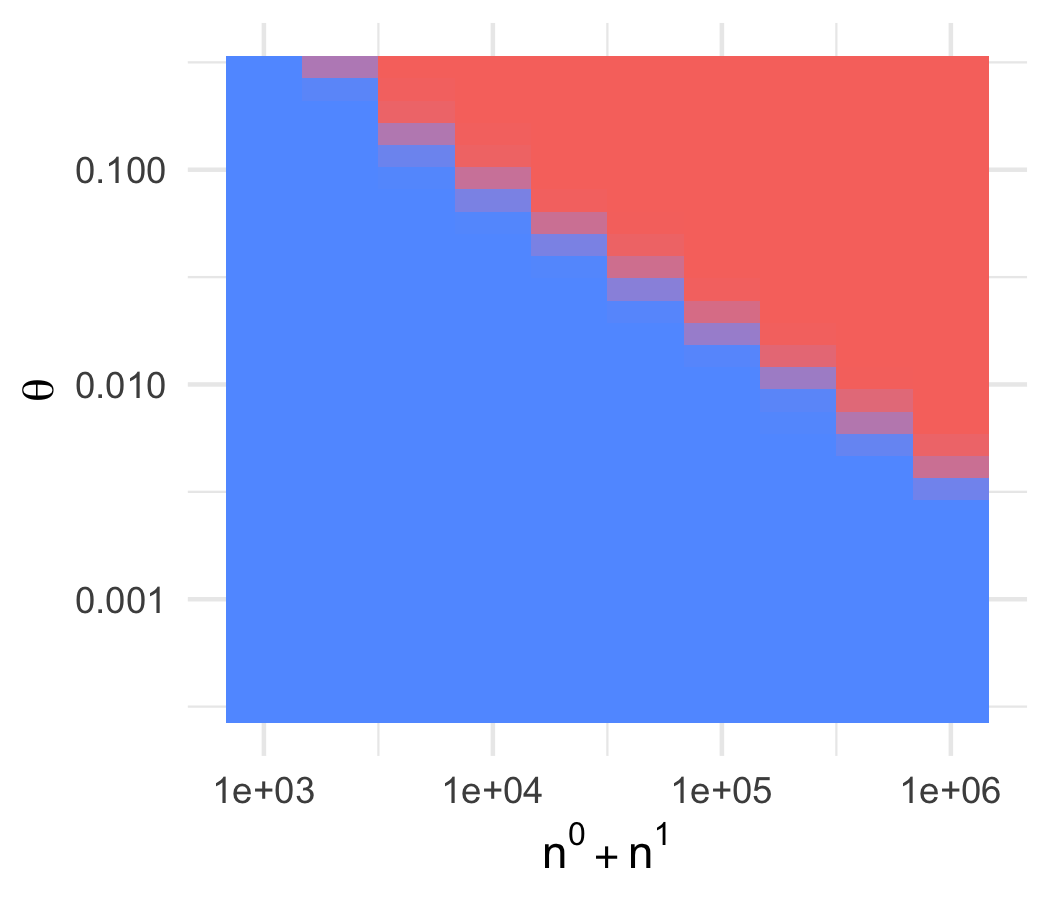}}
\subcaptionbox{\small Setting C. \label{fig:exp9_high}}
{\includegraphics[width=0.34\textwidth]{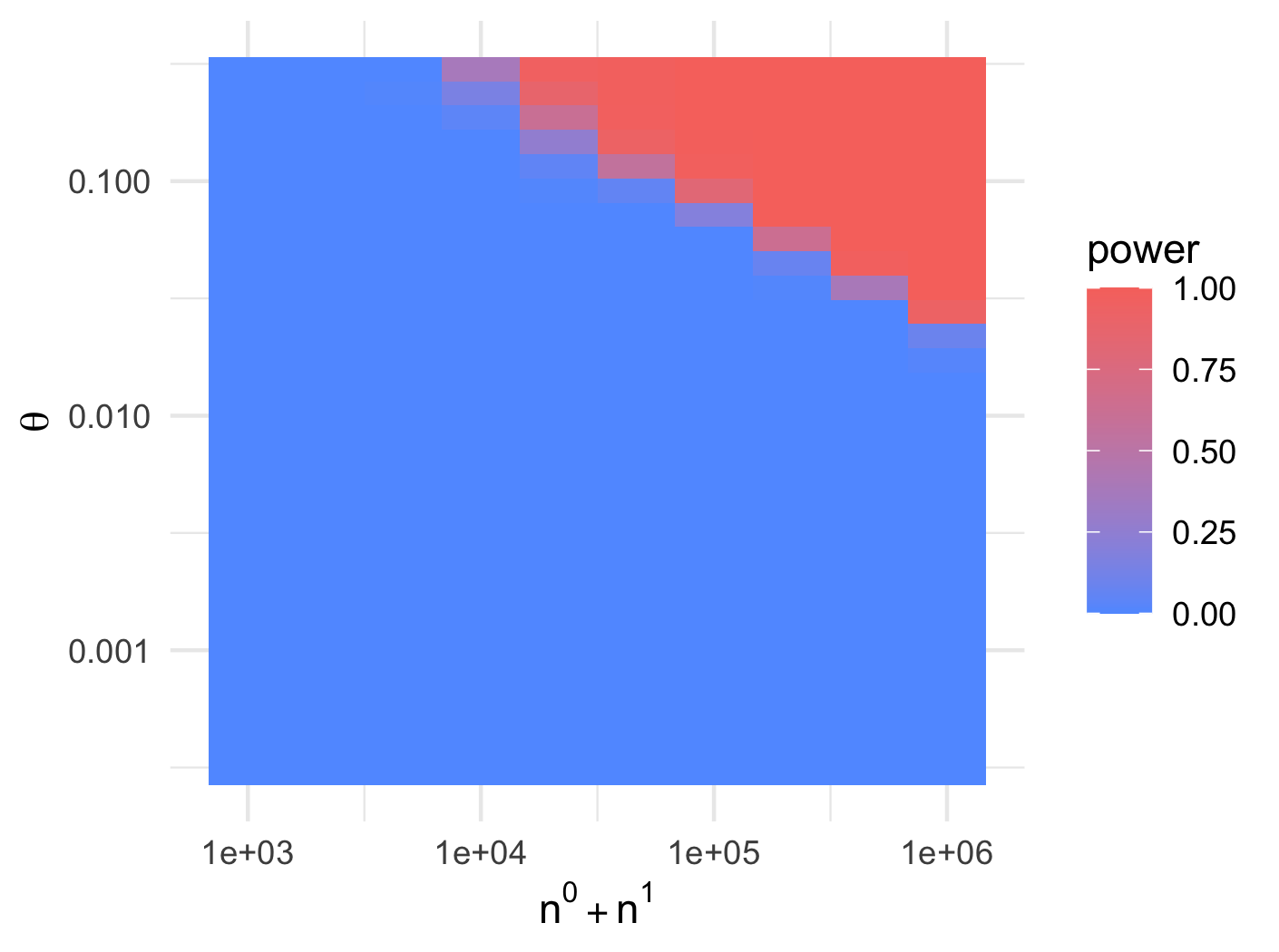}}\\
\caption{Proportion of runs for which the Estimated Density Ratio Test rejected the null hypothesis, for varying values of $\theta$ and of training sample size $n^{train}$. Each experiment is reproduced 100 times.}
\label{fig:reject}
\end{figure}

Figure \ref{fig:reject} presents the power in Settings A, B, and C, for a grid of 30 values of $\theta$. As expected, we observe that the detection threshold tends to $0$ as the sample size increases, as expected, and that is smaller when the signal is strong. Additional experiments in Appendix \ref{app:simu_detection_rate} reveal that the detection rate is of order $n^{-0.97}$ in Setting A, $n^{-0.76}$ in Setting B, and $n^{-0.65}$ in Setting C.

Our experiment also reveals that the Estimated Density Ratio Test has a type I error lower than the desired level $\alpha = 0.05$. In fact, across all tested values of $n^{train}$ and all repetitions carried out for $\theta = 0$, the test never rejected $H_0$. The theoretical bounds on the fluctuations of $\sum_k(\widehat{r}_k - 1)N_k$ under $H_0$, which were used to design the Estimated Density Ratio Test, appear to be overly conservative in practice. To address this issue, we propose in Section \ref{sec:exp_MMD} a new test based on the empirical estimation of this statistic using the bootstrap method.

\subsection{Comparison with two bootstraped tests: Boootstraped Density Ratio and Maximum Mean Riscrepancy} \label{sec:exp_MMD}

\paragraph{Bootstrap Estimated Density Ratio Test} The simulations above suggest that the theoretical bounds on the fluctuations of the statistic $\widehat{S}(\mathbb{X})$ under $H_0$ are overly conservative. In fact, the empirical type I error consistently remains at 0 for all values of the training sample size, whereas a well-calibrated test should exhibit a type I error of approximately 0.05 in these simulations. The excessively conservative nature of the test can result in a reduction in power. To address this issue, we propose a new test that estimates the fluctuations of $\widehat{S}(\mathbb{X})$ using bootstrapping methods. To alleviate the computational burden, the partition $\mathcal{P}$ is chosen once and for all using $\mathbb{X}^0_{part}$ and $\mathbb{X}^1_{part}$. Subsequently, we estimate the distribution of $\widehat{S}(\mathbb{X})$ under $H_0$ on this partition with bootstrap methods, using the samples in $\mathbb{X}^0_{est}$ and $\mathbb{X}^1_{est}$.

More precisely, we partition the training sample $\mathbb{X}^0_{est}$, which is used for the estimation of $h_k^0$, into two samples. We begin by sampling (without replacement) a subset of size $n$ from the training sample $\mathbb{X}^0_{est}$, denoted as $(\mathbb{X}^0_{est})^{I}$. Following this, we sample (with replacement) $n^0$ points from the remaining points of $\mathbb{X}^0_{est}$, and we denote this sample as $(\mathbb{X}^0_{est})^{II}$. Lastly, we sample $n^1$ points from $\mathbb{X}^1_{est}$ (with replacement) and denote this sample as $(\mathbb{X}^1_{est})^{I}$. We estimate the density ratio $\widehat{r}_k$ for the bins $B_k \in \mathcal{P}$ using the samples $(\mathbb{X}^0_{est})^{II}$ and $(\mathbb{X}^1_{est})^{I}$, and we compute the statistic $\widehat{S}((\mathbb{X}^0_{est})^{I})$. This process is repeated numerous times to obtain a threshold $\tau$, corresponding to the empirical quantile of level 0.95 of the statistic  $\widehat{S}((\mathbb{X}^0_{est})^{I})$ for these bootstrapped samples. We then re-estimate the quantities $\widehat{r}_k$ using the full samples $\mathbb{X}^0_{est}$ and $\mathbb{X}^1_{est}$, along with the statistic $\widehat{S}(\mathbb{X})$ corresponding to the true test sample. We reject the hypothesis $H_0$ if $\widehat{S}(\mathbb{X}) > \tau$. This test is hereafter referred to as the Bootstrap Estimated Density Ratio Test.

In this Section, we compare the performance of the Estimated Density Ratio Test with that of the Bootstrap Estimated Density Ratio Test, as well as a classical two-sample test based on Maximum Mean Discrepancy (MMD), due to \cite{gretton}. Before presenting our simulations, we provide a brief reminder on MMD. 

\paragraph{Maximum Mean Discrepancy} Let $(\mathcal{X}, d)$ be a metric space, and $p,q$ be two Borel probability measures defined on $\mathcal{X}$. Consider a class of functions $\mathcal{F}$, where $f: \mathcal{X} \rightarrow \mathbb{R}$. The Maximum Mean Discrepancy between $p$ and $q$ is defined as:
$$\MMD[\mathcal{F}, p, q] = \sup_{f \in \mathcal{F}}\left\{\mathbb{E}_{X \sim p}f(X) - \mathbb{E}_{Y \sim q}f(Y)\right\}.$$
When the function class $\mathcal{F}$ is rich enough, $p=q$ if and only if $\MMD[\mathcal{F}, p, q] = 0$. To compute the MMD of $p$ and $q$, we can rely on the following property. When the class of functions $\mathcal{F}$ is the unit ball in a reproducing kernel Hilbert space $\mathcal{H}$ with kernel $k$, and when $\mathbb{E}_{X \sim p}\sqrt{k(X,X)} < \infty$ and $\mathbb{E}_{X \sim q}\sqrt{k(X,X)} < \infty$, we have:
$$\MMD[\mathcal{F}, p, q] = \mathbb{E}_{X,X' \sim p}\left[k(X,X')\right] + \mathbb{E}_{Y,Y' \sim q}\left[k(Y,Y')\right] - 2 \mathbb{E}_{X \sim p, Y \sim q}\left[k(X,Y)\right].$$
This formula can be used to estimate the MMD between $p$ and $q$, assuming samples $(X_1, ..., X_{n})\sim p$ and $(Y_1, ..., Y_{n'})\sim q$ are available. To do so, we could replace the unknown expectations in the formula with their empirical counterparts. However, estimating the MMD between two samples of size $n$ would require $n^2$ computations, which can be prohibitive in practice. When sample sizes are large, \cite{gretton} suggest approximating the MMD statistic using the following statistic, computed in linear time (assuming, for simplicity, that $n = n'$ is even):
$$\widehat{\MMD}_l = \frac{1}{n/2}\sum_{i \leq n/2}k(X_{2i}, X_{2i+1}) + k(Y_{2i}, Y_{2i+1}) - k(X_{2i}, Y_{2i+1}) - k(X_{2i+1}, Y_{2i}).$$
In practice, its distribution is estimated by bootstrapping, and the hypothesis $p = q$ is rejected if $\widehat{\MMD}_l$ is larger than the corresponding quantile of level $1-\alpha$.

\paragraph{Detection rate of MMD and Bootstrap Estimated Density Ratio Test} We compare the power of MMD with that of the Estimated Density Ratio Test and its bootstrap equivalent. As before, we simulate samples corresponding to Settings A, B, and C. We simulate mixtures sampled from $(1-\theta)\times f^0 + \theta f^1$ for $\theta$ varying between 0.0003 and 0.3. The results are reported in Figure \ref{fig:exp123_power}. Note that the computation time of MMD is much larger than that of the Estimated Density Ratio Test and its bootstraped version, and becomes prohibitive for test sample sizes above $100 000$. For example, MMD is approximately 30 times slower than BEDRT for $n^{train} = 100 000$.
\begin{figure}[h!]
\centering
\subcaptionbox{\small Setting A.\label{fig:exp1_power}}{\includegraphics[width=0.45\textwidth]{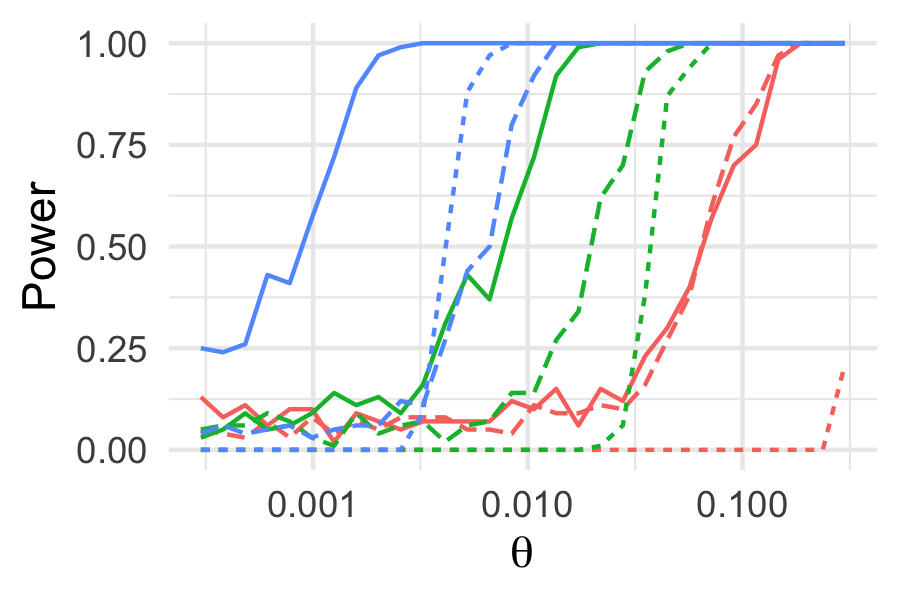}}
\subcaptionbox{\small Setting B.\label{fig:exp2_power}}{\includegraphics[width=0.45\textwidth]{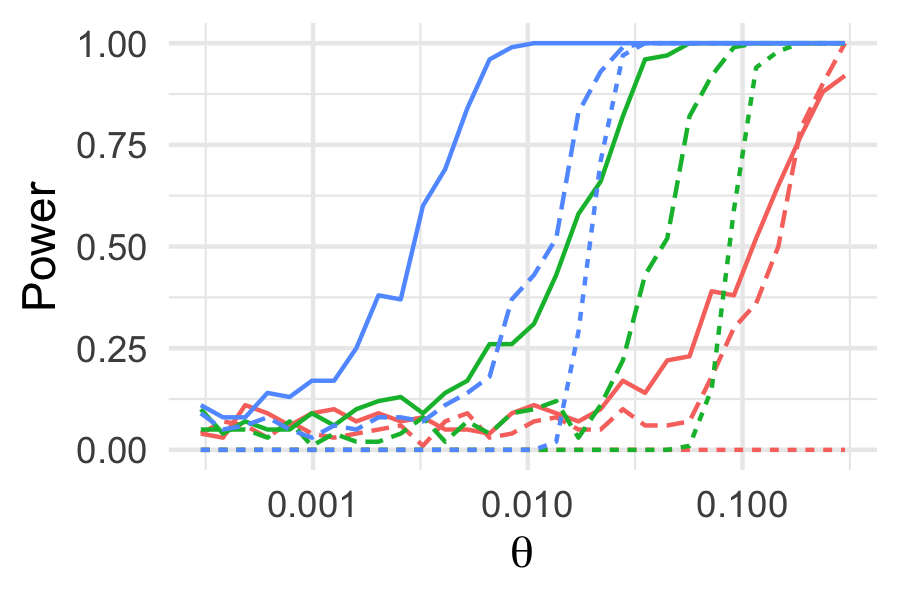}}
\\

\vspace{0.5cm}
\subcaptionbox{\small Setting C. \label{fig:exp3_power}}{\includegraphics[width=0.5\textwidth]{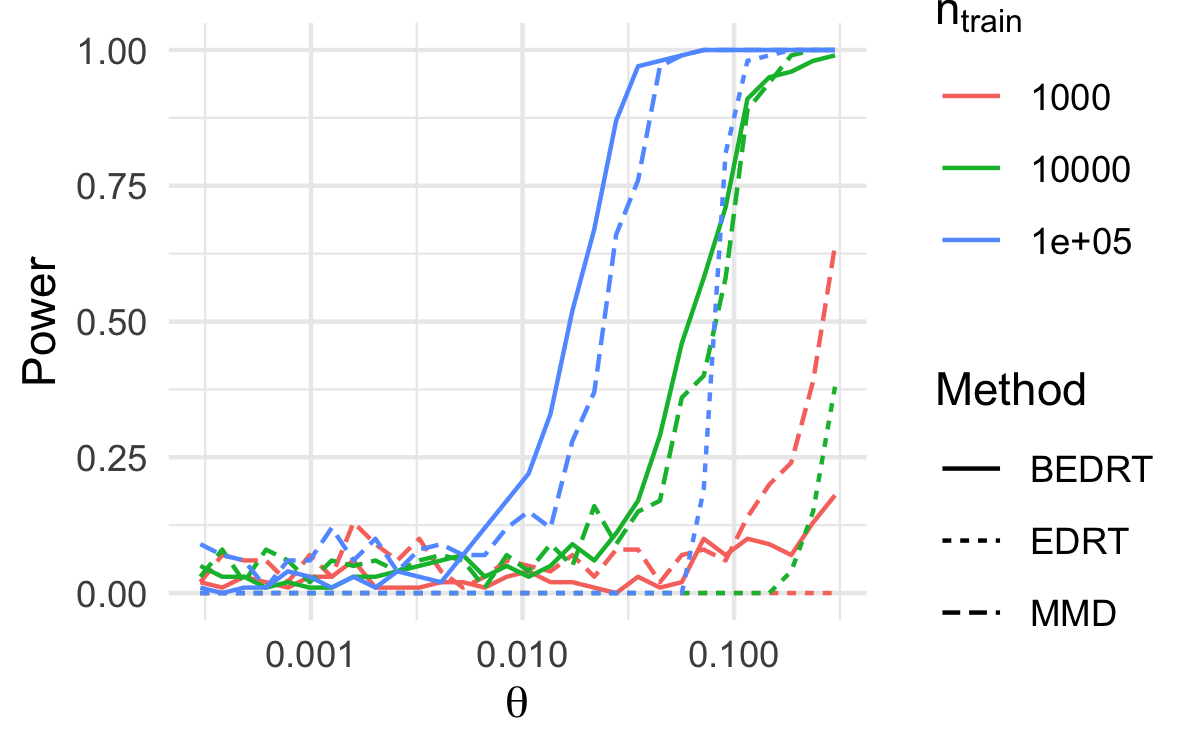}}\\
\caption{Proportion of runs for which the Estimated Density Ratio Test (EDRT), Bootstrap Estimated Density Ratio Test (BEDRT) and MMD test (MMD) rejected the null hypothesis for varying values of $\theta$ and of training sample size $n^{train}$. Each experiment is reproduced 300 times.}
\label{fig:exp123_power}
\end{figure}

Figure \ref{fig:exp123_power} reveals that in settings corresponding to large and intermediate signal strength, the power of the Bootstrap Estimated Density Ratio Test is comparable to that of MMD for small sample sizes, and significantly better as the sample size increases. Additional experiments, presented in the Appendix, confirm this phenomenon by showing that the detection rate for MMD decreases as $n^{-0.5}$, while the detection rate of the Bootstrap Estimated Density Ratio Test approaches $n^{-1}$ in these settings. In contrast, when the sample size is small, and the signal is low, MMD outperforms the Bootstrap Estimated Density Ratio Test. It is important to note that both the Estimated Density Ratio Test and the Bootstrap Estimated Density Ratio Test utilize the sample $\mathbb{X}^1$ to estimate the density $f^1$. When the sample size is small, the information provided by this sample becomes less significant. Moreover, histogram-based tests, such as the Estimated Density Ratio Test and the Bootstrap Estimated Density Ratio Test, retain less information than MMD. These factors may explain the relative loss of performance of these tests for small sample sizes in low signal settings.

\subsection{Application to flow cytometry data}\label{sec:exp_HIPC}
\paragraph{Cytotrol FlowCAP III dataset} We illustrate our test on flow cytometry data by applying it to the T-cell panel of the Cytotrol FlowCAP III dataset provided by the Human Immune Phenotyping Consortium (HIPC), available on the ImmuneSpace website. This dataset consists of 62 samples, representing 3 replicates of analyses conducted on 3 patients across 7 different laboratories (one dataset has been duplicated, so we exclude one of the copies). For each cell, seven biological markers have been measured. These samples have been manually annotated, and the cells have been categorized into the following 10 types: CD4 Effector (CD4 E), CD4 Naive (CD4 N), CD4 Central Memory (CD4 CM), CD4 Effector Memory (CD4 EM), CD4 Activated (CD4 A), CD8 Effector (CD8 E), CD8 Naive (CD8 N), CD8 Central Memory (CD8 CM), CD8 Effector Memory (CD8 EM), and CD8 Activated (CD8 A). Each sample contains a varying number of cells, ranging from 15,554 to 112,318 cells. Additionally, the cell proportions vary significantly, with proportions typically ranging from 1\% to 25\% of the sample. There is notable variability between samples, with standard deviations as high as 94\% of the mean proportion for certain cell populations. 

\paragraph{Experiment} We focus our study on 4 types of cells, each with varying prevalence in the sample. More precisely, we focus on detecting the cells CD4 N, CD8 E, CD8 EM, and CD8 CM whose average prevalence in the sample represent respectively 17\%, 1\%, 11\% and 24\% of the sample cells. We investigate the ability of our test to detect the presence of those cells. For each cell type $c \in \{$CD4 N, CD8 E, CD8 EM, CD8 CM$\}$, we treat cells of type $c$ as being drawn from distribution $f_1$, while the remaining nine cell types are considered drawn from $f_0$. We assess the power of our test against alternatives characterized by different prevalence levels $\theta$. Because of the large computation time required by MMD, we limit the test sample size to 10000 cells. Details of the implementation can be found in the Appendix. The result of this experiment is presented In Figure \ref{fig:exp_HIPC}.

\begin{figure}[h!]
\centering
\includegraphics[width=\textwidth]{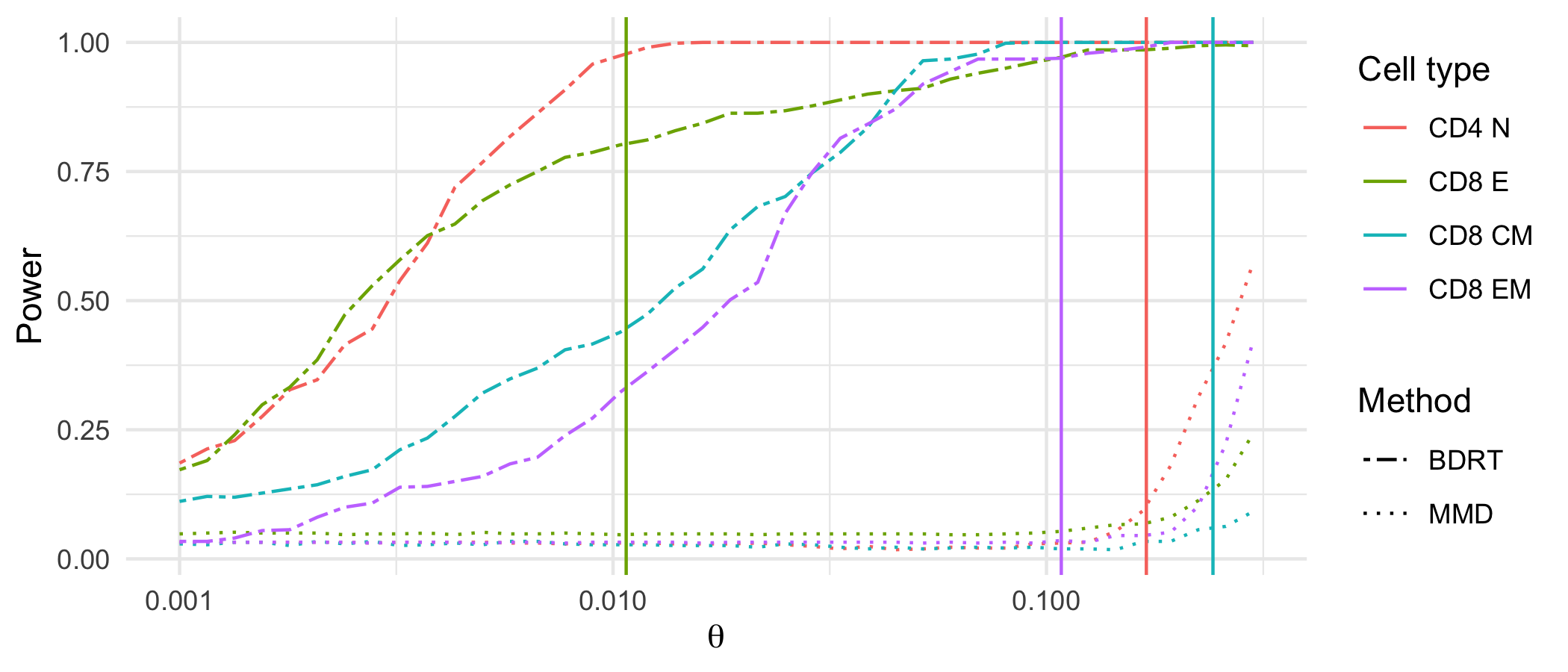}\\
\caption{Detection rates for different types of cells and different prevalence. The full line indicate the true average in-sample cell prevalence. The dashed (resp. dotted) line indicate the power of BEDRT (resp. MMD) against the alternative for varying levels of prevalence $\theta$.}
\label{fig:exp_HIPC}
\end{figure}

\paragraph{Results} The average type I error is approximately $0.040$ for the BEDRT, and $0.038$ for the MMD test. The results illustrate the good performances of BEDRT, demonstrating its ability to accurately identify cells such as ``CD4 N," ``CD8 CM", and ``CD8 EM" with probabilities nearly reaching 1 at their actual prevalence levels. Additionally, BEDRT successfully detects the rare cell ``CD8 E" with a probability exceeding 0.75 at its true prevalence level. In stark contrast, MMD consistently falls short, achieving low power for the specified prevalence levels. Furthermore, the substantial computational time demanded by MMD prohibits the analysis of large samples. In contrast, BEDRT, being significantly faster, can efficiently analyze larger datasets.
\section{Conclusions}
In this paper, we study a problem of supervised contamination detection. We design a test for this problem, and we establish non-asymptotic bounds on its type I error, and on its power. The efficacy of this test depends heavily on the judicious selection of a partition for estimating the densities of inliers and outliers. We introduce a partitioning algorithm designed to efficiently capture the signal for the test. Both the resultant test and partitioning algorithms yield results that are easily interpretable.

Implementation of this test is straightforward using available packages for Classification And Regression Trees. We demonstrate the good performance of the test through extensive simulations. These simulations indicate that the test tends to be excessively conservative. To address this, we propose a calibration method employing a bootstrapping approach. Simulations and experiments conducted on flow cytometry data showcase the superior performance of the bootstrapped test, surpassing benchmark tests for two-sample testing by a significant margin.

\section*{Acknowledgments and Disclosure of Funding} 
The authors thank Bastien Dussap for fruitful discussions. G. B. gratefully acknowledges funding from the grants ANR-21-CE23-0035 (ASCAI), ANR-19-CHIA-0021-01 (BISCOTTE) and ANR-23-CE40-0018-01 (BACKUP) of the French National Research Agency ANR. Part of this research was done when G. B. received support from DFG SFB1294 ``Data Assimilation'', as ``Mercator fellow'' at Universit{\"a}t Potsdam, Germany. F.C. and S.G. thank the ANR TopAI chair (ANR–19–CHIA–0001) for financial support. S.G. gratefully acknowledges funding from the Hadamard Doctoral School of Mathematics (EDMH).

\bibliography{ref_cyto}

\newcommand{\etalchar}[1]{$^{#1}$}
\begin{thebibliography}{VPACLG{\etalchar{+}}21}

\bibitem[Agg17]{Aggarwal2017}
Charu~C. Aggarwal.
\newblock {\em An Introduction to Outlier Analysis}, pages 1--34.
\newblock Springer International Publishing, Cham, 2017.

\bibitem[AKS20]{10.5555/3524938.3524960}
Amr~M. Alexandari, Anshul Kundaje, and Avanti Shrikumar.
\newblock Maximum likelihood with bias-corrected calibration is hard-to-beat at
  label shift adaptation.
\newblock In {\em Proceedings of the 37th International Conference on Machine
  Learning}, ICML'20. JMLR.org, 2020.

\bibitem[BLS10]{BlanchardNovelty}
Gilles Blanchard, Gyemin Lee, and Clayton Scott.
\newblock Semi-supervised novelty detection.
\newblock {\em Journal of Machine Learning Research}, 11:2973--3009, 03 2010.

\bibitem[BV10]{Bordes10AOS}
L.~Bordes and P.~Vandekerkhove.
\newblock Semiparametric two-component mixture model with a known component: An
  asymptotically normal estimator.
\newblock {\em Mathematical Methods of Statistics}, 19(1):22--41, 2010.

\bibitem[CJ10]{CAI10AOS}
T.~Tony Cai and Jiashun Jin.
\newblock {Optimal rates of convergence for estimating the null density and
  proportion of nonnull effects in large-scale multiple testing}.
\newblock {\em The Annals of Statistics}, 38(1):100 -- 145, 2010.

\bibitem[CR10]{Celisse10JSIP}
Alain Celisse and Stéphane Robin.
\newblock A cross-validation based estimation of the proportion of true null
  hypotheses.
\newblock {\em Journal of Statistical Planning and Inference},
  140(11):3132--3147, 2010.

\bibitem[DBCA23]{Dussap}
Bastien Dussap, Gilles Blanchard, and Badr-Eddine Ch\'{e}rief-Abdellatif.
\newblock Label shift quantification with robustness guarantees via
  distribution feature matching.
\newblock In {\em Machine Learning and Knowledge Discovery in Databases:
  Research Track: European Conference, ECML PKDD 2023, Turin, Italy, September
  18–22, 2023, Proceedings, Part V}, page 69–85, Berlin, Heidelberg, 2023.
  Springer-Verlag.

\bibitem[End87]{Enderlein1987HawkinsDM}
G.~Enderlein.
\newblock Hawkins, d. m.: Identification of outliers. chapman and hall, london
  – new york 1980, 188 s., £ 14, 50.
\newblock {\em Biometrical Journal}, 29:198--198, 1987.

\bibitem[EYN06]{6287347}
Ran El-Yaniv and Mordechai Nisenson.
\newblock Optimal single-class classification strategies.
\newblock In B.~Sch\"{o}lkopf, J.~Platt, and T.~Hoffman, editors, {\em Advances
  in Neural Information Processing Systems}, volume~19. MIT Press, 2006.

\bibitem[For08]{count_class}
George Forman.
\newblock Quantifying counts and costs via classification.
\newblock {\em Data Mining and Knowledge Discovery}, 17(2):164--206, 2008.

\bibitem[GBR{\etalchar{+}}12]{gretton}
Arthur Gretton, Karsten~M. Borgwardt, Malte~J. Rasch, Bernhard Sch{{\"o}}lkopf,
  and Alexander Smola.
\newblock A kernel two-sample test.
\newblock {\em Journal of Machine Learning Research}, 13(25):723--773, 2012.

\bibitem[GCARA13]{GONZALEZCASTRO2013146}
Víctor González-Castro, Rocío Alaiz-Rodríguez, and Enrique Alegre.
\newblock Class distribution estimation based on the hellinger distance.
\newblock {\em Information Sciences}, 218:146--164, 2013.

\bibitem[GCnCC17]{10.1145/3117807}
Pablo Gonz\'{a}lez, Alberto Casta\~{n}o, Nitesh~V. Chawla, and Juan
  Jos\'{e}~Del Coz.
\newblock A review on quantification learning.
\newblock {\em ACM Comput. Surv.}, 50(5), sep 2017.

\bibitem[GN15]{giné_nickl_2015}
Evarist Giné and Richard Nickl.
\newblock {\em Mathematical Foundations of Infinite-Dimensional Statistical
  Models}.
\newblock Cambridge Series in Statistical and Probabilistic Mathematics.
  Cambridge University Press, 2015.

\bibitem[GWBL20]{10.5555/3495724.3496001}
Saurabh Garg, Yifan Wu, Sivaraman Balakrishnan, and Zachary~C. Lipton.
\newblock A unified view of label shift estimation.
\newblock In {\em Proceedings of the 34th International Conference on Neural
  Information Processing Systems}, NIPS'20, Red Hook, NY, USA, 2020. Curran
  Associates Inc.

\bibitem[INS14]{pmlr-v32-iyer14}
Arun Iyer, Saketha Nath, and Sunita Sarawagi.
\newblock Maximum mean discrepancy for class ratio estimation: Convergence
  bounds and kernel selection.
\newblock In Eric~P. Xing and Tony Jebara, editors, {\em Proceedings of the
  31st International Conference on Machine Learning}, volume~32 of {\em
  Proceedings of Machine Learning Research}, pages 530--538, Bejing, China,
  22--24 Jun 2014. PMLR.

\bibitem[LSW{\etalchar{+}}19]{cytoComp}
Xiao Liu, Weichen Song, Brandon~Y. Wong, Ting Zhang, Shunying Yu, Guan~Ning
  Lin, and Xianting Ding.
\newblock A comparison framework and guideline of clustering methods for mass
  cytometry data.
\newblock {\em Genome Biology}, 20(1):297, 2019.

\bibitem[LWS18]{DBLP:conf/icml/LiptonWS18}
Zachary~C. Lipton, Yu{-}Xiang Wang, and Alexander~J. Smola.
\newblock Detecting and correcting for label shift with black box predictors.
\newblock In Jennifer~G. Dy and Andreas Krause, editors, {\em Proceedings of
  the 35th International Conference on Machine Learning, {ICML} 2018,
  Stockholmsm{\"{a}}ssan, Stockholm, Sweden, July 10-15, 2018}, volume~80 of
  {\em Proceedings of Machine Learning Research}, pages 3128--3136. {PMLR},
  2018.

\bibitem[MPSV24]{Milhaud24}
Xavier Milhaud, Denys Pommeret, Yahia Salhi, and Pierre Vandekerkhove.
\newblock {Two-sample contamination model test}.
\newblock {\em Bernoulli}, 30(1):170 -- 197, 2024.

\bibitem[NM14]{Nguyen14}
Van~Hanh Nguyen and Catherine Matias.
\newblock On efficient estimators of the proportion of true null hypotheses in
  a multiple testing setup.
\newblock {\em Scandinavian Journal of Statistics}, 41(4):1167--1194, 2014.

\bibitem[{R C}21]{R}
{R Core Team}.
\newblock {\em R: A Language and Environment for Statistical Computing}.
\newblock R Foundation for Statistical Computing, Vienna, Austria, 2021.

\bibitem[TA22]{rpart}
Terry Therneau and Beth Atkinson.
\newblock {\em rpart: Recursive Partitioning and Regression Trees}, 2022.
\newblock R package version 4.1.19.

\bibitem[TPB{\etalchar{+}}19]{ctc}
Xuefei Tan, Roshani Patil, Peter Bartosik, Judith~M. Runnels, Charles~P. Lin,
  and Mark Niedre.
\newblock In vivo flow cytometry of extremely rare circulating cells.
\newblock {\em Scientific Reports}, 9(1):3366, 2019.

\bibitem[Vaa98]{vaart_1998}
A.~W. van~der Vaart.
\newblock {\em Asymptotic Statistics}.
\newblock Cambridge Series in Statistical and Probabilistic Mathematics.
  Cambridge University Press, 1998.

\bibitem[VKV{\etalchar{+}}20]{ijms21072323}
Denis~V. Voronin, Anastasiia~A. Kozlova, Roman~A. Verkhovskii, Alexey~V.
  Ermakov, Mikhail~A. Makarkin, Olga~A. Inozemtseva, and Daniil~N. Bratashov.
\newblock Detection of rare objects by flow cytometry: Imaging, cell sorting,
  and deep learning approaches.
\newblock {\em International Journal of Molecular Sciences}, 21(7), 2020.

\bibitem[VPACLG{\etalchar{+}}21]{VILLAPEREZ2021106878}
Miryam~Elizabeth Villa-Pérez, Miguel~A. Alvarez-Carmona, Octavio
  Loyola-González, Miguel~Angel Medina-Pérez, Juan~Carlos Velazco-Rossell,
  and Kim-Kwang~Raymond Choo.
\newblock Semi-supervised anomaly detection algorithms: A comparative summary
  and future research directions.
\newblock {\em Knowledge-Based Systems}, 218:106878, 2021.

\bibitem[VV06]{JMLR:v7:vert06a}
R{{\'e}}gis Vert and Jean-Philippe Vert.
\newblock Consistency and convergence rates of one-class svms and related
  algorithms.
\newblock {\em Journal of Machine Learning Research}, 7(29):817--854, 2006.

\end{thebibliography}
\bibliographystyle{alpha}

\newpage
\appendix
\section{Proofs}

\subsection{Proof of Theorem \ref{thm:type_I_II}} \label{app:proof_main_th}

To make the proof easier to read, we divide Theorem \ref{thm:type_I_II} into two results, concerning the type I error and type II error respectively. These results are then proved in sections \ref{subsubsec:proof_lem_type_I} and \ref{subsubsec:proof_lem_type_II} respectively. Then, we characterize alternatives corresponding to power $1-\alpha$.

\begin{lemma}\label{lem:type_I}
Assume that $n^0 \geq n^1$, and that $3\epsilon^0\leq \epsilon^1\leq 1$.    Then the Density Ratio Test has type I error lower than $1-\alpha$.
\end{lemma}

The next lemma establishes a lower bound in $\widehat{S}(\mathbb{X})$ that holds with high probability under the alternative.
\begin{lemma}\label{lem:type_II}
Assume that $n^0 \geq n^1$, and that $3\epsilon^0\leq \epsilon^1\leq 1$.    Under Hypothesis $H_1$, with probability larger than $1-\alpha$,
\begin{align*}
    \widehat{S}(\mathbb{X}) > &  \left(\left(1 - \left(\frac{u}{n^1}\right)^{\frac{1}{4}}\right)\widehat{\sigma}^2 - \left(\frac{u}{n^1}\right)^{\frac{1}{4}}\right)\theta - \left(\sqrt{\frac{10uK\widehat{\sigma}^2}{n^0}} + \left(13\epsilon^1\vert \Omega^1\vert  + 4\epsilon^1\vert \Omega^{01}\vert  + 3\epsilon^0 \vert\Omega^0\vert\right)\right)\\
    & - \sqrt{\widehat{\sigma}^2}\sqrt{\frac{6t}{n}\left(1 + \theta + \theta \max_{k}\left \vert \widehat{r}_k - 1\right \vert\right)}  - \frac{t}{3n}\max_k\vert \widehat{r}_k - 1\vert.
\end{align*}
\end{lemma}

Finally, next lemma characterizes alternative such that the test is rejected with probablity larger than $1-\alpha$.
\begin{lemma}\label{lem:power}
Assume that $n^0 \geq n^1$, and that $3\epsilon^0\leq \epsilon^1\leq 1$. Under the alternative hypothesis $H_1$, the Density Ratio Test rejects the hypothesis $H_0$ with probability larger than $1-\alpha$ if $\theta$ verifies
\begin{align*}
\theta \geq& \left(353\sqrt{\frac{t}{n\widehat{\sigma}^2}} + \frac{400\sqrt{u}K^1}{\widehat{\sigma}^2\sqrt{n^1}} + \frac{30\epsilon^0K^0}{\widehat{\sigma}^2}+64\sqrt{\frac{uK}{n^0\widehat{\sigma}^2}}\right).
\end{align*}
where $K^1 = \vert \Omega^1\vert + \vert \Omega^{0}\vert$ and $K^0 = \vert \Omega^0\vert$.
\end{lemma}

\subsubsection{Concentration results on the estimated probabilities $\widehat{h}_k^0$ and $\widehat{h}_k^1$}
 Before proving Lemmas \ref{lem:type_I}, \ref{lem:type_II}, and \ref{lem:power}, we provide a list of concentration results relating the estimators $\widehat{h}^0_k$ and $\widehat{h}^1_k$ to their population counterpart $h_k^0$ and $h_k^1$. The concentration results bellow rely on Bernstein's theorem, which we recall here for completeness:

\begin{theorem}[Bernstein inequality, see, e.g.,  Theorem 3.1.7 in 
\cite{giné_nickl_2015}] \label{thm:Bernstein}
   Let $X_i$, $1 \leq i \leq n$ be independent centered random variables a.s. bounded by $c < \infty$ in absolute value. Set $\sigma^2 = 1/n \sum_{i\leq n}\mathbb{E}\left[X_i^2\right]$ and $S_n = \sum_{i\leq n}X_i$. Then, for all $u >0$,
   $$ \mathbb{P}\left(S_n \geq \sqrt{2n\sigma^2u} + \frac{cu}{3} \right)\leq e^{-u}.$$
\end{theorem}

Lemma \ref{lem:concentration_histogram} shows that on an event $\mathcal{E}$ of large probability, $\widehat{h}^0_k$ and $\widehat{h}^1_k$ are close to $h_k^0$ and $h_k^1$ simultaneously for all $k$.
\begin{lemma}\label{lem:concentration_histogram}
    On an event $\mathcal{E}$ of probability at least $1 - 4Ke^{-u}$, simultaneously for all $k \leq K$ and $a \in \{0,1 \}$, one has
    \begin{align}
    \frac{N^a_k}{n^a} & \in \left [h^a_k - \sqrt{2\frac{uh^a_k}{n^a}} - \frac{u}{3n^a}; h^a_k + \sqrt{2\frac{uh^a_k}{n^a}} + \frac{u}{3n^a}\right] \label{eq:bound_freq}\\
    h^a_k & \in \left [\frac{N_k^a}{n^a} - \frac{\sqrt{2uN_k^a}}{n^a};\frac{N_k^a}{n^a} + \frac{\sqrt{2uN_k^a}}{n^a}+ \frac{3u}{n^a}\right].\label{eq:bound_est}
\end{align}
\end{lemma}
The proof of Lemma \ref{lem:concentration_histogram} is postponed to Appendix \ref{app:proof_concentration_histogram}.

\bigskip

The following lemma provides sharp bounds on the distance $\left\vert h_k^a - \widehat{h}_k^a \right\vert$ for $a = 0$ and $a=1$, for bins where the density have been thresholded and for bins in which they have not.
\begin{lemma}\label{lem:error_bound}
Assume that $n^0 \geq n^1$ and that $\epsilon^0, \epsilon^1\leq 1$. On the event $\mathcal{E}$ of Lemma \ref{lem:concentration_histogram}, for $a \in \{0, 1\}$ and for all $k \in \Omega^{a}\cup \Omega^{01}$,
\begin{align}\label{eq:bound_h_Omega}
    h_k^a \leq \widehat{h}_k^a.
\end{align}
Moreover, for all $k \notin \Omega^{a}\cup \Omega^{01}$,
\begin{align}
\left\vert h_k^a - \widehat{h}_k^a \right\vert \leq \sqrt{\frac{10u\widehat{h}_k^a}{n^a}}\label{eq:bound_h}.
\end{align}
and
\begin{align}
    \left\vert h_k^a - \widehat{h}_k^a \right\vert \leq 2\widehat{h}_k^a\label{eq:bound_h_abs}.
\end{align}
\end{lemma}
The proof of Lemma \ref{lem:error_bound} is postponed to Appendix \ref{app:proof_error_bound}.

\bigskip
Finally, the following Lemma provides a lower bound on the signal term $\sum_k \left(\widehat{r}_k - 1\right)h_k^1$ depending on the estimated signal $\widehat{\sigma}^2$.
\begin{lemma}\label{lem:lower_bound_B1}
    Assume that $n^0 \geq n^1$ and that $\epsilon^0, \epsilon^1\leq 1$. On the event $\mathcal{E}$ of Lemma \ref{lem:concentration_histogram}, 
    $$\sum_k \left(\widehat{r}_k - 1\right)h_k^1 \geq \left(1 - \left(\frac{u}{n^1}\right)^{\frac{1}{4}}\right)\widehat{\sigma}^2 -  \left(\frac{u}{n^1}\right)^{\frac{1}{4}} - \left(\epsilon^1\left(13\vert \Omega^1\vert + 4\vert \Omega^{01}\vert\right) + 3\epsilon^0\vert \Omega^0\vert\right).$$
\end{lemma}
The proof of Lemma \ref{lem:lower_bound_B1} is postponed to Appendix \ref{app:proof_lower_bound_B1}.


\subsubsection{Proof of Lemma \ref{lem:type_I}} \label{subsubsec:proof_lem_type_I}
    We apply Bernstein's inequality \eqref{thm:Bernstein} conditionally on the samples $\mathbb{X}^0$ and $\mathbb{X}^1$ to the random variables $A_i = \sum_{k\leq K}\left(\widehat{r}_k - 1\right)\left(\mathds{1}\{X_i \in B_k\} - h_k^0\right)$. Note that the random variables $A_i$ are centered, that $A_i \leq \max_{k \leq K}\vert \widehat{r}_k - 1\vert$ a.s., and that $\frac{1}{n}\sum_{i\leq n}\mathbb{E}[A_i^2] \leq \sum_{k \leq K}\left(\widehat{r}_k - 1\right)^2h_k^0$. Then, under $H_0$, for all $t>0$, 
\begin{align}\label{eq:Bernstein_H0}
    \mathbb{P}^{H_0}_{\vert \mathbb{X}^0, \mathbb{X}^1}\left( \widehat{S}(\mathbb{X}) \geq  \sum_{k \leq K}\left(\widehat{r}_k - 1\right)h_k^0 + \sqrt{\frac{2t}{n}\sum_{k \leq K}\left(\widehat{r}_k - 1\right)^2h_k^0} + \frac{t}{3n}\max_{k\leq K}\vert \widehat{r}_k - 1\vert \right)\leq& e^{-t}.
\end{align}
To prove Lemma \ref{lem:type_I}, we provide upper bounds on the bias 
$$B_0 = \sum_{k \leq K}\left(\widehat{r}_k - 1\right) h_k^0,$$
and on the variance 
$$V_0 = \sum_{k \leq K}\left(\widehat{r}_k - 1\right)^2h_k^0.$$

\paragraph{Bound on the bias $\mathbf{B_0}$}
We begin by bounding the bias term. We decompose the $B_0$ as follows:
\begin{align*}
    B_0 &= \sum_{k \leq K}\left(\widehat{r}_k - 1\right)\widehat{h}_k^0 + \sum_{k \leq K}\left(\widehat{r}_k - 1\right)\left(h_k^0 - \widehat{h}_k^0 \right).
\end{align*}
Using the definition of $\widehat{r}_k = \frac{\widehat{h}_k^1}{\widehat{h}_k^0}$ yields
\begin{align*}
    B_0 &= \sum_{k \leq K}\widehat{h}_k^1 - \sum_{k \leq K}\widehat{h}_k^0 + \sum_{k \leq K}\left( \widehat{r}_k - 1\right) \left( h_k^0 - \widehat{h}_k^0 \right).
\end{align*}
By definition of $\widehat{h}_k^1$, $\sum_{k \leq K}\widehat{h}_k^1 \leq 1 + 3\epsilon^1 \left(\vert \Omega^1\vert + \vert \Omega^{01}\vert\right)$, and $\sum_{k \leq K}\widehat{h}_k^0 \geq 1$. Thus, 
\begin{align*}
    B_0 &\leq  3\epsilon^1 \left(\vert \Omega^1\vert + \vert \Omega^{01}\vert\right) + \sum_{k \leq K}\left( \widehat{r}_k - 1\right) \left( h_k^0 - \widehat{h}_k^0 \right).
\end{align*}
To bound $\sum_{k \leq K}\left( \widehat{r}_k - 1\right) \left( h_k^0 - \widehat{h}_k^0 \right)$, we consider separately the cases where $k \in \Omega^{0}\cup \Omega^{01}$ and $k \notin \Omega^{0}\cup \Omega^{01}$. On the one hand, for all $k \in \Omega^{01}$, $\widehat{h}^0_k = 3\epsilon^1 = \widehat{h}^1_k$, so $\widehat{r}_k=1$ and
\begin{align*}
    \sum_{k \in \Omega^{01}}\left( \widehat{r}_k - 1\right)\left(h_k^0 - \widehat{h}_k^0 \right)= 0.
\end{align*}
Similarly, for all $k \in \Omega^{0}$, $\widehat{h}^0_k = 3\epsilon^0$. Moreover, $\widehat{h}^1_k > \epsilon^1 \geq 3\epsilon^0$. Thus, for all $k \in \Omega^{0}$, we see that $\widehat{r}_k \geq 1$. Moreover, Equation \eqref{eq:bound_h_Omega} in Lemma \ref{lem:error_bound} shows that for $k \in \Omega^{0}\cup\Omega^{01}$, $h_k^0 \leq \widehat{h}_k^0$. Thus,
\begin{align*}
    \sum_{k \in \Omega^{0}}\left( \widehat{r}_k - 1\right)\left(h_k^0 - \widehat{h}_k^0 \right)\leq 0.
\end{align*}
It remains to bound $\sum_{k\notin \Omega^0 \cup \Omega^{01}}\left( \widehat{r}_k - 1\right)\left( h_k^0 - \widehat{h}_k^0 \right)$. Equation \eqref{eq:bound_h} in Lemma \ref{lem:error_bound} implies that on the event $\mathcal{E}$,
\begin{align*}
    \sum_{k\notin \Omega^0 \cup \Omega^{01}}\left\vert \widehat{r}_k - 1\right\vert \left\vert h_k^0 - \widehat{h}_k^0 \right\vert &\leq \sqrt{\frac{10u}{n^0}}\sum_{k\notin \Omega^0}\sqrt{\left(\widehat{r}_k - 1\right)^2\widehat{h}_k^0}\\
    &\leq \sqrt{\frac{10u}{n^0}}\sum_{k\leq K}\sqrt{\left(\widehat{r}_k - 1\right)^2\widehat{h}_k^0}
\end{align*}
Now, Jensen's inequality implies
\begin{align*}
    \sum_{k\leq K}\sqrt{\left( \widehat{r}_k - 1\right)^2\widehat{h}_k^0} &\leq \sqrt{K\widehat{\sigma}^2}.
\end{align*}
Thus,
\begin{align}\label{eq:bound_B0}
    B_0 &\leq  3\epsilon^1 \left(\vert \Omega^1\vert + \vert \Omega^{01}\vert\right) + \sqrt{\frac{10uK\widehat{\sigma}^2}{n^0}}.
\end{align}

\paragraph{Bound on $\mathbf{V_0}$}
Next, we control the variance term under $H_0$, defined as 
$$V_0 = \sum_{k \leq K}\left(\widehat{r}_k - 1\right)^2h_k^0.$$
On the event $\mathcal{E}$, Lemma \ref{lem:error_bound} implies

\begin{align}\label{eq:bound_V0}
V_0 &\leq 3\sum_{k \leq K}\left(\widehat{r}_k - 1\right)^2\widehat{h}_k^0 = 3\widehat{\sigma}^2.
\end{align}

\paragraph{Conclusion}
Combing Equations \eqref{eq:Bernstein_H0}, \eqref{eq:bound_B0}, and \eqref{eq:bound_B0}, we find that under $H_0$, on the event $\mathcal{E}$, with probability $1 - e^{-t}$,
\begin{align*}
    \widehat{S}(\mathbb{X}) < &  3\epsilon^1 \left(\vert \Omega^1\vert + \vert \Omega^{01}\vert\right)  + \sqrt{\frac{10uK\widehat{\sigma}^2}{n^0}} + \sqrt{\frac{6t\widehat{\sigma}^2}{n}} + \frac{t}{3n}\max_{k \leq K}\vert \widehat{r}_k - 1\vert.
\end{align*}
Note that for $u = \log(\frac{4K}{\alpha})$, the event $\mathcal{E}$ holds with probability $1-\frac{\alpha}{2}$. Choosing $t = \log(\frac{2}{\alpha})$, we see that the test described in Lemma has type one error lower than $\alpha$.


\subsubsection{Proof of Lemma \ref{lem:type_II}}\label{subsubsec:proof_lem_type_II}

Recall that under $H_1$, the sample $X$ has a density given by $\theta \times f_0 + (1-\theta) \times f_1$. Applying Bernstein's inequality conditionally on the samples $\mathbb{X}^0$ and $\mathbb{X}^1$, we see that under $H_1$, for all $t>0$, 
\begin{align}\label{eq:Bernstein_H1}
    \mathbb{P}^{H_1}_{\vert \mathbb{X}^0, \mathbb{X}^1}\Bigg(& \widehat{S}(\mathbb{X}) \leq  (1-\theta)\sum_{k \leq K}\left(\widehat{r}_k - 1\right)h_k^0 + \theta \sum_{k \leq K}\left(\widehat{r}_k - 1\right)h_k^1 \nonumber\\
    &- \sqrt{\frac{2t}{n}\left(\sum_{k \leq K}\left(\widehat{r}_k - 1\right)^2h_k^0 + \theta\sum_{k \leq K}\left(\widehat{r}_k - 1\right)^2(h_k^1 - h_k^0)\right)} - \frac{t}{3n}\max_{k\leq K}\vert \widehat{r}_k - 1\vert \Bigg)\leq& e^{-t}.
\end{align}
To prove Lemma \ref{lem:type_II}, we provide lower bounds on the bias under $H_1$
$$B_1 = (1-\theta)\sum_{k \leq K}\left(\widehat{r}_k - 1\right)h_k^0 + \theta \sum_{k \leq K}\left(\widehat{r}_k - 1\right)h_k^1,$$
as well as upper bounds on the variance
$$V_1 = \sum_{k \leq K}\left(\widehat{r}_k - 1\right)^2h_k^0 + \theta\sum_{k \leq K}\left(\widehat{r}_k - 1\right)^2(h_k^1 - h_k^0).$$

\paragraph{Bound on $\mathbf{B_1}$} We begin by providing a lower bound on
$$B_1 = (1-\theta)\sum_{k \leq K}\left(\widehat{r}_k - 1\right)h_k^0 + \theta \sum_{k \leq K}\left(\widehat{r}_k - 1\right)h_k^1.$$
To do so, we first provide a lower bound on $\sum_{k \leq K}\left(\widehat{r}_k - 1\right)h_k^0$. Following the proof of Lemma \ref{lem:type_I}, we note that for all $k \in \Omega^0 \cup \Omega^{01}$, $\widehat{r}_k \geq 1$. Thus, 
\begin{align*}
    \sum_{k}\left(\widehat{r}_k - 1\right)h_k^0 & 
    \geq \sum_{k\notin \Omega^0\cup  \Omega^{01}}\left(\widehat{r}_k - 1\right)h_k^0\\
    & \geq \sum_{k\notin \Omega^0\cup  \Omega^{01}}\widehat{h}_k^0 - \widehat{h}_k^1
    - \sum_{k\notin \Omega^0\cup  \Omega^{01}}\left\vert\widehat{r}_k - 1\right\vert \left\vert h_k^0 - \widehat{h}_k^0\right\vert.
\end{align*}
Now, $\sum_{k\notin \Omega^0\cup  \Omega^{01}}\widehat{h}_k^0 \geq 1 - \left(\epsilon^0\vert \Omega^0 \vert +\epsilon^1\vert \Omega^{01}\vert\right)$, and  $\sum_{k\notin \Omega^0\cup  \Omega^{01}}\widehat{h}_k^1 \leq 1 + 3\epsilon^1\vert \Omega^1 \vert$. Thus, 
\begin{align*}
    \sum_{k}\left(\widehat{r}_k - 1\right)h_k^0 & \geq -  \left(3\epsilon^{1}\vert \Omega^1 \vert + \epsilon^{0}\vert \Omega^0 \vert + \epsilon^{1}\vert \Omega^{01}\vert\right)
    - \sum_{k \notin \Omega^0\cup\Omega^{01}}\left\vert\widehat{r}_k - 1\right\vert \left\vert h_k^0 - \widehat{h}_k^0\right\vert.
\end{align*}
Following the proof of Lemma \ref{lem:type_I}, we find that on the event $\mathcal{E}$,
\begin{align*}
    \sum_{k \notin \Omega^0\cup \Omega^{01}}\left\vert\widehat{r}_k - 1\right\vert \left\vert h_k^0 - \widehat{h}_k^0\right\vert &\leq \sqrt{\frac{10u}{n^0}}\sum_{k\notin \Omega^0\cup \Omega^{01}}\sqrt{\left(\widehat{r}_k - 1\right)^2\widehat{h}_k^0}\\ &\leq \sqrt{\frac{10u}{n^0}} \sqrt{K\widehat{\sigma}^2}
\end{align*}
Combining these results, we find that
\begin{align*}
    (1-\theta)\sum_{k \leq K}\left(\widehat{r}_k - 1\right)h_k^0 \geq - (1-\theta)\left(3\epsilon^{1}\vert \Omega^1 \vert + \epsilon^{0}\vert \Omega^0 \vert + \epsilon^{1}\vert \Omega^{01}\vert\right) - \sqrt{\frac{10uK\widehat{\sigma}^2}{n^0}}.
\end{align*}
To conclude the lower bound on $B_1$, we rely on Lemma \ref{lem:lower_bound_B1}. We find that
\begin{align}\label{eq:bound_B1}
    B_1 \geq  \left(\left(1 - \left(\frac{u}{n^1}\right)^{\frac{1}{4}}\right)\widehat{\sigma}^2 - \left(\frac{u}{n^1}\right)^{\frac{1}{4}}\right)\theta - \left(\sqrt{\frac{10uK\widehat{\sigma}^2}{n^0}} +  \left(13\epsilon^1\vert \Omega^1\vert  + 4\epsilon^1\vert \Omega^{01}\vert  + 3\epsilon^0 \vert\Omega^0\vert\right)\right).
\end{align}

\paragraph{Bound on $\mathbf{V_1}$} Next, we provide an upper bound on the variance
$$V_1 = \sum_{k \leq K}\left(\widehat{r}_k - 1\right)^2h_k^0 + \theta\sum_{k \leq K}\left(\widehat{r}_k - 1\right)^2(h_k^1 - h_k^0).$$
On the one hand, as shown previously, on $\mathcal{E}$,
\begin{align*}
    \sum_{k \leq K}\left(\widehat{r}_k - 1\right)^2h_k^0 \leq 3 \widehat{\sigma}^2.
\end{align*}
On the other hand,
\begin{align*}
    \sum_{k \leq K}\left(\widehat{r}_k - 1\right)^2(h_k^1 - h_k^0) &\leq \max_{k}\frac{\left \vert h_k^1 - h_k^0\right\vert }{\widehat{h}_k^0}\sum_{k \leq K}\left(\widehat{r}_k - 1\right)^2\widehat{h}_k^0\\
    &\leq \max_{k}\frac{\left \vert h_k^1 - h_k^0\right \vert }{\widehat{h}_k^0} \times \widehat{\sigma}^2
\end{align*}
To bound $\max_{k}\frac{\left \vert h_k^1 - h_k^0\right \vert }{\widehat{h}_k^0}$, note that on the $\mathcal{E}$, Lemma \ref{lem:concentration_histogram} implies that $h_k^1 \leq 3 \widehat{h}_k^1$ and $h_k^0 \leq 3 \widehat{h}_k^0$. Thus,
\begin{align*}
    \max_{k}\frac{\left \vert h_k^1 - h_k^0\right \vert }{\widehat{h}_k^0} &\leq \max_{k}\frac{h_k^1}{\widehat{h}_k^0} \lor\max_{k}\frac{h_k^0}{\widehat{h}_k^0}\\
    &\leq 3 (\max_{k}\widehat{r}_k \lor 1).
\end{align*}
Moreover, necessarily $\max_k \widehat{r}_k \geq 1$.
Finally, this implies that
\begin{align*}
    \sum_{k \leq K}\left(\widehat{r}_k - 1\right)^2(h_k^1 - h_k^0)  \leq 3\left(\max_{k} \left \vert \widehat{r}_k - 1\right \vert + 1\right)\times \widehat{\sigma}^2
\end{align*}
\begin{align}\label{eq:bound_V1}
    V_1 \leq 3\widehat{\sigma}^2\left(1 + \theta + \theta \max_{k}\left \vert \widehat{r}_k - 1\right \vert\right).
\end{align}

\paragraph{Conclusion} 
Equation \eqref{eq:Bernstein_H1} implies that for all $t>0$, with probability at least $1-e^{-t}$,

\begin{align*}
    \widehat{S}(\mathbb{X}) >&  B_1 - \sqrt{\frac{2t}{n}V_1} - \frac{t}{3n}\max_{k\leq K}\vert \widehat{r}_k - 1\vert.
\end{align*}
Plugging the bounds on $B_1$ and $V_1$ given respectively in Equations \eqref{eq:bound_B1} and \eqref{eq:bound_V1}, we find that under $H_1$, on the event $\mathcal{E}$, with probability at least $1-e^{-t}$,
\begin{align*}
    \widehat{S}(\mathbb{X}) > &  \left(\left(1 - \left(\frac{u}{n^1}\right)^{\frac{1}{4}}\right)\widehat{\sigma}^2 - \left(\frac{u}{n^1}\right)^{\frac{1}{4}}\right)\theta - \left(\sqrt{\frac{10uK\widehat{\sigma}^2}{n^0}} + \left(13\epsilon^1\vert \Omega^1\vert  + 4\epsilon^1\vert \Omega^{01}\vert  + 3\epsilon^0 \vert\Omega^0\vert\right)\right)\\
    & - \sqrt{\widehat{\sigma}^2}\sqrt{\frac{6t}{n}\left(1 + \theta + \theta \max_{k}\left \vert \widehat{r}_k - 1\right \vert\right)}  - \frac{t}{3n}\max_k\vert \widehat{r}_k - 1\vert.
\end{align*}


\subsubsection{Proof of Lemma \ref{lem:power}}

To prove Lemma \ref{lem:power}, note that the test presented above rejects the null hypothesis if 
$$\widehat{S}(\mathbb{X}) \geq \sqrt{\widehat{\sigma}^2}\left(\sqrt{\frac{10uK}{n^0}} +\sqrt{\frac{6t}{n}}\right)+ \frac{t}{3n}\max_{k}\vert \widehat{r}_k - 1\vert +3 \epsilon^1 \left(\vert \Omega^{01}\vert + \vert \Omega^1\vert \right).$$

Thus, according to Lemma \ref{lem:type_II}, the test rejects the hypotheses $H_0$ with probability larger than $1-\alpha$ under the alternative hypothesis $H_1$ if $\theta$ is such that 
\begin{align*}
\sqrt{\widehat{\sigma}^2}\left(\sqrt{\frac{10uK}{n^0}} +\sqrt{\frac{6t}{n}}\right)&+ \frac{t}{3n}\max_{k}\vert \widehat{r}_k - 1\vert +3 \epsilon^1 \left(\vert \Omega^{01}\vert + \vert \Omega^1\vert \right)\\
 \leq & \left(\left(1 - \left(\frac{u}{n^1}\right)^{\frac{1}{4}}\right)\widehat{\sigma}^2 -  \left(\frac{u}{n^1}\right)^{\frac{1}{4}}\right)\theta  - \left(13\epsilon^1\vert \Omega^1\vert  + 4\epsilon^1\vert \Omega^{01}\vert  + 3\epsilon^0 \vert\Omega^0\vert\right) \\
 & - \sqrt{\widehat{\sigma}^2} \left(\sqrt{\frac{6t}{n}\left(1 + \theta + \theta \max_{k}\left \vert \widehat{r}_k - 1\right \vert\right)}+\sqrt{\frac{10uK}{n^0}} \right)  - \frac{t}{3n}\max_k\vert \widehat{r}_k - 1\vert 
\end{align*}
or equivalently if
\begin{align*}
  &\left(\left(1 - \left(\frac{u}{n^1}\right)^{\frac{1}{4}}\right)\widehat{\sigma}^2 - \left(\frac{u}{n^1}\right)^{\frac{1}{4}}\right)\theta  - \left(16\epsilon^1\vert \Omega^1\vert  + 7\epsilon^1\vert \Omega^{01}\vert  + 3\epsilon^0 \vert\Omega^0\vert\right)\\ 
  &\qquad -\sqrt{\widehat{\sigma}^2} \left(\sqrt{\frac{6t}{n}}\left(\sqrt{1 + \theta + \theta \max_{k}\left \vert \widehat{r}_k - 1\right \vert } + 1\right)+\sqrt{\frac{40uK}{n^0}} \right)
  - \frac{2t}{3n}\max_k\vert \widehat{r}_k - 1\vert \geq 0
\end{align*}
To simplify this expression, note that it is verified if
\begin{align*}
  &\left(\left(1 - \left(\frac{u}{n^1}\right)^{\frac{1}{4}}\right)\widehat{\sigma}^2 - \left(\frac{u}{n^1}\right)^{\frac{1}{4}}\right)\theta  - \sqrt{\frac{6t\widehat{\sigma}^2}{n}\left(1 +\max_{k}\left \vert \widehat{r}_k - 1\right \vert \right)} \sqrt{\theta} \\
   &\qquad - \left(\left(16\epsilon^1\vert \Omega^1\vert  + 7\epsilon^1\vert \Omega^{01}\vert  + 3\epsilon^0 \vert\Omega^0\vert\right)+\sqrt{\widehat{\sigma}^2} \left(\sqrt{\frac{24t}{n}} +\sqrt{\frac{40uK}{n^0}} \right)
  + \frac{2t}{3n}\max_k\vert \widehat{r}_k - 1\vert\right) \geq 0
\end{align*}
Using $(\sqrt{a}+\sqrt{b})^2\leq 2a + 2b$, we see that this is verified if
\begin{align*}
  \theta \geq& \frac{\frac{6t\widehat{\sigma}^2}{n}\left(1 +\max_{k}\left \vert \widehat{r}_k - 1\right \vert \right)}{\left(\left(1 - \left(\frac{u}{n^1}\right)^{\frac{1}{4}}\right)\widehat{\sigma}^2 - \left(\frac{u}{n^1}\right)^{\frac{1}{4}}\right)^2} \\
  &+ 2\frac{\left(16\epsilon^1\vert \Omega^1\vert  + 7\epsilon^1\vert \Omega^{01}\vert  + 3\epsilon^0 \vert\Omega^0\vert\right)+\sqrt{\widehat{\sigma}^2} \left(\sqrt{\frac{24t}{n}} +\sqrt{\frac{40uK}{n^0}} \right)
  + \frac{2t}{3n}\max_k\vert \widehat{r}_k - 1\vert}{\left(1 - \left(\frac{u}{n^1}\right)^{\frac{1}{4}}\right)\widehat{\sigma}^2 - \left(\frac{u}{n^1}\right)^{\frac{1}{4}}}
\end{align*}
or equivalently if 
\begin{align}\label{eq:condition_puissance}
  \theta \geq& \frac{1}{\left(1 - \left(\frac{u}{n^1}\right)^{\frac{1}{4}}\left( 1 + \widehat{\sigma}^2\right)\right)^2} \times  \left(\frac{6t}{n\widehat{\sigma}^2}\left(1 +\max_{k}\left \vert \widehat{r}_k - 1\right \vert \right) \right)\nonumber\\
  &+ \frac{1}{\left(1 - \left(\frac{u}{n^1}\right)^{\frac{1}{4}}\left( 1 + \widehat{\sigma}^2\right)\right)} \times \left(\frac{32\epsilon^1\vert \Omega^1\vert  + 14\epsilon^1\vert \Omega^{01}\vert  + 6\epsilon^0 \vert\Omega^0\vert}{\widehat{\sigma}^2}+ \sqrt{\frac{96t}{n\widehat{\sigma}^2}} +\sqrt{\frac{160uK}{n^0\widehat{\sigma}^2}}
  + \frac{4t}{3n\widehat{\sigma}^2}\max_k\vert \widehat{r}_k - 1\vert\right)
\end{align}
To simplify this expression, note that if $\epsilon^{1} \leq 1$, $\frac{u}{n^1} \leq \frac{1}{3}$, so $\left(1 - \left(\frac{u}{n^1}\right)^{\frac{1}{4}}\left( 1 + \widehat{\sigma}^2\right)\right)^{-1} \leq 5.$ Thus, Equation \eqref{eq:condition_puissance} is verified if
\begin{align*}
  \theta \geq& \frac{150t}{n\widehat{\sigma}^2}\left(1 +\max_{k}\left \vert \widehat{r}_k - 1\right \vert \right) + \frac{160\epsilon^1\vert \Omega^1\vert  + 70\epsilon^1\vert \Omega^{01}\vert  + 30\epsilon^0 \vert\Omega^0\vert}{\widehat{\sigma}^2}+ 5\sqrt{\frac{96t}{n\widehat{\sigma}^2}} \\&+5\sqrt{\frac{160uK}{n^0\widehat{\sigma}^2}}
  + \frac{10t}{3n\widehat{\sigma}^2}\max_k\vert \widehat{r}_k - 1\vert.
\end{align*}
Moreover, setting $k^* \in \argmax_k\vert\widehat{r}_k - 1\vert$, we have
\begin{align*}
    \max_k\vert \widehat{r}_k - 1\vert^2\epsilon^0 &\leq \vert \widehat{r}_{k^*} - 1\vert^2 \times \widehat{h}_{k^*}^0 \leq \widehat{\sigma}^2,
\end{align*}
so $$\frac{t}{n}\max_k\vert \widehat{r}_k - 1\vert \leq \sqrt{\frac{t}{n}}\sqrt{\epsilon^0\max_k\vert \widehat{r}_k - 1\vert^2} \leq \sqrt{\frac{t\widehat{\sigma}^2}{n}}.$$
Thus, Equation \eqref{eq:condition_puissance} is verified if
\begin{align*}
\theta \geq& \left(353\sqrt{\frac{t}{n\widehat{\sigma}^2}} + \frac{400\sqrt{u}K^1}{\widehat{\sigma}^2\sqrt{n^1}} + \frac{30\epsilon^0K^0}{\widehat{\sigma}^2}+64\sqrt{\frac{uK}{n^0\widehat{\sigma}^2}}\right).
\end{align*}
where $K^1 = \vert \Omega^1\vert + \vert \Omega^{0}\vert$ and $K^0 = \vert \Omega^0\vert$.

\subsubsection{Proof of Lemma \ref{lem:concentration_histogram}}\label{app:proof_concentration_histogram}

For $a\in \{0,1\}$ and $k \leq K$, applying Bernstein's inequality \ref{thm:Bernstein} to the random variables $\left(\mathds{1}\{X_i^a \in B_k\} - h^a_k\right)_{i \leq n^a}$, we see that 
\begin{align}
    \mathbb{P}\left(\left \vert \frac{N^a_k}{n^a} -  h^a_k \right \vert\geq \sqrt{2\frac{uh^a_k}{n^a}} + \frac{u}{3n^a} \right)\leq& 2e^{-u}.
\end{align}

Thus, the event $\mathcal{E} = \left \{\left \vert \frac{N^a_k}{n^a} -  h^a_k \right \vert\leq \sqrt{2\frac{uh^a_k}{n^a}} + \frac{u}{3n^a} \text{ for all }  k\leq K, a \in \{0,1\}\right \}$ occurs with probability larger than $1 - 4Ke^{-u}$. Note that on this event, Equation \eqref{eq:bound_freq} is verified for all $a\in \{0,1\}$ and $k \leq K$. Moreover,
\begin{align*}
    h^a_k - \frac{N^a_k}{n^a} - \sqrt{2\frac{uh^a_k}{n^a}} - \frac{u}{3n^a} &< 0.
\end{align*}
Thus, $h^a_k < r^2$, where $r$ is the positive root of the polynomial $x^2 - \sqrt{2\frac{u}{n^a}}x - \left(\frac{N^a_k}{n^a} + \frac{u}{3n^a}\right)$. Direct calculations show that
\begin{align}\label{eq:r_plus}
r^2 &\leq \frac{N_k^a}{n^a} + \frac{\sqrt{2uN_k^a}}{n^a}+ \frac{3u}{n^a}.
\end{align}
Similarly, on $\mathcal{E}$, 
\begin{align*}
    h^a_k - \frac{N^a_k}{n^a} + \sqrt{2\frac{uh^a_k}{n^a}} + \frac{u}{3n^a} &>  0.
\end{align*}
Thus, $h^a_k > r'^2$, where $r'$ is the positive root of the polynomial $x^2 + \sqrt{2\frac{u}{n^a}}x - \left(\frac{N^a_k}{n^a} - \frac{u}{3n^a}\right)$. Straightforward calculations show that
\begin{align}\label{eq:r_minus}
r'^2 &\geq \frac{N_k^a}{n^a} - \frac{\sqrt{2uN_k^a}}{n^a}.
\end{align}
Combining Equations \eqref{eq:r_plus} and \eqref{eq:r_minus} yields Equation \eqref{eq:bound_est}.

\subsubsection{Proof of Lemma \ref{lem:error_bound}}
\label{app:proof_error_bound}
We begin by proving Equation \eqref{eq:bound_h_Omega}. For $a = 1$ and $k \in \Omega^1\cup \Omega^{01}$, we have that $N_k^1 \leq \epsilon^1 n^1$. On the event $\mathcal{E}$, Lemma \ref{lem:concentration_histogram} implies that $$h^1_k \leq \epsilon^1 + \sqrt{\frac{2u\epsilon^1}{n^1}} + \frac{3u}{n^1} = \epsilon^1\left(1 + \sqrt{\frac{2\epsilon_1}{3}}\right) + \left(\epsilon^1\right)^2.$$
Since $\epsilon^1 \leq 1$, we see that $h^1_k \leq 3\epsilon^1 = \widehat{h}_k^1,$
which proves Equation \eqref{eq:bound_h_Omega} for $a = 1$.

For $a = 0$ and $k \in \Omega^0$, we have that $N_k^0 \leq \epsilon^0 n^0$. On the event $\mathcal{E}$, Lemma \ref{lem:concentration_histogram} implies that $$h^0_k \leq \epsilon^0 + \sqrt{\frac{2u\epsilon^0}{n^0}} + \frac{3u}{n^0} \leq 3\epsilon^0 = \widehat{h}_k^0.$$ 

Finally, for $a = 0$ and $k \in \Omega^{01}$, we have that $N_k^0 \leq \epsilon^1 n^0$. On the event $\mathcal{E}$, Lemma \ref{lem:concentration_histogram} implies that $$h^0_k \leq \epsilon^1 + \sqrt{\frac{2u\epsilon^1}{n^0}} + \frac{3u}{n^0}.$$ Since $n^1 \leq n^0$, this implies that 
$$h^0_k \leq \epsilon^1 + \sqrt{\frac{2u\epsilon^1}{n^1}} + \frac{3u}{n^1}.$$ The same arguments as previously show that this implies that $$h^0_k \leq 3\epsilon^1 = \widehat{h}^1_k.$$

\bigskip

We now turn to the proof of Equation \eqref{eq:bound_h}. For $a = 0$ and $k \notin \Omega^0\cup \Omega^{01}$, we have that $\widehat{h}_k^0 = \frac{N_k^0}{n^0}$, so Lemma \ref{lem:concentration_histogram} implies that on the event $\mathcal{E}$, $\left\vert h_k^0 - \widehat{h}_k^0 \right\vert \leq \sqrt{\frac{2u\widehat{h}_k^0}{n^0}} + \frac{3u}{n^0}.$ Moreover, since $k \notin \Omega^0\cup \Omega^{01}$, $\widehat{h}_k^0 \geq \epsilon^0 = \frac{3u}{n^0}$. Thus, $$\left\vert h_k^0 - \widehat{h}_k^0 \right\vert \leq \sqrt{\frac{2u\widehat{h}_k^0}{n^0}} + \sqrt{\frac{3u\widehat{h}_k^0}{n^0}} \leq \sqrt{\frac{10u\widehat{h}_k^0}{n^0}}.$$
Using $\frac{3u}{n^0} \leq \widehat{h}_k^0$, we also find that 
$$\left\vert h_k^0 - \widehat{h}_k^0 \right\vert \leq \sqrt{\frac{2}{3}} \widehat{h}_k^0 + \widehat{h}_k^0\leq 2\widehat{h}_k^0.$$

We then prove Equation \eqref{eq:bound_h} for $a = 1$ and $k \notin \Omega^1\cup \Omega^{01}$ : in this case, we have that $\widehat{h}_k^1 = \frac{N_k^1}{n^1}$, so Lemma \ref{lem:concentration_histogram} implies that on the event $\mathcal{E}$, $\left\vert h_k^1 - \widehat{h}_k^1 \right\vert \leq \sqrt{\frac{2u\widehat{h}_k^1}{n^1}} + \frac{3u}{n^1}.$ Moreover, since $k \notin \Omega^0\cup \Omega^{01}$, $\widehat{h}_k^1 \geq \epsilon^1 \geq \left(\epsilon^1\right)^2 = \frac{3u}{n^1}$. Thus, $$\left\vert h_k^1 - \widehat{h}_k^1 \right\vert \leq \sqrt{\frac{2u\widehat{h}_k^1}{n^1}} + \sqrt{\frac{3u\widehat{h}_k^1}{n^1}} \leq \sqrt{\frac{10u\widehat{h}_k^1}{n^1}}.$$
Using $\frac{3u}{n^1} \leq \sqrt{\frac{3u}{n^1}} \leq \widehat{h}_k^1$, we also find that 
$$\left\vert h_k^1 - \widehat{h}_k^1 \right\vert \leq \sqrt{\frac{2}{3}} \widehat{h}_k^1 + \widehat{h}_k^1\leq 2\widehat{h}_k^1.$$

\subsubsection{Proof of Lemma \ref{lem:lower_bound_B1}}\label{app:proof_lower_bound_B1}

We begin by decomposing $\sum_{k} (\widehat{r}_k - 1)h_k^1$ as
\begin{align*}
    \sum_{k} (\widehat{r}_k - 1)h_k^1 = \sum_{k \in \Omega^1} (\widehat{r}_k - 1)h_k^1 + \sum_{k \in \Omega^{01}} (\widehat{r}_k - 1)h_k^1 + \sum_{k \notin \Omega^1\cup \Omega^{01}} (\widehat{r}_k - 1)h_k^1.
\end{align*}
We obtain a lower bound on the first term by noticing that $\widehat{r}_k - 1 \geq -1$. Thus,
\begin{align*}
    \sum_{k \in \Omega^1} (\widehat{r}_k - 1)h_k^1 \geq - \sum_{k \in \Omega^1}h_k^1.
\end{align*}
Using Equation \eqref{eq:bound_h_Omega} in Lemma \ref{lem:error_bound}, we find that for all $k \in \Omega^1$, $h_k^1 \leq \widehat{h_k}^1 = 3\epsilon^1$. Thus, 
\begin{align*}
    \sum_{k \in \Omega^1} (\widehat{r}_k - 1)h_k^1 \geq - 3\epsilon^1\vert\Omega^1\vert.
\end{align*}
On the other hand, for $k \in \Omega^1$, $\widehat{r}_k \leq 3$. Thus, 
\begin{align*}
    \sum_{k \in \Omega^1} (\widehat{r}_k - 1)\widehat{h}_k^1 \leq 6\epsilon^1 \vert\Omega^1 \vert.
\end{align*}
This shows that 
\begin{align}\label{eq:borne_omega_1}
    \sum_{k \in \Omega^1} (\widehat{r}_k - 1)h_k^1 \geq \sum_{k \in \Omega^1} (\widehat{r}_k - 1)\widehat{h}_k^1 - 9\epsilon^1 \vert\Omega^1 \vert.
\end{align}

To obtain a lower bound on the second term, we note that for $k \in \Omega^{01}$,  $\widehat{h}_k^0 = \widehat{h}_k^1$. Thus, $\widehat{r}_k - 1 = 0$, and
\begin{align}\label{eq:borne_omega_0}
    \sum_{k \in \Omega^{01}} (\widehat{r}_k - 1)h_k^1 = \sum_{k \in \Omega^{01}} (\widehat{r}_k - 1)\widehat{h}_k^1.
\end{align}

Finally, we derive a lower bound on $\sum_{k \notin \Omega^1\cup \Omega^{01}} (\widehat{r}_k - 1)h_k^1$. We note that by Lemma \ref{lem:concentration_histogram}, on the event $\mathcal{E}$, for all $k\leq K$,
$$h^1_k  \geq  \frac{N^1_k}{n^1} \left(1 -  \sqrt{\frac{2u}{N^1_k}}\right).$$
Moreover, for all $k \notin \Omega^1 \cup \Omega^{01}$, $\widehat{h}_k^1 = \frac{N^1_k}{n^1} \geq \sqrt{\frac{3u}{n^1}}$, so $N^1_k \geq \sqrt{3un^1}$. This implies in particular that for all $k \notin \Omega^1 \cup \Omega^{01}$, $h^1_k  \geq  \widehat{h}_k^1\left(1 - \left(\frac{u}{n^1}\right)^{\frac{1}{4}}\right)$. Thus, 
\begin{align*}
    \sum_{k \notin \Omega^1\cup \Omega^{01}} \widehat{r}_kh_k^1 & \geq \left(1 - \left(\frac{u}{n^1}\right)^{\frac{1}{4}}\right)\sum_{k \notin \Omega^1\cup \Omega^{01}} \widehat{r}_k\widehat{h}_k^1.
\end{align*}

This implies that
\begin{align}
    \sum_{k \notin \Omega^1\cup \Omega^{01}} \left(\widehat{r}_k - 1\right) h_k^1 & \geq \left(1 - \left(\frac{u}{n^1}\right)^{\frac{1}{4}}\right)\left(\sum_{k \notin \Omega^1\cup \Omega^{01}} (\widehat{r}_k - 1)\widehat{h}_k^1\right) + \sum_{k \notin \Omega^1\cup \Omega^{01}}(\widehat{h}_k^1 - h_k^1) - \left(\frac{u}{n^1}\right)^{\frac{1}{4}}\sum_{k \notin \Omega^1\cup \Omega^{01}}\widehat{h}_k^1 \nonumber\\
\end{align}

\noindent Now, on the one hand, $\sum_{k \notin \Omega^1\cup \Omega^{01}} h_k^1 \leq 1$. On the other hand, $$\sum_{k \notin \Omega^1\cup \Omega^{01}} \widehat{h}_k^1 = 1 - \sum_{k \in \Omega^1\cup \Omega^{01}} \frac{N_k^1}{n^1} \geq 1 - \epsilon^1\left(\vert \Omega^1\vert + \vert \Omega^{01}\vert \right)$$ and $\sum_{k \notin \Omega^1\cup \Omega^{01}} \widehat{h}_k^1\leq 1$. Thus, 

\begin{align}\label{eq:borne_sigma_reste}
    \sum_{k \notin \Omega^1\cup \Omega^{01}} \left(\widehat{r}_k - 1\right) h_k^1 &\geq \left(1 - \left(\frac{u}{n^1}\right)^{\frac{1}{4}}\right) \sum_{k \notin \Omega^1\cup \Omega^{01}} (\widehat{r}_k- 1)\widehat{h}_k^1 + \left(\frac{u}{n^1}\right)^{\frac{1}{4}} - \epsilon^1\left(\vert \Omega^1\vert + \vert \Omega^{01}\vert \right)
\end{align}
Combining Equations \eqref{eq:borne_omega_1}, \eqref{eq:borne_omega_0}, and \eqref{eq:borne_sigma_reste}, we find that
\begin{align*}
    \sum_k \left(\widehat{r}_k - 1\right)h_k^1 &\geq \left(1 - \left(\frac{u}{n^1}\right)^{\frac{1}{4}}\right)\sum_k \left(\widehat{r}_k - 1\right)\widehat{h}_k^1  -  \left(\frac{u}{n^1}\right)^{\frac{1}{4}} - \epsilon^1\left(10\vert \Omega^1\vert + \vert \Omega^{01}\vert\right).
\end{align*}
To conclude the proof, we note that 
\begin{align*}
    \widehat{\sigma}^2 &= \sum_{k} \left(\widehat{r}_k - 1\right)^2\widehat{h}_k^0\\
    &= \sum_{k} \left(\widehat{r}_k - 1\right)\widehat{h}_k^1 + \sum_{k} \widehat{h}_k^0 - \sum_{k}\widehat{h}_k^1.
\end{align*}
Thus, 
\begin{align*}
    \sum_{k} \left(\widehat{r}_k - 1\right)\widehat{h}_k^1
    &= \widehat{\sigma}^2  + \sum_{k} \widehat{h}_k^1 - \sum_{k}\widehat{h}_k^0.
\end{align*}
As previously, we have that $\sum_{k} \widehat{h}_k^0\leq 1 + 3\epsilon^0\vert \Omega^0 \vert + 3\epsilon^1\vert \Omega^{01} \vert$, and $\sum_{k} \widehat{h}_k^1\geq 1.$ Thus,
\begin{align*}
    \sum_{k} \left(\widehat{r}_k - 1\right)\widehat{h}_k^1 &\geq \widehat{\sigma}^2 - 3\left(\epsilon^0 \vert \Omega^0\vert  + \epsilon^1 \vert \Omega^{01}\vert\right),
\end{align*}
which implies
\begin{align*}
    \sum_k \left(\widehat{r}_k - 1\right)h_k^1 &\geq \left(1 - \left(\frac{u}{n^1}\right)^{\frac{1}{4}}\right)\widehat{\sigma}^2 -  \left(\frac{u}{n^1}\right)^{\frac{1}{4}} - \left(\epsilon^1\left(13\vert \Omega^1\vert + 4\vert \Omega^{01}\vert\right) + 3\epsilon^0\vert \Omega^0\vert\right).
\end{align*}

\subsection{Proof of main lemmas and corollaries} \label{app:proof_main_lemmas}
\subsubsection{Proof of Lemma \ref{lem:known_f}}\label{app:proof_known_f}
Before proving Lemma  \ref{lem:known_f}, we emphasize that when $f^0$ and $f^1$ are positive on $[0,1]^d$, $\sigma^2$ is finite, and $\left(\frac{f^1(X_1)}{f^0(X_1)} - 1\right)$ belongs to $L^2((1-\theta)\times f^0 + \theta\times f^1)$ for all $\theta \in [0,1]$.

Moreover, under $H_0$, $\mathbb{E}\left[\left(\frac{f^1(X_1)}{f^0(X_1)} - 1\right)\right] = 0$, and $\mathbb{V}\left[\left(\frac{f^1(X_1)}{f^0(X_1)} - 1\right)\right] = \sigma^2$. Then, the Central Limit Theorem implies that 
$$\frac{\sqrt{n}S_n(\mathbb{X})}{\sqrt{\sigma^2}} \overset{d}{\rightarrow} \mathcal{N}\left(0,1\right).$$
Thus, the test presented in Lemma \ref{lem:known_f} is asymptotically of level $\alpha$.

To prove the second part of the Lemma, we consider alternatives $\theta_n = \frac{h}{\sqrt{n}}$ for some $h > 0$. Let us denote $\mathbb{P}_n$ the law of the sample $X_1, ..., X_n$ under $H_0$, and $\mathbb{Q}_n$ the law of the sample $X_1, ..., X_n$ under the alternative indexed by $\theta_n$. Under the assumption that $f^0$ and $f^1$ are positive on $[0,1]^d$, it is easy to see that the family of models indexed by $\theta$ is differentiable in quadratic mean (using, e.g., Lemma 7.6 in \cite{vaart_1998}). Then, using Theorem 7.2 in \cite{vaart_1998}, we find that
\begin{align*}
    \log\left(\frac{d\mathbb{Q}_n}{d\mathbb{P}_n}(\mathbb{X})\right) & = \frac{h}{\sqrt{n}}\sum_{i\leq n}\left(\frac{f^1(X_i)}{f^0(X_i)} - 1\right) - \frac{h^2}{2n}\sum_{i\leq n}\left(\frac{f^1(X_i)}{f^0(X_i)} - 1\right)^2+ o_{\mathbb{P}_n}(1).
\end{align*}
The Law of Large Numbers implies that $$\frac{h^2}{2n}\sum_{i\leq n}\left(\frac{f^1(X_i)}{f^0(X_i)} - 1\right)^2 \overset{\mathbb{P}}{\rightarrow} \frac{h^2\sigma^2}{2},$$ and the Central Limit Theorem implies that $$\frac{h}{\sqrt{n}}\sum_{i\leq n}\left(\frac{f^1(X_i)}{f^0(X_i)} - 1\right) \overset{d}{\rightarrow} \mathcal{N}\left(0, h^2\sigma^2\right).$$
Thus, under $\mathbb{P}_n$,
$$\left({\sqrt{n}S_n(\mathbb{X}) \atop \log\left(\frac{d\mathbb{Q}_n}{d\mathbb{P}_n}(\mathbb{X})\right)}\right) \overset{d}{\rightarrow}\mathcal{N}\left(\left({0 \atop -\frac{h^2\sigma^2}{2}}\right), \left({\sigma^2\atop h\sigma^2} \ \ \  {h\sigma^2\atop h^2\sigma^2}\right)\right).$$
Now, Le Cam's third Lemma (see, e.g., Example 6.7 in \cite{vaart_1998}) implies that under $\mathbb{Q}_n$,
$$\sqrt{n}S_n(\mathbb{X}) \overset{d}{\rightarrow} \mathcal{N}\left(h\sigma^2, \sigma^2\right).$$
Applying Theorem 14.7 in \cite{vaart_1998}, we find that the power function $$\pi_n\left(\frac{h}{\sqrt{n}}\right) = \mathbb{Q}_n\left(S_n(\mathbb{X}) \geq \frac{\sqrt{\sigma^2} \Phi(1-\alpha)^{-1}}{\sqrt{n}}\right)$$ converges as follows :
\begin{align*}
    \pi_n\left(\frac{h}{\sqrt{n}}\right) \rightarrow 1 - \Phi(\Phi^{-1}(1-\alpha) - h\sigma).
\end{align*}
This concludes the proof of Lemma \ref{lem:known_f}.

\subsubsection{Proof of Lemma \ref{lem:known_f_somewhat_fast_rates}} \label{app:proof_known_f_fast}
Let us define 
$$p^0_{R_n} = \int_{R_n} f^0(x)dx \quad \quad \text{ and } \quad \quad p^1_{R_n} = \int_{R_n} f^1(x)dx.$$ To prove the first part of the Lemma, note that under $H_0$, 
$$\mathbb{P}\left(\sum_{i\leq n}\mathds{1}\{X_i\in R_n\} = 0\right) = \exp\left(n\log\left(1-p^0_{R_n}\right)\right).$$
Using the bound $\log(1-x)\geq \frac{-x}{2}\times \frac{2-x}{1-x}$ for $x < 1$, we find that 
$$\mathbb{P}\left(\sum_{i\leq n}\mathds{1}\{X_i\in R_n\} = 0\right) \geq \exp\left(-np^0_{R_n}\times \frac{2-p^0_{R_n}}{2-2p^0_{R_n}} \right).$$
Thus, this test has type I error lower than $\alpha$ if
$$\alpha \geq 1-\exp\left(-A\times \frac{2n-A}{2(n-A)}\right).$$
This proves the first part of Lemma \ref{lem:known_f_somewhat_fast_rates}.

\bigskip

To prove the second part of Lemma \ref{lem:known_f_somewhat_fast_rates}, note that under the alternative $H_1$,
$$\mathbb{P}_{\theta_n}\left(\sum_{i\leq n}\mathds{1}\{X_i \in R\}  = 0\right) = \left(1 - \theta p_R^1\right)^n.$$
Using the bound $\log(1 - x) \leq \frac{-2x}{2-x}$ for $x\in (0,1),$ we find that
\begin{align*}
    \mathbb{P}_{\theta_n}\left(\sum_{i\leq n}\mathds{1}\{X_i \in R\}  = 0\right) &\leq \exp\left(\frac{-2p_R^1n\theta}{2-p_R^1\theta}\right) \\
    &\leq \exp\left(\frac{-2Bn^{1-\gamma}\theta}{2-n^{-\gamma}\theta}\right).
\end{align*}
Thus, to ensure that the test has power $\beta$, it is enough to have
\begin{align*}
    \frac{2Bn^{1-\gamma}\theta}{2-n^{-\gamma}\theta}\geq \log(1/\beta).
\end{align*}
This is in particular verified if $\theta \geq \frac{\log(1/\beta)}{Bn^{1-\gamma}}$.

\subsubsection{Proof of Lemma \ref{lem:control_sigma}} \label{app:proof_control_sigma}

Recall that $\overline{\sigma}^1_{\mathcal{P}}$ is assumed to be finite. Defining $\mathcal{B} = \{k : h_k^0 > 0\}$ and $\overline{\mathcal{B}} = \{k : h_k^0 = 0\}$, we see that $\{k : h_k^1 = 0\} \subset \overline{\mathcal{B}}$. This implies in particular that $\Omega^0 = \emptyset$, and that $\overline{\mathcal{B}} \subset \Omega^{01}$. Now,  for all $k \in \Omega^{01}$, $\widehat{r}_k = 1$, so
\begin{align*}
    \widehat{\sigma}^2 &= \sum_{k\in \mathcal{B}}\left(\frac{\widehat{h}^1_k}{\widehat{h}^0_k} - 1\right)^2\widehat{h}^0_k\\
    &=\sum_{k\in \mathcal{B}}\left( \frac{(\widehat{h}^1_k)^2}{\widehat{h}^0_k} - 2\widehat{h}^1_k + \widehat{h}^0_k\right).
\end{align*}
Using $\left \vert\sum_{k\in \mathcal{B}}\widehat{h}^1_k - 1\right\vert \leq \epsilon^1 \left(\vert \Omega^1\vert + \vert \Omega^{01}\vert\right)$ and $\left \vert\sum_{k\in \mathcal{B}}\widehat{h}^0_k - 1\right\vert \leq \epsilon^1 \vert \Omega^{01}\vert$, we find
\begin{align*}
    \left \vert \widehat{\sigma}^2 - \left(\sum_{k\notin \overline{\mathcal{B}}}\frac{(\widehat{h}^1_k)^2}{\widehat{h}^0_k} - 1\right)\right \vert \leq 5\epsilon^1\vert \Omega\vert.
\end{align*}
Similarly, we have $\overline{\sigma}^2_{\mathcal{P}} =  \sum_{k\in \mathcal{B}}\frac{(h^1_k)^2}{h^0_k} - 1$. Combining these expressions, we find that
\begin{align*}
    \left \vert \widehat{\sigma}^2- \overline{\sigma}^2_{\mathcal{P}} \right \vert &\leq 5\epsilon^1 \vert \Omega \vert + \sum_{k\in \mathcal{B}}\left \vert \frac{(\widehat{h}^1_k)^2}{\widehat{h}^0_k} - \frac{(h^1_k)^2}{h_k^0}\right \vert\\
    &\leq 5\epsilon^1 \vert \Omega \vert + \sum_{k\in \mathcal{B}}\frac{\left\vert \widehat{h}^1_k -h_k^1\right \vert \left(\widehat{h}^1_k+h_k^1\right)}{h_k^0} + \sum_{k\in \mathcal{B}}\left \vert\frac{h^0_k}{\widehat{h}^0_k} - 1\right \vert h^0_k\widehat{h}^1_k
\end{align*}
Now, Lemma \ref{lem:concentration_histogram} shows that on the event $\mathcal{E}$, for all $k\notin \Omega^{1}\cup \Omega^{01}$
\begin{align*}
    \left\vert \widehat{h}^1_k - h_k^1\right \vert \leq \sqrt{\frac{10u}{n^1}}.
\end{align*}
On the other hand, for all $k\in \Omega^{1}\cup \Omega^{01}$, $$\left\vert \widehat{h}^1_k - h_k^1\right \vert \leq 3\epsilon^1 \leq \sqrt{\frac{27u}{n^1}}.$$
Moreover, Lemma \ref{lem:concentration_histogram} shows that on the event $\mathcal{E}$, $h_k^1 \leq 3\widehat{h}_k^1$. Thus,
$$\sum_{k\in \mathcal{B}}\frac{\left\vert \widehat{h}^1_k -h_k^1\right \vert \left(\widehat{h}^1_k+h_k^1\right)}{\widehat{h}^0_k} \leq \sqrt{\frac{27u}{n^1}}\sum_{k\in \mathcal{B}}\frac{4h^1_k}{h^0_k}.$$
Similarly, Lemma \ref{lem:concentration_histogram} reveals that for all $k$, 
\begin{align*}
    \left \vert\frac{h^0_k}{\widehat{h}^0_k} - 1\right \vert \leq \frac{\sqrt{\frac{2u}{n^0h_k^0}} + \frac{u}{3n^0h_k^0}}{1 - \sqrt{\frac{2u}{n^0h_k^0}} - \frac{u}{3n^0h_k^0}}.
\end{align*}
This proves that 
\begin{align*}
    \left \vert \widehat{\sigma}^2- \overline{\sigma}^2_{\mathcal{P}} \right \vert \leq 5\sqrt{\frac{3u}{n^1}}\vert \Omega \vert + \sqrt{\frac{27u}{n^1}}\sum_{k\in \mathcal{B}}\frac{2h^1_k}{h^0_k} + \sum_{k\in \mathcal{B}}\frac{\sqrt{\frac{2uh_k^0}{n^0}} + \frac{uh_k^0}{3n^0}}{1 - \sqrt{\frac{2u}{n^0h_k^0}} - \frac{u}{3n^0h_k^0}}
\end{align*}
Now, for all $k\in \mathcal{B}$, $h_k^0$ is stricly positive. Then, for $n^0$, $n^1$ and $n$ large enough, $\min_{k}(h^0_k - \sqrt{2\frac{uh^0_k}{n^0}} - \frac{u}{3n^0}) > \epsilon^1$, so $\Omega^{01} \cap \mathcal{B} = \emptyset$. When $n^0$ is large enough so that $\sqrt{\frac{2uh_k^0}{n^0}} + \frac{uh_k^0}{3n^0} \leq 1/2$, there exists a constant $C_{\overline{f}_{\mathcal{P}}^0, \overline{f}_{\mathcal{P}}^1, K}$ depending only on $\overline{f}_{\mathcal{P}}^0$ $ \overline{f}_{\mathcal{P}}^1$, and $K$ such that
\begin{align*}
    \left \vert \widehat{\sigma}^2- \overline{\sigma}^2_{\mathcal{P}} \right \vert \leq C_{\overline{f}_{\mathcal{P}}^0, \overline{f}_{\mathcal{P}}^1, K}\left(\sqrt{\frac{u}{n^1}}+\sqrt{\frac{u}{n^0}}\right).
\end{align*}

\subsubsection{Proof of Corollary \ref{cor:power_slow_rate}} \label{app:proof_power_slow_rate}

Let us assume that $n$, $n^0$ and $n^1$ are large enough so that Lemma \ref{lem:control_sigma} holds, and that on the event $\mathcal{E}$,
$$\left \vert \widehat{\sigma}^2- \overline{\sigma}^2_{\mathcal{P}} \right \vert \leq\frac{\overline{\sigma}^2_{\mathcal{P}} }{2}.$$ 
Arguments similar to that of Lemma \ref{lem:control_sigma} show that for $n$, $n^0$ and $n^1$ large enough, $K^0 = K^1 = 0$. Moreover, for $n$, $n^0$ and $n^1$ large enough, $K^0 \sqrt{\frac{u}{n^0\widehat{\sigma}^2}} \leq \sqrt{K}$ and $ K^0\sqrt{\frac{t}{n\widehat{\sigma}^2}} \leq 1$. Then, Theorem \ref{thm:type_I_II} proves that the Density Ratio Test has power $1-\alpha$ against alternatives such that
\begin{align*}
    \theta \geq \frac{C\sqrt{K(u\vee t)}}{\overline{\sigma}^2_{\mathcal{P}}}\left(\sqrt{\frac{1}{n}} + \sqrt{\frac{1}{n^0}}\right).
\end{align*}

\subsubsection{Proof of Lemma \ref{lem:lb_h}} \label{app:proof_lb_h}
The proof of Lemma \ref{lem:power} follows from noticing that if $h_k^1 \geq \sqrt{\frac{5u}{n^1}}$, Lemma \ref{lem:concentration_histogram} implies that on $\mathcal{E}$, $\frac{N_k^1}{n^1} > \sqrt{\frac{3u}{n^1}}$, so $k \notin \Omega^1 \cup \Omega^{01}$. On the other hand, the same reasoning as in the proof of Lemma \ref{lem:lower_bound_B1} shows that this implies that $\widehat{h}_k^1 \geq h_k^1\left(1 - \left(\frac{u}{n^1}\right)^{\frac{1}{4}}\right)$. Moreover, since $ \sqrt{\frac{5u}{n^1}}\leq 1$, $\left(1 - \left(\frac{u}{n^1}\right)^{\frac{1}{4}}\right) \geq \frac{1}{3}$.
Thus,
$$\widehat{\sigma}^2 \geq \frac{(\widehat{h}_k^1)^2}{\widehat{h}_k^0} \geq \left(1 - \left(\frac{u}{n^1}\right)^{\frac{1}{4}}\right)^2 \frac{(h_k^1)^2}{\epsilon^0} \geq \frac{(h_k^1)^2}{3\epsilon^0}.$$
Now, applying Theorem \ref{thm:type_I_II}, we find that the Density Ratio Test has power $1-\alpha$ against alternatives such that
\begin{align*}
\theta \geq& CK\left(\sqrt{\frac{t}{n\widehat{\sigma}^2}} + \frac{\sqrt{u}}{\widehat{\sigma}^2\sqrt{n^1}} + \sqrt{\frac{u}{n^0\widehat{\sigma}^2}}\right).
\end{align*}
Note that 
$$\frac{\sqrt{u}}{\widehat{\sigma}^2\sqrt{n^1}} \leq \sqrt{\frac{5u}{n^1}}\times \frac{3\epsilon^0}{\sqrt{5}h_1^2}.$$ 
Since $\sqrt{n^1}h_k^1\geq 1$,
$$\frac{\sqrt{u}}{\widehat{\sigma}^2\sqrt{n^1}} \leq \frac{2\epsilon^0}{h_1}.$$ 
Moreover, $\epsilon^0 = \frac{t}{n} \vee \frac{3u}{n^0}$, so the Density Ratio Test has power $1-\alpha$ against alternatives such that
\begin{align*}
\theta \geq& \frac{C_{\alpha, K}\epsilon^0}{h_k^1}
\end{align*}
for some constant $C_{\alpha, K}$ depending on $\alpha$ and $K$.
\subsubsection{Proof of Corollary \ref{cor:max_signal}}

First, we note that by construcrion, there are at most $2K_{max}$ different bins in the partition sequence $\left(\mathcal{P}_K\right){k\leq K_{max}}$), and that those nodes correspond to inner and outer node of the tree corresponding to $\mathcal{P}_{K_{max}}$. Therefore, to control the error of $h^0_k$ and $h^1_k$ uniformly for all bins $k$ in all $K_{max}$ partitions $\mathcal{P}_K$, it is enough to control the error of $h^0_k$ and $h^1_k$ for the $2K_{max}$ bins $k$ corresponding to inner and outer nodes of the largest tree. We choose $u = \log(8K_{max}/\alpha)$ and $t = \log(2/\alpha)$, and we follow the same proof as in Lemma \ref{lem:concentration_histogram} to show that the error bounds hold with the same probabilities for all bins in partition $\mathcal{P}_{K^*}$. Then, the control of the type I and type II error follows similarly. 
\subsubsection{Signal decreases after approximation} \label{app:proof_signal_decrease}

In this Section, we show that approximating the densities $f^0$ and $f^1$ using piece-wise constant functions leads to a decrease in the test signal. 

\begin{lemma}\label{app:sigma_decroit}
    For all functions $f^0$, $f^1$ positive on $[0,1]^d$, and for all partition $\mathcal{P}$ of $[0,1]^d$, 
\begin{align*}
    \overline{\sigma}_{\mathcal{P}}^2\leq \sigma^2.\label{eq:sigma_decroit}
\end{align*}
\end{lemma}

\begin{proof}

Note that it is enough to show that for all bin $B_k$ of partition $\mathcal{P}$, 
\begin{align*}
    \left(\frac{h^1_k}{h_k^0}-1\right)^2h_k^0 \leq \int_{B_k}\left(\frac{f^1(x)}{f^0(x)}-1\right)^2f^0(x)dx.
\end{align*}
The results will follow by summing over the different intervals. Now,
\begin{align*}
    \left(\frac{h^1_k}{h_k^0}-1\right)^2h_k^0 = \underset{(h^1_k)^2}{h_k^0} + 2 h_k^1 - h_k^0,
\end{align*}
and 
\begin{align*}
    \int_{B_k}\left(\frac{f^1(x)}{f^0(x)}-1\right)^2f^0(x)dx &= \int_{B_k}\frac{f^1(x)^2}{f^0(x)}dx - 2\int_{B_k}f^1(x)dx + \int_{B_k}f^0(x)dx\\
    &=\int_{B_k}\frac{f^1(x)^2}{f^0(x)}dx - 2h_k^1 + h_k^0dx.
\end{align*}
We conclude the proof by showing that $\underset{(h^1_k)^2}{h_k^0} \leq \int_{B_k}\frac{f^1(x)^2}{f^0(x)}dx$. To do so, note that 
\begin{align*}
    (h_k^1)^2 &= \left(\int_{B_k} f^1(x) dx\right)^2\\
    & = \left(\int_{B_k} \frac{f^1(x)}{\sqrt{f^0(x)}}\sqrt{f^0(x)} dx\right)^2.
\end{align*}
Applying Cauchy-Schwarz inequality, we find
\begin{align*}
    (h_k^1)^2 &\leq \int_{B_k}\frac{f^1(x)^2}{f^0(x)} dx \int_{B_k}f^0(x)dx
\end{align*}
so
   \begin{align*}
    \frac{(h_k^1)^2}{h_k^0} &\leq \int_{B_k} \left(\frac{f^1(x)}{\sqrt{f^0(x)}}\right)^2 dx,
\end{align*}
which proves the claim.
\end{proof}
\section{Additional simulations} \label{app:simu}


\subsection{Robustness in the case of mixture model $f^0$}
In some scenarios, there might be reasons to suspect that the distribution $f^0$ could undergo changes between the training and test phases. For instance, in the context of flow cytometry data, $f^0$ may be viewed as a multi-modal distribution, with each mode corresponding to a different type of cell. In contrast, $f^1$ is often a uni-modal distribution representing a specific type of cells of interest. The proportions of the different cell types in $f^0$ can vary between individuals, and $f^0$ may be more accurately modeled as a mixture distribution with varying proportions of its components. In the following experiment, we explore the robustness of the Density Ratio Test in the face of this phenomenon.

More precisely, we consider $f^0$ as a mixture of two truncated Gaussians, denoted as $g_a$ and $g_b$, with covariance matrix $\left({1/100 \atop 0} {0 \atop 1/100}\right)$ and means equal to $(3/10, 6/10)$ and $(6/10, 3/10)$, respectively. The distribution $f^1$ is chosen as in Setting 2: a truncated Gaussian with covariance matrix $\left({1/100 \atop 0} {0 \atop 1/100}\right)$ and mean equal to $(6/10, 6/10)$. We set the training sample size $n^0+n^1$ to $1000000$ with $n^0 = 0.7\times(n^{train})$. We investigate the type I error of the test, as well as its power against alternatives $\theta = 0.015$. In the first scenario, we assume that at training time, $f^0 = 0.5\times g_a + 0.5\times g_b$, while in the test sample, $f^0 = \pi\times g_a + (1-\pi) \times g_b$ for $\pi$ varying from $0.5$ to $0.9$. In the second scenario, we assume that at training time, $f^0 = \pi\times g_a + (1-\pi) \times g_b$ for $\pi$ varying from $0.5$ to $0.9$, while in the test sample, $f^0 = 0.5\times g_a + 0.5\times g_b$. Each experiment is reproduced 300 times. The results are presented in Tables \ref{tab:first_shift} and \ref{tab:second_shift}, respectively.

\begin{table}[h!]
\centering
\begin{tabular}{ c || c | c | c | c |c }
   $\pi$ & 0.5 & 0.6 & 0.7 & 0.8 & 0.9 \\ \hline
   Proportion of rejects under $H_0$ & 0.00 & 0.00 & 0.00 & 0.00 & 0.00  \\
   Proportion of rejects under $H_1$ with $\theta = 0.015$ & 1 & 1 &1 & 1& 1
 \end{tabular}

\caption{\label{tab:first_shift} Proportion of experiments where $H_0$ is rejected. The training sample $\mathbb{X}^0$ has distribution $f^0 = 0.5\times g_a + 0.5\times g_b$, while for the test sample $f^0 = \pi\times g_a + (1-\pi) \times g_b$.
}
\end{table}
\begin{table}[h!]
\centering
\begin{tabular}{ c || c | c | c | c |c }
   $\pi$ & 0.5 & 0.6 & 0.7 & 0.8 & 0.9 \\ \hline
   Proportion of rejects under $H_0$ & 0.00 & 0.00 & 0.00 & 0.00 & 1  \\
   Proportion of rejects under $H_1$ with $\theta = 0.015$ & 1 & 1 &1 & 1& 1
 \end{tabular}
\caption{\label{tab:second_shift} Proportion of experiments where $H_0$ is rejected. The training sample $\mathbb{X}^0$ has distribution $f^0 = \pi\times g_a + (1-\pi) \times g_b$, while for the test sample $f^0 = 0.5\times g_a + 0.5\times g_b$.
}
\end{table}

The results presented in Table \ref{tab:first_shift} indicate a certain robustness of the Density Ratio Test against reasonable variations in the proportions of the components of the mixture distribution $f^0$. Indeed, the type I error remains lower than the desired level $\alpha = 0.05$ for all proportions considered in the first experiment, and as long as $\pi \geq 0.2$ in the second experiment. This should come as no surprise: the partition chosen by our algorithm favors bins where the density under $f^0$ is low, while the density under $f^1$ is relatively large. When such bins exist, they contain the majority of the signal and carry the most weight in the test, leading to rejection of the hypothesis when the number of points falling into these bins is sufficiently large. On the contrary, the relative weight of the different regions with a large weight under $f^0$ can vary between the training and test times without significantly increasing the type I error, as these regions have a limited influence on the test result.

\subsection{Detection rate of Density Ratio Test, Bootstrap Density Ratio Test and Maximum Mean Discrepancy}
\label{app:simu_detection_rate} 

In this section, we investigate the detection rate of the Density Ration Test, the Bootstrap Density Ratio Test, and Maximum Mean Discrepancy. Figure \ref{fig:reject} suggests that the detection threshold decreases more rapidly in the high signal setting than in the low signal setting. To confirm this intuition, we investigate the values of $\theta$ such that the Density Ratio Test has power between $0.15$ and $0.85$ for different training sample sizes. The corresponding values of the pairs $(n^{train}, \theta)$ are plotted on a logarithmic scale in Figure \ref{fig:puissance_add}.

We consider the distributions $f^0$ and $f^1$ of Settings A, B and C, corresponding respectively to an easy, intermediate and hard contamination detection problem. We assume again that $n^0 = 0.7 \times (n^{train})$, and we let the number of training samples vary from $1000$ to $1000000$ (for DRT and BDRT) or to  $100000$ (for MMD, due to larger computation times). We then conduct test with a test sample size given by $n = 0.1\times (n^{train})$. We simulate mixtures sampled from $(1-\theta)\times f^0 + \theta f^1$ for $\theta$ varying between 0.0003 and 0.3.
We investigate the values of $\theta$ such that the Density Ratio Test has power between $0.15$ and $0.85$ (for DRT) and between $0.3$ and $0.7$ (for BDRT and MMD) for different training sample sizes. The corresponding values of the pairs $(n^{train}, \theta)$ are plotted on a logarithmic scale in Figure \ref{fig:puissance_add}.

\begin{figure}[h!]
\centering\subcaptionbox{\small Detection rate for DRT.\label{fig:rate_drt_power}}
{\includegraphics[width=0.45\textwidth]{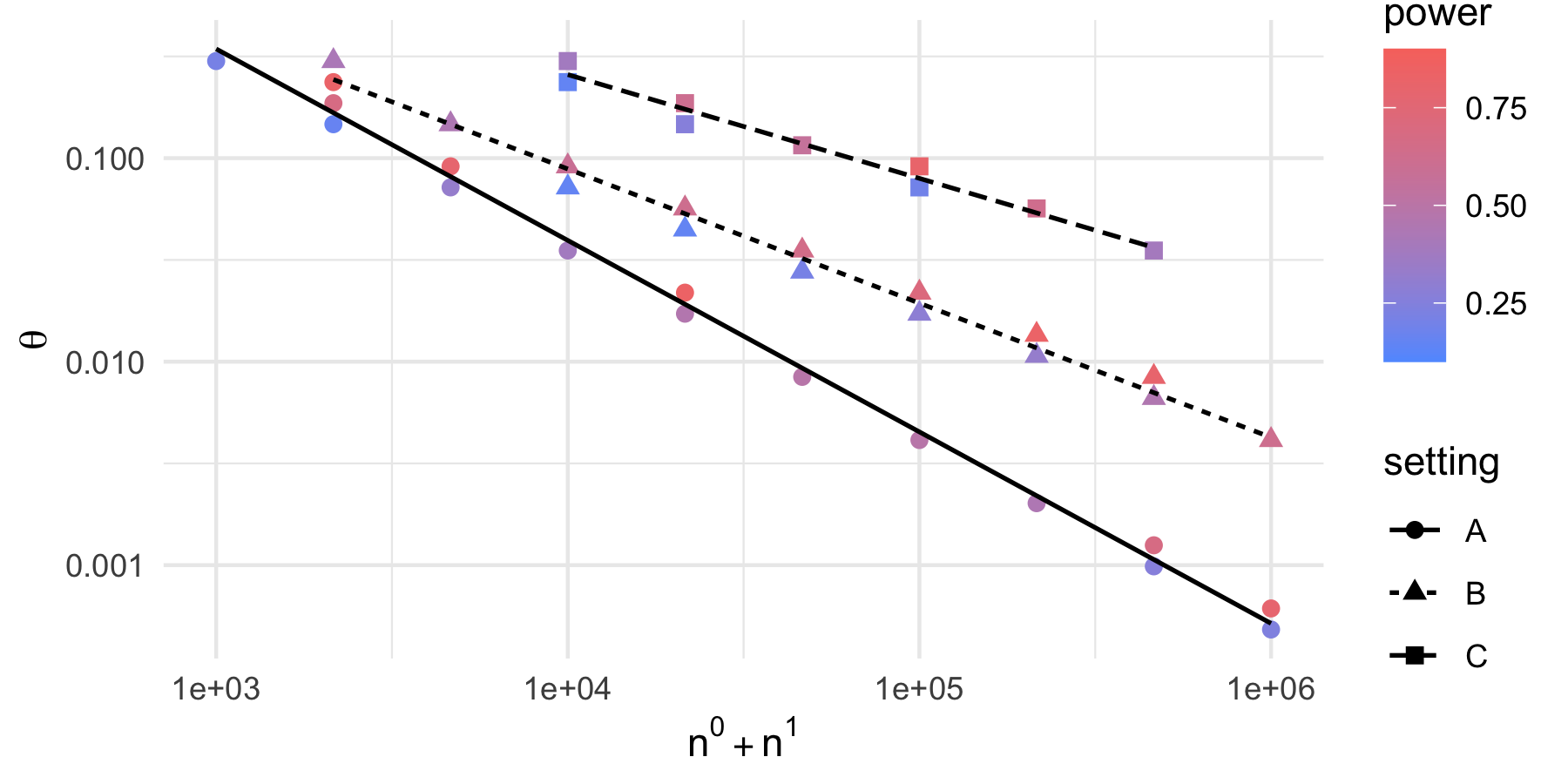}}
\subcaptionbox{\small Detection rate for BDRT.\label{fig:rate_dr_power}}
{\includegraphics[width=0.45\textwidth]{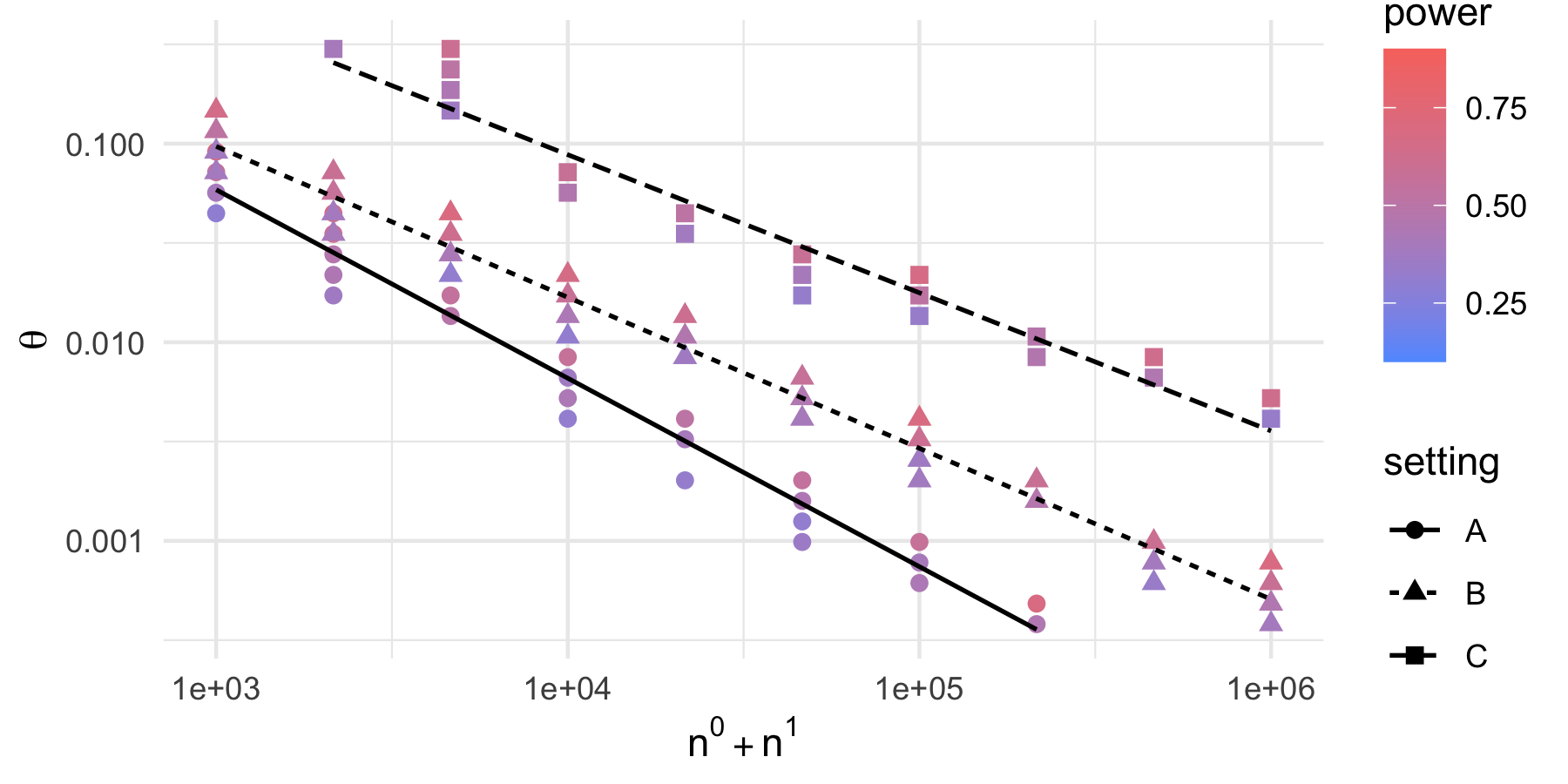}}
\subcaptionbox{\small Detection rate for MMD.\label{fig:rate_mmd_power}}
{\includegraphics[width=0.45\textwidth]{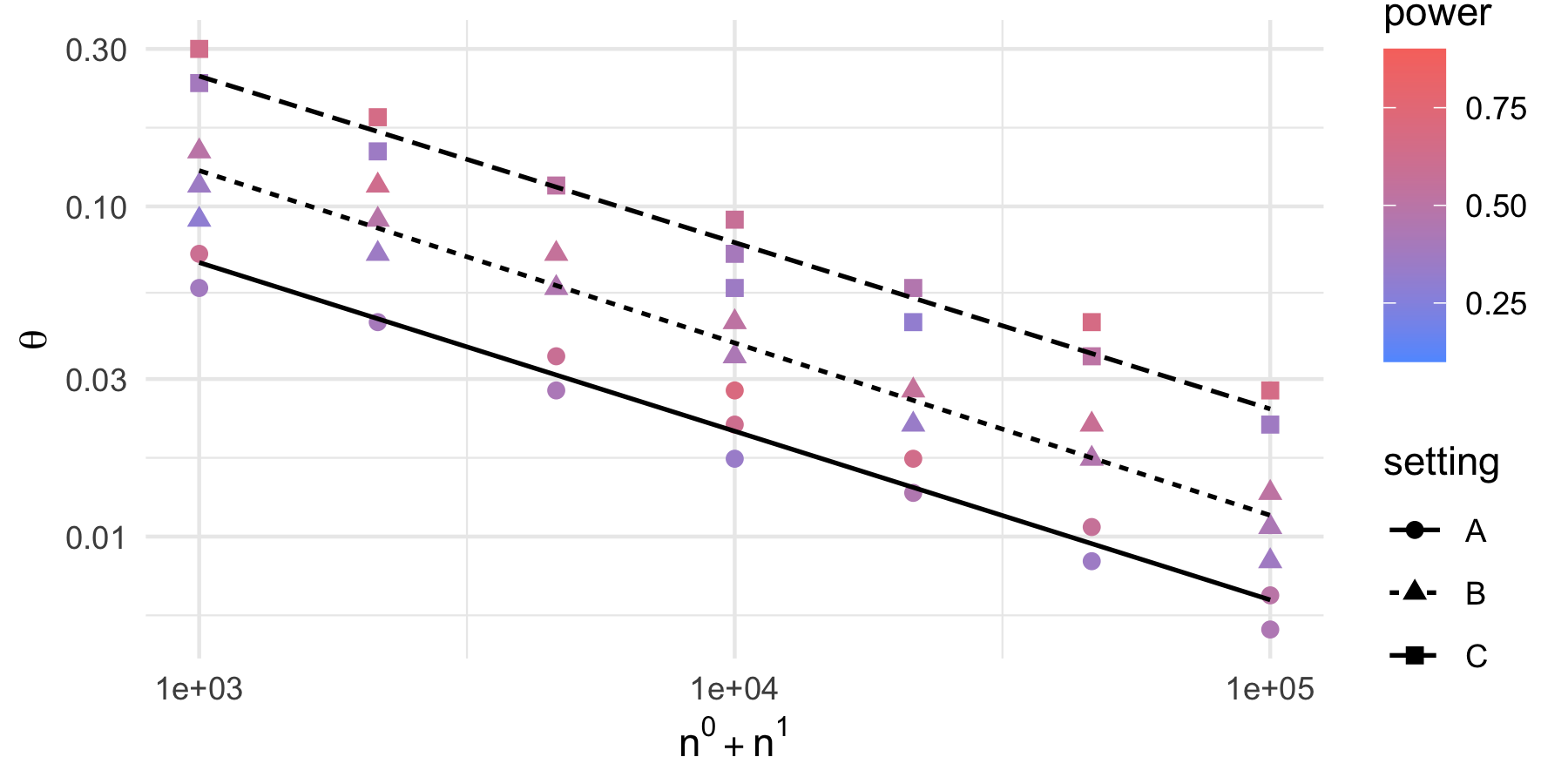}}
\caption{Values of pairs $(n^{train}, \theta)$ such that the Bootstrap Density Ratio Test (left) and MMD Test (right) have empirical power between $0.15$ and $0.85$ for DRT, and between $0.3$ and $0.7$ for BDRT and MMD for an alternative parameterized by $\theta$ when the training sample size is $n^0+n^1$, plotted on a logarithmic scale. The color correspond to the empirical power, the shape of the dots to the setting of the experiment (low, intermediate or high signal).}
\label{fig:puissance_add}
\end{figure}

Figure \ref{fig:puissance_add} reveals that the slope of the linear regression $\log(\theta) = a \log(1/n^{train}) + b$ is larger (in absolute value) in the high signal setting (Setting A): the detection threshold indeed decreases more rapidly in this case. More precisely, we note that the slope $a$ in Setting A is approximately $-0.97$: the test detects alternatives corresponding to approximately constant power at a rate close to $n^{-0.94}$ (up to a multiplicative constant). In this large signal setting, the Density Ratio Test almost achieves the fast detection rate $n^{-1}$. The slope $a$ in the intermediate setting is somewhat lower, approximately $-0.76$. By contrast, the slope is approximately $-0.65$ in the low signal setting. By contrast, the slope are respectively $-0.51$, $-0.52$ , and $-0.5$ for the MMD test (in settings A, B, and C), and $-0.95$, $-0.76$, and $-0.7$ for BDRT.

\subsection{Analysis of the HIPC dataset}
\paragraph{Data pre-processing} We pre-process the observations so as to ensure that each measurement lies in the range $[0,1]$. More precisely, we center each coordinates, and we apply an inverse sinus function. Finally, we rescale the result to the range $[0,1]$.

\paragraph{Experiment} We conduct the following experiment for each cell type, and each of the 62 samples. We consider the current sample of cells as the test sample, with the remaining 61 samples serving as training samples. Initially, we use 12 randomly selected samples (out of 61 samples) to determine the partition. Subsequently, we calibrate the Bootstrap Density Ratio Test to achieve a desired significance level $\alpha$ with a type one error of $5\%$. This involves estimating the distribution of the statistic under the null hypothesis using 25 samples (excluding all cells labeled ``c") and applying the test on the remaining 24 samples (also excluding all cells labeled ``c"). We choose $\alpha$ as the largest value that ensures the test declares at most one false positive. 

Finally, we execute our test for various prevalence levels $\theta$ of cell type $c$. To achieve this, we re-estimate the distribution of the test statistic under the null hypothesis using the same 49 samples used for calibration. Subsequently, we generate a synthetic test sample with a sample size of 10 000 cells, and with varying prevalence levels $\theta$ for cell type $c$. This is done by randomly selecting the corresponding number of cells of type $c$ and cells of other types (without replacement if the number is smaller than the actual count, with replacement otherwise). For each prevalence level $\theta$, we repeat this sampling process 10 times. Averaging across the 62 samples, we obtain an estimate of the test's power for detecting each of the three cell types against alternatives characterized by varying levels of prevalence $\theta$.

We proceed similarly to estimate the power of MMD.


\end{document}